%% file: DNF-revised.tex
\documentclass[11pt]{article}
\usepackage[utf8]{inputenc}	
\usepackage{caption}

\usepackage[dvipsnames]{xcolor}
\usepackage{amsmath,amsthm,amsfonts,amssymb,amscd}
\usepackage{sn-preamble-v2}

\usepackage{tocloft}

\usepackage{prerex}
\usetikzlibrary{fit}

\usepackage[T1]{fontenc}
\usepackage{multirow,booktabs}
\usepackage{yfonts}
\usepackage{fullpage}
\usepackage{lastpage}
\usepackage{enumitem}
\usepackage{fancyhdr}
\usepackage{algorithm}
\usepackage{algpseudocode}
\usepackage{mathrsfs}
\usepackage{wrapfig}
\usepackage{setspace}

\DeclareMathAlphabet{\mathfrak}{U}{jkpmia}{m}{it}
\SetMathAlphabet{\mathfrak}{bold}{U}{jkpmia}{bx}{it}

\usepackage{setspace}

\usepackage{pdfpages}

\usepackage{calc}
\usepackage{multicol}
\usepackage{fontenc}
\usepackage{cancel}
\usepackage[retainorgcmds]{IEEEtrantools}

\usepackage{amsmath}
\usepackage{url}

\usepackage{bbm}
\newcommand{\cube}{\ensuremath{\mathsf{Cube}}}
\usepackage{empheq}
\usepackage{framed}

\usepackage{xcolor}
\usepackage[most]{tcolorbox}

\usepackage{hyperref}
\hypersetup{
    colorlinks=true,
    linkcolor=blue,
    filecolor=magenta,      
    urlcolor=cyan,
    breaklinks=true
    }
\usepackage[nameinlink]{cleveref} \colorlet{shadecolor}{orange!15}

\newcommand{\red}[1]{{\color{red} {#1}}}

\newcommand{\cyan}[1]{{\color{NavyBlue} {#1}}}

\newcommand{\var}{{\sf{var}}}

\newcommand{\fg}{\mathfrak{G}}
\newcommand{\fh}{\mathfrak{H}}

\newcommand{\fJ}{\mathfrak{J}}
\newcommand{\fD}{\mathfrak{D}}

\newcommand{\sU}{{{\mathsf{\overline{R}}}}}
\newcommand{\sR}{{{\mathsf{R}}}}
\newcommand{\sV}{{{\mathsf{V}}}}

\newcommand{\sY}{{\mathsf{Y}}}

\newcommand{\sX}{{{\mathsf{X}}}}

\newcommand{\sS}{{{\mathsf{S}}}}

\def\testeq{\algolink{Test-Equivalence}}
\def\inpool{\algolink{In-Pool}}

\newcommand{\He}{{\sf Heavy}}

\newcommand{\rd}{{\sf reldist}}
\newcommand{\SAMP}{\mathrm{SAMP}}
\newcommand{\MQ}{\mathrm{MQ}}

\newcommand\Algphase[1]{%
\vskip0.05in
\Statex\hspace*{-\algorithmicindent}\textbf{#1}%
\vskip0.05in
}

\newcommand{\pointdist}{\mathrm{dist}}

\newcommand{\Approximator}{\textsc{DNF-Approx}}

\algdef{SE}[REPEATN]{RepeatN}{End}[1]{\algorithmicrepeat\ #1 \textbf{times}}{\algorithmicend}

\newcommand{\algolink}[1]{\protect\hyperlink{#1}{\textbf{\color{violet}{#1}}}}

\newcommand{\SimSAMPA}{\algolink{Sim-SAMP-J}}

\newcommand{\Extract}{\algolink{Extract}}
\newcommand{\TrimCan}{\algolink{FindCandidate}}
\newcommand{\CheckLit}{\algolink{CheckLit}}

\newcommand{\SimMQA}{\algolink{Sim-MQ-J}}

\newcommand{\ConsCheck}{\algolink{ConsCheck}}

\newcommand{\SampleSub}{\algolink{Sim-SAMP}} 
\newcommand{\MQGamma}{\algolink{Sim-MQ-$\Gamma$}} 

\newcommand{\DNFLearner}{\algolink{DNFLearner}}

\newcommand{\FRB}{\algolink{Approximate}}

\newcommand{\CreateOracles}
{\algolink{Find-Factored-DNFs}}

\newcommand{\testdnf}{\algolink{Test-Factored-DNF}}

\newcommand{\Clustering}
{\textsf{Clustering}}

\newcommand{\bAA}{A}

\algdef{SE}[REPEATN]{RepeatN}{End}[1]{\algorithmicrepeat\ #1 \textbf{times}}{\algorithmicend}

\newcommand{\ConjTest}{\algolink{ConjTest}}

\newcommand{\UniformJunta}{\algolink{UniformJunta}}

\newcommand{\maindnf}{\algolink{Test-DNF}}

\author{Xi Chen \\ Columbia University \and William Pires \\ Columbia University \and Toniann Pitassi \\ Columbia University \and Rocco A. Servedio\\
Columbia University}

\title{DNF formulas are efficiently testable with relative error}

\begin{document}

\maketitle

\begin{abstract}
We give a $\poly(s,1/\eps)$-query algorithm for testing whether an unknown and arbitrary function $f: \zo^n \to \zo$ is an $s$-term DNF, in the challenging \emph{relative-error} framework for Boolean function property testing that was recently introduced and studied in a number of works \cite{CDHLNSY25,rel-error-conj-DL,rel-error-junta,rel-error-LTF}.  
This gives the first example of a rich and natural class of functions which may depend on a super-constant number of variables and yet is efficiently testable in the relative-error model with constant query complexity.

A crucial new ingredient enabling our approach is a novel decomposition of any $s$-term DNF formula into ``local clusters'' of terms. Our results demonstrate that this new decomposition can be usefully exploited for algorithms even when the $s$-term DNF is not explicitly given; we believe that this decomposition may have applications in other contexts.

\end{abstract}

\pagenumbering{gobble}

\newpage

\setcounter{tocdepth}{2}

\begin{spacing}{1}
\setlength{\cftbeforesecskip}{3pt}       
\setlength{\cftbeforesubsecskip}{0.75pt}    
\setlength{\cftbeforesubsubsecskip}{0pt} 

 {\small
      \tableofcontents}
\end{spacing}

\newpage

\renewcommand{\listfigurename}{List of Algorithms with simplified description (detailed descriptions are provided in the caption of each algorithm).}
\renewcommand{\figurename}{Algorithm}

\listoffigures
\newpage

\newpage

\pagenumbering{arabic}

\input{sections/DNF-intro2}

\input{sections/DNF-technical-overview}

\input{sections/DNF-preliminaries}

\input{sections/DNF-clustering}
\input{sections/DNF-main-alg}

\input{sections/DNF-Create-Oracles}

\input{sections/DNF-bounded-tester}

\input{sections/DNF-learning}

\input{sections/DNF-consistency-checking}

\clearpage \ \ \  \vspace{0.5cm}
\begin{figure}[h!]
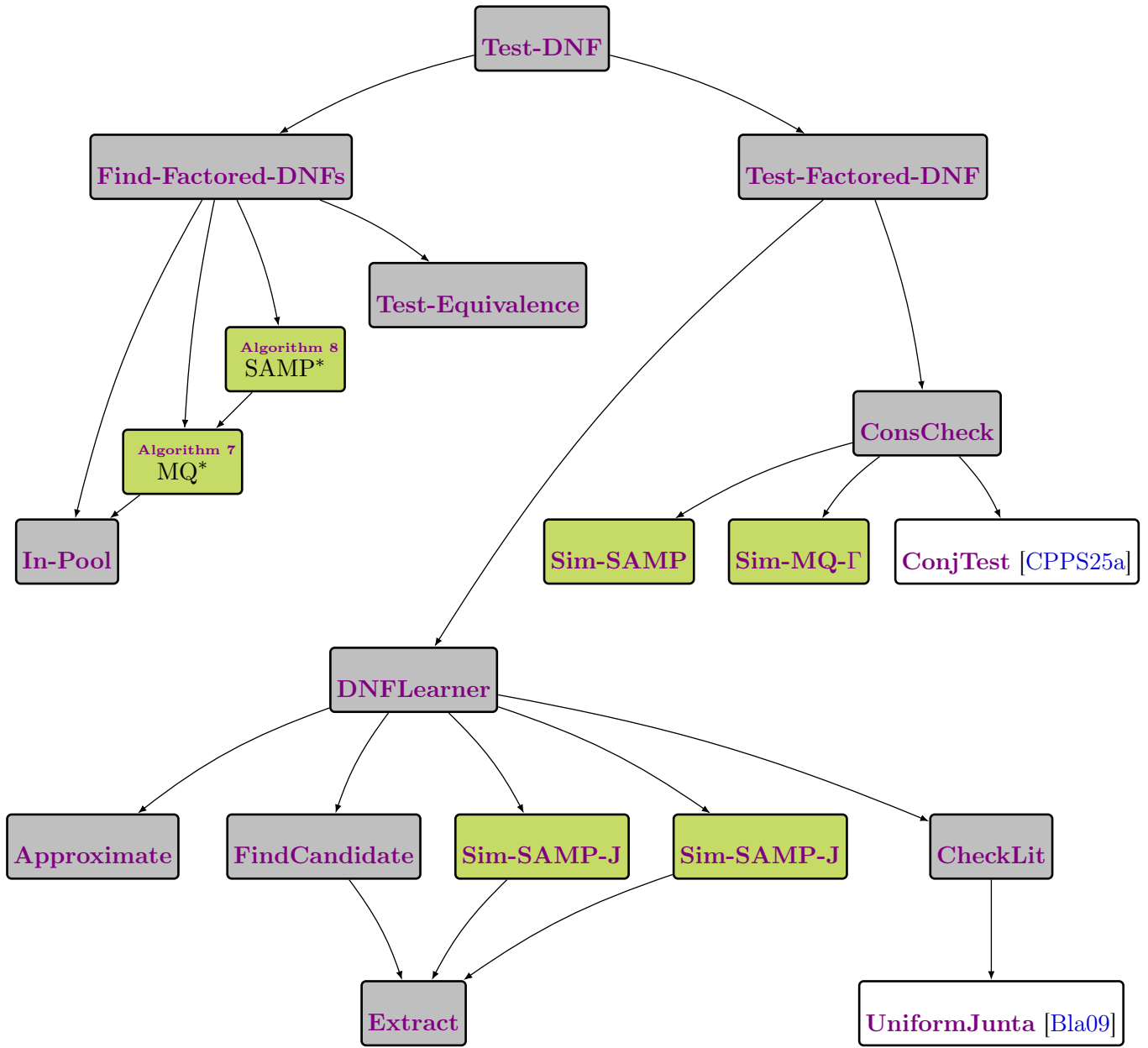

       \caption*{Figure 2: The structure of our algorithm $\testdnf$. An arrow from Algorithm A to Algorithm B means that A is calling B. Algorithms in green are used to simulate various $\SAMP$ and $\MQ$ oracles throughout $\testdnf$. Algorithms in white come from previous work in the literature. \vspace{-1cm}}
\begin{chart}
\reqhalfcoursec 40,60:{}{\maindnf}{}{lightgray}
\reqhalfcoursec 15,50:{}{\CreateOracles}{}{lightgray}
\reqhalfcoursec 65,50:{}{\testdnf}{}{lightgray}

\prereq 40,60,15,50:
\prereq 40,60,65,50:

\reqhalfcoursec 3,20:{}{\inpool}{}{lightgray}

\reqhalfcoursec 35,40:{}{\testeq}{}{lightgray}

\reqhalfcoursec 20,35:{}{\textcolor{black}{$\SAMP^*$}}{\protect\hyperlink{algo:simulator-for-samp}{\textbf{\color{violet}{Algorithm 8}}}}{SpringGreen}

\reqhalfcoursec 12,27:{}{\textcolor{black}{$\MQ^*$}}{\protect\hyperlink{algo:simulator-for-MQ}{\textbf{\color{violet}{Algorithm 7}}}}{SpringGreen}

\prereq 20,35,12,27:
\prereq 12,27,3,20:

\prereq 15,50,3,20:
\prereq 15,50,35,40:

\prereq 15,50,20,35:
\prereq 15,50,12,27:

\reqhalfcoursec 30,10:{}{\DNFLearner}{}{lightgray}
\reqhalfcoursec 70,30:{}{\ConsCheck}{}{lightgray}

\prereq 65,50,30,10:
\prereq 65,50,70,30:

\reqhalfcoursec 60,20:{}{\MQGamma}{}{SpringGreen}

\reqhalfcoursec 77,20:{}{\ConjTest\ {\color{black}\cite{rel-error-conj-DL}}}{}{white}
\reqhalfcoursec 46,20:{}{\SampleSub}{}{SpringGreen}

\prereq 70,30,60,20:
\prereq 70,30,77,20:
\prereq 70,30,46,20:

\reqhalfcoursec 5,-3:{}{\FRB}{}{lightgray}
\reqhalfcoursec 23,-3:{}{\TrimCan}{}{lightgray}
\reqhalfcoursec 40,-3:{}{\SimSAMPA}{}{SpringGreen}

\reqhalfcoursec 57,-3:{}{\SimSAMPA}{}{SpringGreen}
\reqhalfcoursec 75,-3:{}{\CheckLit}{}{lightgray}

\prereq 30,10,5,-3:
\prereq 30,10,23,-3:
\prereq 30,10,40,-3:
\prereq 30,10,57,-3:
\prereq 30,10,75,-3:

\reqhalfcoursec 30,-16:{}{\Extract}{}{lightgray}

\prereq 23,-3,30,-16:
\prereq 40,-3,30,-16:
\prereq 57,-3,30,-16:

\reqhalfcoursec 75,-16:{}{\UniformJunta\  {\color{black}\cite{Blaisstoc09}}}{}{white}

\prereq 75,-3,75,-16:
\end{chart}
\end{figure}\label{fig:diagram}

\clearpage
\begin{flushleft}
\bibliographystyle{alpha}
\bibliography{allrefs}
\end{flushleft}

\appendix

\input{sections/ap-DDS}

\input{sections/DNF-incomparable}

\input{sections/ConjTest}

\end{document}

%% file: sections/DNF-intro2.tex
\section{Introduction}

Property testing, {starting with the works of \cite{BLR93,RS96,GGR98}}, provides a general algorithmic framework for determining whether a given object satisfies a property or is far from satisfying it, by making a small number of queries to the object. 
The key idea and motivation is that for many combinatorial structures such as graphs, functions, or distributions, global properties of the structure can often be approximately tested by examining only a very tiny fraction of the input. One of the most intensively studied property testing domains is  Boolean function classes where efficient testers have been designed  for a broad range of function classes;  moreover, the techniques developed have  led to new insights and connections with other areas, including Fourier analysis of Boolean functions and learning theory.

In the {\it standard} %
testing model for a class of Boolean functions,  %
the tester is given query access to an input function $f$ and parameter $\epsilon$, and should distinguish between $f$ belonging to the class, versus $f$ having Hamming distance at least $\epsilon$ from every function in the class. 
(See \Cref{sec:related-work} below for more background.)
While property testing in the standard model has been incredibly  successful,
a significant drawback of the standard model is that it is poorly suited for testing \emph{sparse} Boolean functions, which have only a very small fraction of all $2^n$ points in $\zo^n$ as satisfying assignments, simply because any such function is trivially close to the constant-0 function.  This limitation --- that sparse functions are trivial to test under the uniform distribution --- arguably sidesteps much of the richness and complexity of a number of interesting Boolean function classes, including the fundamental class of \emph{$s$-term DNF formulas} which is the subject of this work.  For example, the class of $s$-term DNF formulas is notoriously challenging to deal with in the context of efficient PAC learning, and the source of difficulty is precisely the fact that DNF formulas may have terms that are ``wide'' (containing many variables).
However, %
$s$-term DNF formulas with only ``wide'' terms are sparse, and hence are trivial to test in the standard model since they are close to the constant-$0$ function.

In an effort to meaningfully extend the scope of property testing to encompass questions about testing sparse Boolean functions, \cite{CDHLNSY25} introduced a new variant of Boolean function property testing that is known as \emph{relative-error} testing.  
The relative-error testing model defines the distance from the unknown $f: \zo^n \to \zo$ that is being tested to another function $g: \zo^n \to \zo$ to be 
\[
\rd(f,g) := {\frac {|f^{-1}(1) \hspace{0.08cm} \triangle \hspace{0.08cm} g^{-1}(1)|}{|f^{-1}(1)|}}.
\]

In contrast to the standard model,
in the relative error model, the distance from $f$ to $g$ is measured as a function of $f$'s  support size, $f^{-1}(1)$, and not just globally across all $2^n$ inputs. 
This distinction is crucial in applications like feature selection, or fault detection where the underlying functions often describe rare but important events, or in large but sparse data domains.
Relative error property testing also connects property testing to several other areas where distances are measured in a relative sense e.g.,  
PAC learning under general distributions, noise stability and approximate counting.

Note that if only black-box $\MQ(f)$ queries were allowed, as in the standard model, then an enormous number of queries could be required to obtain any information about $f$ at all for very sparse $f$. Hence, the model also allows the testing algorithm to obtain independent uniform elements of $f^{-1}(1)$ from a sampling oracle $\SAMP(f)$. 
(See \Cref{sec:relative-error-testing} for a detailed description of the relative-error testing model.)

It is intuitively obvious and simple to show (see Fact~3.1 of \cite{CDHLNSY25}) that relative-error testing is at least as hard as standard-model testing.  
{Moreover, an easy example of an (artificial) class of functions {(see \cite{CDHLNSY25})} shows that relative-error testing can be hard, requiring $2^{\Omega(n)}$ calls to $\MQ(f)$ or $\SAMP(f)$, even for classes that are easily testable with $O(1)$ queries in the standard model.}  However, it is \emph{a priori} unclear how easy (or hard) it is to do relative-error testing for natural function classes of interest.

\subsection{Our Main Result: Relative-Error Testing for DNFs} 

In this paper we focus on the class of $s$-term DNF formulas; the class of all functions $f: \zo^n \to \zo$ that can be expressed as $f=T_1 \vee \cdots \vee T_{s'}$ for some $s' \leq s$, where each $T_i$ is a conjunction of Boolean literals.
DNF formulas are a rich and intensively studied class of functions with a long history in logic and theoretical computer science.  As a universal model that can express any Boolean function, DNF formulas are one of the most important classes of Boolean functions, central in the study of Boolean function analysis, learning theory, computational complexity, and testing.

DNFs are also particularly interesting from the vantage point of relative-error property testing.
On the one hand, $s$-term DNF formulas are a highly expressive  class which unlike juntas (and subclasses of juntas), can depend on a superconstant number of variables.
On the other hand, they appear to be quite challenging in the relative-error testing model compared to the standard model.
In the standard-error model, $s$-term DNFs are testable with only
$\tilde{O}(s/\eps)$ queries (\cite{Bshouty20,DLM+:07,CGM11}).
In contrast, 
 a straightforward consequence of previous work in learning theory gives a relative-error tester for $s$-term DNFs that makes 
$n^{O(\log(s/\eps))}$ queries. (See Appendix \ref{ap:DDS} for details.)

A key challenge in improving relative-error testing of DNFs is the lack of smoothness: small changes to a DNF such as deleting a single variable from a term can cause large changes in the number of satisfying assignments\footnote{For example, deleting a single variable from a term can potentially double the number of satisfying assignments of the DNF.}, complicating both testing and approximation.
Understanding the complexity of DNF testing in the relative error model thus provides a window into how combinatorial structure interacts with multiplicative accuracy guarantees, and whether efficient testers can exist in settings where absolute and relative notions of distance diverge.  

Our main result 
is a relative-error testing algorithm for DNF formulas with query complexity \emph{independent} of the ambient dimension $n$: 
\begin{theorem}[Relative-error testing of $s$-term DNF] \label{thm:main}
For $0<\eps\leq 1/2,$
there is an $\eps$-relative-error testing algorithm for the class of $s$-term DNF formulas over $\{0,1\}^n$ which makes $\poly(s,1/\eps)$ queries.

\end{theorem}

This is a dramatic improvement to the $n^{O(\log(s/\epsilon))}$ query relative-error DNF tester mentioned above.
We remark that our new tester has query-complexity that is equivalent to the state-of-the-art DNF tester in the standard model (\cite{Bshouty20,DLM+:07,CGM11}), up to polynomial factors.

\subsection{Background and Related Work}\label{sec:related-work}

\noindent{\bf Standard-Error Testing.} Over the past thirty or so years, a great deal of research effort has been dedicated to understanding which properties of Boolean functions (equivalently, which function classes) are efficiently testable in the standard model. Starting with the fundamental works of \cite{BLR93,GGR98}, a wide range of different function classes have been considered, and in a surprising number of cases it turns out that constant-query testability (i.e.,~query complexity independent of the ambient dimension $n$) is indeed possible.    
A partial list of function classes that are now known to be constant-query testable  would include literals \cite{PRS02}; conjunctions \cite{PRS02,GoldreichRon20}; parity functions \cite{BLR93}; decision lists \cite{DLM+:07,Bshouty20}; linear threshold functions \cite{MORS10}; $s$-juntas \cite{FKRSS03,Blais08,Blaisstoc09}; subclasses of $s$-juntas (this covers many natural classes including size-$s$ decision trees, size-$s$ branching programs, size-$s$ Boolean formulas, and functions with Fourier degree at most $d$ \cite{DLM+:07,CGM11,Bshouty20}); and classes of functions each of which can be well-approximated under the uniform distribution by some junta (this covers many other classes including $s$-term monotone DNF \cite{PRS02}, $s$-term DNF, $s$-sparse polynomials over $\F_2$, and size-$s$ Boolean circuits with unbounded fan-in \cite{DLM+:07,CGM11,Bshouty20}).

\bigskip

\noindent{\bf Relative-Error Testing.} 
Recall that in the relative-error testing model, 
the distance from $f$ to $g$ is measured at the ``relative'' scale of $|f^{-1}(1)|$ rather than at the ``absolute'' scale of~$2^n =$ $ |\zo^n|$.
The model is thus similar in spirit to well-studied models in graph property testing (see \cite{GoldreichRon02,ParnasRon02} and Chapters~9 and~10 of \cite{PropertyTestingICS}) which are suitable for testing sparse $N$-vertex graphs that have $o(N^2)$ edges.

Several recent works have given efficient relative error testers for a variety of well-studied classes.  \cite{rel-error-junta} showed that $s$-juntas and subclasses of $s$-juntas that are closed under renaming variables (and satisfy an additional mild technical condition, see \cite{rel-error-junta} for details) can be efficiently tested in the relative-error model, using only $\poly(s,1/\eps)$ oracle calls. Building on very different techniques, another recent work \cite{rel-error-conj-DL} showed that the classes of Boolean conjunctions and $1$-decision lists, of any length, are also efficiently relative-error testable, using only $\tilde{O}(1/\eps)$ queries. The \cite{rel-error-conj-DL} result provides an interesting contrast to \cite{rel-error-junta} because neither conjunctions nor decision lists are subclasses of juntas; however, both conjunctions and decision lists are rather simple and inexpressive types of Boolean functions.
On the negative side, a slightly richer class which generalizes both conjunctions and 1-decision lists, namely the class of \emph{linear threshold functions}, was recently shown \emph{not} to be constant-query testable: \cite{rel-error-LTF} proved that any relative-error testing algorithm for LTFs must have $\tilde{\Omega}(\log n)$ query complexity.
\bigskip

{\noindent {\bf Relation to distribution-free testing.}}

We remark that relative-error testing is not the first Boolean function property testing model to measure the distribution between two functions $f$ and $g$ using a distance notion other than the fraction of points in $\zo^n$ on which they disagree.  
In the distribution-free testing model, the distance between two functions $f$ and $g$ is measured according to $\Pr_{\bx \sim {\cal D}}[f(\bx) \ne g(\bx)]$, where ${\cal D}$ is an unknown and arbitrary distribution over the domain $\zo^n$.  
In addition to making black-box oracle queries, distribution-free testing algorithms can also draw i.i.d.~samples from the unknown distribution~${\cal D}$.  
However, it turns out that many classes of functions which are constant-query testable in the standard model, such as decision lists, linear threshold functions, and conjunctions (which are $s$-term DNF for $s=1$), require a superconstant --- in fact, polynomial in $n$ --- number of queries in the distribution-free setting, see e.g. \cite{GlasnerServedio:09toc,CX16,CP22,CFP24}.  

We further remark that the relative-error testing model is incomparable to distribution-free testing; one reason for this is because the distance between functions is measured differently in the two settings. In the distribution-free model, the distance $\Pr_{\bx \sim {\cal D}}[f(\bx) \neq g(\bx)]$ between two functions $f,g$ is measured vis-a-vis the distribution ${\cal D}$ that the testing algorithm can sample from, and these samples will typically include both positive and negative examples of the function $f$ that is being tested. In contrast, in the relative-error setting the distribution ${\cal D}$ that the tester can sample from is the uniform distribution over $f^{-1}(1)$, i.e.~it is supported entirely on positive examples of $f$. Moreover, the relative distance $\rd(f,g)$ is \emph{not} equal to $\smash{\Pr_{\bx \sim f^{-1}(1)}[f(\bx)\neq g(\bx)]}$: a function $g$ can have large relative distance from $f$ because $g$ disagrees with $f$ on many points outside of $f^{-1}(1)$, which are points that have zero weight under the uniform distribution over $f^{-1}(1)$.
  
To concretely specify two function classes demonstrating that the two models are incomparable, we note that the class of conjunctions is known to require $\poly(n)$ queries in the distribution-free setting \cite{GlasnerServedio:09toc,CX16}, whereas it is $O(1/\eps)$-query testable in the relative-error model \cite{rel-error-conj-DL}. In the other direction, in  \Cref{sec:incomparable} we give a simple example of a function class that is trivially distribution-free testable but requires $2^{\Omega(n)}$ queries for any relative-error tester.

%% file: sections/DNF-technical-overview.tex
\def\bx{{\boldsymbol{x}}}
\def\by{{\boldsymbol{y}}}
\def\testeq{\algolink{Test-Equivalence}}

\section{Technical Overview} \label{sec:techniques}

In this section we  explain the main ideas behind our algorithm, \maindnf,  for relative-error testing of DNFs. 
In \Cref{sec:highlevel}, we begin with a conceptual overview of the high-level ideas underlying our approach and describe its two algorithmic components, \CreateOracles\ and \testdnf. 
We then proceed from that to discuss some of the most significant hurdles that need to be overcome to make the approach work in \Cref{sec:hurdle1,sec:hurdle2}.
See \hyperref[fig:diagram]{Figure 2} for an illustration of the detailed structure of our testing algorithm.

{
\subsection{High-level ideas}\label{sec:highlevel}

\noindent {\bf Factored-DNFs.}  To explain our approach, a crucial preliminary notion is that of a \emph{factored-DNF}.
 Roughly speaking, a factored-DNF is a function of the form $H \wedge f_1,$ where $H$ is a term (conjunction) of \emph{arbitrary} length and $f_1$ is a $(\le s)$-term DNF \emph{that depends on only $\poly(s)$ many variables}; see \Cref{def:factored-DNF} for a more precise definition.
Note that we can think of the conjunction $H$ (which again may be of arbitrary length) as a common part that has been ``factored out'' of each of the terms of an initial DNF.  We sometimes refer to $H$ as the ``\emph{head}'' and $f_1$ as the ``\emph{tail}'' of the factored-DNF.

With the notion of a factored-DNF in hand, we are ready to describe our new decomposition.
}
\medskip

\noindent {\bf A new decomposition of DNFs.}
Our approach crucially employs a new technical and conceptual ingredient, which is a \emph{novel structural decomposition result for arbitrary $s$-term DNFs}.  
This structural decomposition result, informally speaking, says the following:

\noindent
\begin{center}

    \fbox{\begin{minipage}{40em}\vspace{0.2cm}
          Given any $s$-term DNF $f$, there is a sub-DNF $f'$ of $f$, created by removing some terms\\ of $f$ (so $\smash{f'^{-1}(1) \subseteq f^{-1}(1)}$), which has the following properties:
\begin{flushleft}\begin{enumerate}
    \item[(i)] Almost all satisfying assignments of $f$ are also satisfying assignments of $f'$;
    \item[(ii)] $f'$ can be partitioned into an OR of no more than $s$ many sub-DNFs $h_1,\dots,h_{\le s}$\\ (so each term in each $h_i$ is one of the original $s$ terms of $f$), each of which is a \emph{factored-DNF} as described earlier;
    \item[(iii)]
    A uniform random satisfying assignment of $f'$ (satisfying some sub-DNF $h_i$) is extremely likely to disagree with every term in every other sub-DNF $\smash{h_j, j \neq i}$,\\ in a large number of coordinates.\vspace{0.2cm}
\end{enumerate}\end{flushleft}
\end{minipage}
}
\begin{figure}[H]
    \caption*{Figure 1: Structural decomposition}
\end{figure} \label{fig:decomposition}
\end{center}

The above is a rough statement; see \Cref{sec:clustering}, in particular \Cref{thm:narrow-term-any-term-outside-cluster-far-from-random-a-OLD}, for more details.
Given the existence of such a decomposition, we now give an overview of the first stage of our algorithm.

Given an input function $f:\{0,1\}^n\rightarrow \{0,1\}$ (with access to oracles $\MQ(f)$ and $\SAMP(f)$),
the first stage of $\maindnf$ runs a procedure called \CreateOracles\ to generate a pair of oracles $\SAMP(h_i)$ and $\MQ(h_i)$ for each factored-DNF $h_i$~in the above decomposition, when $f$ is an $s$-term DNF. On the other hand,
when $f$ is~far from $s$-term DNFs~in relative error distance, \CreateOracles\ either rejects or returns pairs of oracles for a sequence of functions $h_1,\ldots,h_{\le s}$ that also satisfy $\smash{\bigcup h_i^{-1}(1)\subseteq f^{-1}(1)}$ and that $\bigcup h_i^{-1}(1)$ contains almost all satisfying assignments of $f$ (though, unlike the yes case, $h_i$'s here are not necessarily factored-DNFs). 
The query complexity of \CreateOracles\ is $\poly(s/\eps)$.

We emphasize that the performance guarantees of \CreateOracles\ described above have been substantially simplified to give intuition about how $\maindnf$ works.
In particular, being able to return \emph{perfect} oracles for the factored-DNF $h_i$'s in the decomposition may sound too good to be true.
Indeed, as discussed in \Cref{sec:hurdle1} as one of the hurdles that need to be~overcome when these oracles are used during the second stage, \CreateOracles\ returns ``\emph{approximate simulators}'' for
$\MQ(h_i)$ and $\SAMP(h_i)$: They are randomized algorithms built by $\CreateOracles$ that make $\poly(s/\eps)$ calls to $\MQ(f)$ and $\SAMP(f)$ to simulate a call to $\MQ(h_i)$ or $\SAMP(h_i)$; for clarity, we will refer to them later as $\MQ^*(h_i)$ and $\SAMP^*(h_i)$ to highlight that they are simulators and may potentially make mistakes. %
For the overview in this subsection though, the reader may assume that they are perfect simulators so we continue to denote them by $\MQ(h_i)$ and $\SAMP(h_i)$.\medskip 

\noindent {\bf Testing factored-DNFs.} Given that $h_1,\ldots, h_{\le s}$
  form a partition of a sub-DNF $f'$ of the $s$-term DNF $f$,
  there exist positive integer $r_i$'s such that
  $\sum_i r_i\le s$ and 
  each $h_i$ is a factored-DNF with $r_i$ terms in its tail.
On the other hand, when $f$ is far from $s$-term DNFs in relative distance, using that the $h_i$'s returned by \CreateOracles\ satisfy
  $\bigcup h_i^{-1}(1)\subseteq f^{-1}(1)$ and that $\bigcup h_i^{-1}(1)$ contains almost all satisfying assignments of $f$,
  it is not difficult to show that there cannot exist  integer $r_i'$'s 
  such that  $\sum_i r_i'\le s$ and 
  each $h_i$ is very close in relative distance to a factored-DNF with $r_i'$ terms in its tail.
The second stage of the algorithm is to distinguish between these two cases, for which it is natural to run a tester for factored-DNFs
  on each $h_i$.

This is achieved by our second algorithmic component, $\testdnf$.
Given a function $h$ and a parameter $r\le s$, it uses $\poly(s/\eps)$ many queries to $\MQ(h)$ and $\SAMP(h)$ to distinguish the two cases when $h$ is a factored-DNF with at most $r$ terms in the tail and when $h$ is far in relative distance from any factored-DNF with at most $r$ terms in the tail.
We discuss challenges behind the design of 
  such an efficient tester later in  \Cref{sec:hurdle2}.
We remark that, given that \CreateOracles\ only generates ``approximate simulators,'' \testdnf\ needs to be robust enough to
accommodate their inaccuracies. %
We will return to this point later in \Cref{sec:hurdle2}.
\medskip

\noindent{\textbf{Overview of the main algorithm.}}
In light of the above discussion, we can now give a  high-level overview of our relative-error testing algorithm \maindnf\ for $s$-term DNFs: %
\begin{flushleft}\begin{enumerate}
\item[]\textbf{Stage 1:} Given the input function $f$, run \CreateOracles; it either rejects or   generates pairs of oracles $\SAMP(h_i)$ and $\MQ(h_i)$ for a sequence of functions $h_1,\ldots, h_{\le s}$.

\item[]\textbf{Stage 2:}
For each $h_i$, run \testdnf\ on $h_i$  with parameter $r$ set to $1,2,\ldots,s$ (using oracles generated by \CreateOracles). Let $r_i'$ denote the smallest integer such that $\testdnf$ accepts $h_i$ with $r_i'$.
Accept if $\sum_i r_i'\le s$ and reject otherwise.
\end{enumerate}\end{flushleft}
The overall query complexity of \maindnf\ to $\MQ(f)$ and $\SAMP(f)$ is clearly $\poly(s/\eps)$.

When $f$ is an $s$-term DNF, 
  $\CreateOracles$ returns oracles for $h_i$'s given in the decomposition of $f$.
Let $r_i$ be the number of terms in the tail of each $h_i$, satisfying that $\sum_i r_i\le s$.
We have that $r_i'\le r_i$ in the second stage for each $i$ 
  (note that $h_i$ could be close to a factored-DNF with less than 
  $r_i$ terms in the tail) and thus,
  $\maindnf$ accepts.

When $f$ is far from every $s$-term DNF in relative distance,  
  $\CreateOracles$ either rejects or returns oracles for functions $h_i$ for which, as discussed earlier, there cannot exist integer $r_i'$'s
  such that  $\sum_i r_i'\le s$ and 
  each $h_i$ is close in relative distance to a factored-DNF with $r_i'$ terms in its tail.
From this one can conclude that $f$ is rejected in Stage 2 of \maindnf.

\subsection{Finding factored DNFs}\label{sec:hurdle1}

We now give more details on how \CreateOracles\ generates oracle simulators for factored-DNFs 
$h_1,\ldots,h_{\le s}$ that form a decomposition of $f$ described earlier,  when $f$ is an $s$-term DNF.

This starts with \Cref{sec:clustering}, where we describe a
  simple greedy clustering procedure that,
  given \emph{explicitly} the terms $T_1,\ldots,T_s$ of $f$, partitions 
  them into clusters $D_1,\ldots,D_{\le s}$ such that by setting $g_i$ to be the sub-DNF defined using terms in $D_i$, we would 
  obtain a decomposition of $f$ into
  $g_1,\ldots,g_{\le s}$ that satisfy all three
  properties given earlier in \hyperref[fig:decomposition]{Figure 1}.
(Here we used $g_i$'s instead of $h_i$'s because these are not exactly the functions returned by \CreateOracles.)
This is done by making sure that, roughly speaking, (1) every~two terms in the same cluster are ``close'' and (2) every~two terms in two distinct clusters are ``far,'' where the reader may consider the distance between two terms measured by their symmetric difference (though it needs to be done in a more careful manner in \Cref{sec:clustering}).
Naturally, (1) leads to property (ii) and (2) leads to property (iii) in the decomposition of \hyperref[fig:decomposition]{Figure 1}.
The latter is because every two terms in the same cluster $D_\ell$ are close to each other and thus, they can be factored out to form a factored-DNF
(see \Cref{thm:narrow-term-any-term-outside-cluster-far-from-random-a-OLD}).  
We emphasize that our  algorithm does not and cannot carry out this procedure, since it does not ``know'' the terms. This clustering $D_1,\ldots,D_{\le s}$ should be considered as the ideal goal for $\CreateOracles$,  which will play  a crucial conceptual role in its design and analysis.
For convenience, we will refer to $D_1,\ldots,D_{\le s}$ as the \emph{canonical} clusters in the rest of this section.

At a high level, \CreateOracles\ starts by drawing $\poly(s/\eps)$ samples from $f^{-1}(1)$.
For clarity of discussion within this section, let's assume that every point $\bx$ that is~sampled satisfies a unique term in $f$; we denote it by $T_{\bx}$ and denote the canonical cluster that $T_\bx$ lies in by $D_{\bx}$.~Using property (iii) on $g_1,\ldots,g_{\le s}$ (or equivalently on the canonical clustering $D_1,\ldots,D_{\le s}$), the following can be shown to hold, with high probability, for any two points $\bx,\by$ that are sampled: 
\begin{flushleft}\begin{enumerate}[label=(\alph*)]
\item If $\bx$ and $\by$ come from the same term $T_{\bx}=T_{\by}$, then every point $z$ in the cube defined by $\bx$ and $\by$\footnote{The cube defined by $\bx$ and $\by$ is the set of points $z$ with $z_i \in \{\bx_i, \by_i\}$ for every $i \in [n]$. } satisfies $T$ and thus, is a satisfying assignment of $f$ (this is trivially always true);
\item If $\bx$ and $\by$ come from two different canonical clusters (i.e., $D_{\bx}\ne D_\by$), then most points in the cube defined by $\bx$ and $\by$ are not satisfying assignments of $f$.
To see intuitively that this should be the case, by property (iii) of the decomposition, we have that most likely, $\bx$ is far from satisfying every term not in $D_\bx$ and $\by$ is far from satisfying every term not in $D_\by$, from which one can argue that most points in the cube they define do not satisfy any term of $f$.
\end{enumerate}\end{flushleft}
Distinguishing these two cases can be done  by 
  querying a small number of random  points  from~the cube.
We run  this subroutine, $\testeq$, on every pair of samples $\bx,\by$  and
 add~an~edge $(\bx,\by)$ between them when
 $\testeq$ is convinced that the cube defined by $\bx,\by$
 is full of satisfying assignments of $f$.
With high probability, we have for every pair $\bx,\by$: (a) If $T_\bx=T_\by$, then $(\bx,\by)$ is an edge and (b) If $D_\bx\ne D_\by$, then $(\bx,\by)$ is not an edge. 
This graph partitions our samples into ``pools'' $\calP_1,\ldots,\calP_{\le s}$, where each pool $\calP_i$ is a connected component of the graph.

We know the number of pools must be at most $s$ because of (a) and the fact that $f$ has no more than $s$ terms.
Letting $C_i$ be the set of terms of $f$ that are each satisfied by at least one point in $\calP_i$, then we have from (b) that $C_1 , \ldots,C_{\le s}$ is a refinement of the canonical clustering $D_1,\ldots,D_{\le s}$ and thus, the sub-DNF that each $C_i$ defines must also be a 
  factored-DNF, which we denote by $h_i$. 
  {%
We hope that $h_1,\ldots,h_{\le s}$ can give us the desired 
  decomposition of $f$.
However, unlike $g_i$'s (as sub-DNFs of the canonical $D_i$'s), it is no longer clear if property (iii) of \hyperref[fig:decomposition]{Figure 1}  continues to hold:
  $C_i$ and $C_j$ may come
  from the same $D_\ell$, in which case
  their terms can be close to each other and thus, it could hold that a sample from $f^{-1}(1)$ that satisfies a term in $C_i $ is likely to satisfy a term in $C_j$ as well.

To overcome this obstacle, we perform a crucial process of \emph{merging pools} using~a batch of fresh samples from $f^{-1}(1)$ and running $\testeq$ between them and points in every pool.
We observe that if a sample drawn from $f^{-1}(1)$ that satisfies a term in $C_i$ is likely to satisfy a term in $C_j$ as well, then we are likely to get such a point $\bz$ in the newly drawn batch of samples, %
{and the}
\testeq\ subroutine would
  be convinced that the cube defined between $\bz$ and some point in $C_i$ and the cube defined between $\bz$ and some point in $C_j$ are both full of satisfying assignments.
In this case, we \emph{merge} $\calP_i$ and $\calP_j$ into a single pool $\calP_i\cup \calP_j$
  and by doing so, $C_i$ and $C_j$ are merged as well.
On the other hand, if $C_i$ and $C_j$ come from different canonical clusters, then since (iii) is satisfied by the canonical clustering, they are unlikely to be merged by $\testeq$. 

At the end we obtain pools $\calP_1, \ldots,\calP_{\le s}$ such that
  the $h_i$'s they define form the desired
  decomposition of $f$ into factored-DNFs.
The rest of \CreateOracles\ is to 
  build simulators for oracles of each $h_i$ using pool $\calP_i$, denoted from now on by $\MQ^*(h_i)$ and $\SAMP^*(h_i)$. %
We focus on the case when $f$ is an $s$-term DNF~and each $h_i$ is a factored-DNF defined from
  a pool $\calP_i$. 
Ideally, we hope $\MQ^*(h_i)$ and $\SAMP^*(h_i)$ behave {similarly to}
$\MQ(h_i)$ and $\SAMP(h_i)$ so that the probability of $\testdnf$ accepting $h_i$ on $\MQ^*(h_i)$ and $\SAMP^*(h_i)$ is at least the 
  probability of it accepting $h_i$ on $\MQ(h_i)$ and $\SAMP(h_i)$ (which is at least $2/3$) minus $0.01$. 
To this end, standard analysis based on coupling suggests 
  the following desired performance guarantees: %
\begin{enumerate}
\item Each call to $\SAMP^*(h_i)$ returns with 
  probability at least $1-\delta$ a draw from $h_i^{-1}(1)$;\vspace{-0.1cm}
\item Each call to $\MQ^*(h_i)$ about a point $x\in \{0,1\}^n$ returns $h_i(x)$ with probability at least $1-\delta.$ 
\end{enumerate}
We would be done if $1/\delta$ is larger than the  query complexity of $\testdnf$.

The good news is that   $\SAMP^*(h_i)$  achieves exactly the above. It starts by drawing $\bx\sim f^{-1}(1)$ and runs $\testeq$ between $\bx$ and every point in $\calP_i$; it returns $\bx$ if $\testeq$ is convinced that the cube defined by $\bx$ and some point in $\calP_i$ is full of satisfying assignments of $f$; otherwise, it draws a new sample from $f^{-1}(1)$ and repeats.
Whenever $\bx\in h_i^{-1}(1)$, $\testeq$ would always tell us this is the case; on the other hand, by property (iii) of the decomposition on $h_i$'s, it is very unlikely to draw an $\bx\notin h^{-1}(1)$ that can fool $\testeq$.\footnote{{To bound the number of samples that $\SAMP^*(h_i)$ needs,  we need to perform an additional preprocessing step of pool trimming, by discarding pools that are ``light'' and only keeping those that are ``heavy,'' {meaning that} the probability of  $\bz \sim f^{-1}(1)$ satisfying $h_i(\bz)=1$ is at least $1/\poly(s/\eps)$. We skip the details but note that trimming light pools is ok because property (i) in the decomposition only requires $f'$ to cover most of $f^{-1}(1)$.} 
 }

The bad news is that implementing $\MQ^*(h_i)$ to achieve the above (worst-case type) performance guarantees seems infeasible. Consider an $x\in f^{-1}(1)$ {for which} $h_i(x)=0$ but the cube between $x$~and some $y\in \calP_i$ contains mostly points in $f^{-1}(1)$. It is unlikely {that an algorithm can}  distinguish this from the case when $x\in h_i^{-1}(1)$.
On the other hand, as property (iii) of the decomposition {of \hyperref[fig:decomposition]{Figure 1}} suggests, 
answering $h_i(\bx)$ correctly for a ``\emph{random-ish}'' point $\bx$ is easy. 
As an example, consider $\bx\sim f^{-1}(1)$. By property (iii), most likely either   $h_i(\bx)=1$ (and $\bx$ satisfies at least one term in $C_i$), or 
  $h_i(\bx)=0$ and $\bx$ is far from satisfying any term in $C_i$. These two cases can be
  easily distinguished by $\testeq$ with high probability.
The question left is whether this is sufficient for our purpose? 
In particular, in the yes case,
  would $\testdnf$ still accept $h_i$ with high probability when running on $\MQ^*(h_i)$ instead of $\MQ(h_i)$?

{To this end, we now turn to how \testdnf\ works, starting with the simpler case when it has access to  \emph{perfect} oracles $\MQ(h_i)$ and $\SAMP(h_i)$.}

\subsection{Testing factored DNFs}\label{sec:hurdle2}
{How can we design an efficient
  relative-error tester  for factored-DNFs?  
A clue to this is provided by the recent work of \cite{rel-error-conj-DL}, which gave a $\tilde{O}(1/\eps)$-query relative-error tester for \emph{decision lists}.  
The connection between decision lists and factored-DNFs is as follows:  as is clear from a moment's thought, any nontrivial decision list %
can be viewed as being of the form $H \wedge L$, where $H$ is a conjunction (of arbitrary length) and $L$ is another decision list that is ``\emph{dense}'':  %
$L$ outputs $1$ on a constant fraction of all input assignments. 
At a high level, the \cite{rel-error-conj-DL} algorithm essentially works by running a standard-model algorithm for testing decision lists on the list $L$ 
(note that this is not at all straightforward to do because of the presence of the head $H$; however, we omit discussion of those issues in this overview).} %

{In the current context, instead of a decision list of the form $H \wedge L$ where $L$ is a ``dense'' decision list, we have a factored-DNF of the form $H \wedge f_1$, where $f_1$ is an $s$-term DNF depending on $\poly(s)$ variables.  Standard-model testing algorithms are known for the class of $s$-term DNFs \cite{Bshouty20}, and it is natural to think that by swapping in such an algorithm for the standard-model decision-list testing algorithm used in \cite{rel-error-conj-DL}, we might be able to obtain a testing algorithm for factored-DNFs.  However, the standard-model $s$-term DNF testers do not suffice for our purpose, because an $s$-term, $\poly(s)$-variable DNF $f_1$ need not be ``dense;'' it could have only a $2^{-\poly(s)}$ fraction of satisfying assignments (this would be the case if, for example, each term were of length $s$). }

Instead, %
our approach to testing factored-DNFs brings in ideas and tools from a different paper \cite{rel-error-junta}.  
The main result of \cite{rel-error-junta} is a $\poly(k/\eps)$-algorithm~for~relative-error testing of any subclass of $k$-juntas which is closed under permuting input variables; note that the class of $s$-term DNFs which depend on only $k=\poly(s)$ many variables clearly meets this condition.  At a high level, our approach to testing factored-DNFs is to use a \emph{relative-error} tester for $s$-term, $\poly(s)$-variable DNFs in place of the (standard-model) decision list tester {to ``\emph{implicitly learn}'' an approximation of $f_1$ (up to small relative error) within the framework of the overall \cite{rel-error-conj-DL} tester.}

We give more details on how all this is done later in \Cref{sec:factored-tester}, but before moving on, let's give a bit more detail on the concept of ``implicit learning,'' or ``learning up to the names of the relevant variables.''
By ``implicit learning'' an approximation 
of the tail $f_1$ (i.e., an $s$-term, $\poly(s)$-variable DNF), %
we mean that if $f_1$ depends on $k$ variables from $x_1,\dots,x_n$, the implicit learning algorithm identifies a function $g:\zo^k \to \zo$, depending on variables which we call $y_1, \ldots, y_k$, such that after some (unknown) mapping of the $y_i$ variables to $\{x_1,\dots,x_n\}$, the function $g$ is close to $f$.\footnote{We mention that the ``implicit learning'' approach has been used in a number of prior works, see e.g.~\cite{rel-error-junta,Bshouty20,DLM+:07}.}  In particular, this ``implicit learning'' lets us essentially determine how many terms there are in the tail DNF $f_1$ (having this value is crucial for us, as described earlier). Calling this approximation of the tail $f_1'$, our factored-DNF testing algorithm then checks that   $f$ is indeed relative-error close to $H' \land f'_1$ for some conjunction $H'$. This checking is performed by a \ConsCheck\
subroutine; we remark that this check is more involved than an analogous check in the framework of \cite{rel-error-conj-DL}, since the function $f_1'$ can be sparse (while the analogous function $L$ in the \cite{rel-error-conj-DL} framework is dense).   Full details are presented in  \Cref{sec:checking}.

Finally, we turn to the real scenario: $\testdnf$ running on $h_i$ with access to app\-ro\-ximate simulator oracles $\MQ^*(h_i)$ and $\SAMP^*(h_i)$ only. Let's focus on the case when $h_i$ is a factored-DNF, where we need $\testdnf$ to still accept with probability close to $2/3$ (as it would on $\MQ(h_i)$ and $\SAMP(h_i)$).
As discussed earlier in \Cref{sec:hurdle1}, the main concern is about $\MQ^*(h_i)$: {could the fact that} $\MQ^*(h_i)$ only works well on ``random-ish'' query points significantly lower the probability of $\testdnf$ accepting $h_i$?

Surprisingly, even though $\testdnf$, as discussed above, is a highly involved tester with many algorithmic components (see e.g. \hyperref[fig:diagram]{Figure 2}), 
{all of the membership queries it makes are of}  the following two types:
\begin{flushleft}
\begin{enumerate}
    \item Queries that are ``\emph{random-ish}.''
    Roughly speaking, such a query is for a point $z$ that lies in the cube between
    a point $\bx$ drawn from ${h_i^{-1}(1)}$
    and a point $\by$ drawn uniformly at random from $\{0,1\}^n$, except at a small set of randomly picked coordinates $i$, where {$z_i$} can be set arbitrarily with no randomness.
        We defer the detailed definition (of these ``\emph{type-1}'' queries) to \Cref{sec:testersummary}.
    Abstracting this formulation from the inner workings of $\testdnf$, while ensuring that $\MQ^*(h_i)$ works correctly with such queries, are major challenges we need to overcome.

\item Queries that are not necessarily random-ish but the value returned must be $1$, or else $\testdnf$ rejects immediately. 
It turns out that such (so-called ``\emph{type-$2$}'') queries are easy to handle because 
  the implementation of $\MQ^*(h_i)$ is \emph{one-sided}:
  it always return $1$ on a query $z$ if $h_i(z)=1$ (because roughly speaking, $\testeq$
  never errs when the cube contains satisfying assignments only).
As a result, for type-$2$ queries, running $\testdnf$ on $\MQ^*(h_i)$ can only increase its acceptance probability.
\end{enumerate}\end{flushleft}

In summary, while $\MQ^*(h_i)$ is not a general-purpose membership oracle for $h_i$, it is built specifically to work in tandem with our relative-error tester for factored-DNFs. As a result,  
 the~probability of $\testdnf$ accepting $h_i$  using simulators $\MQ^*(h_i)$ and $\SAMP(h_i)^*$ is at least that of it running on the perfect oracles $\MQ(h_i)$ and $\SAMP(h_i)$ minus $0.01$. %

\medskip

%% file: sections/DNF-preliminaries.tex
\section{Preliminaries}

\subsection{Relative-error testing}
\label{sec:relative-error-testing}

We recall that in the {\it standard} testing model for a class ${\cal C}$ of $n$-variable Boolean functions, the testing algorithm is given oracle access $\MQ(f)$ to an unknown and arbitrary function $f: \zo^n \to \zo$.
The algorithm must output ``yes'' with high probability (say at least 2/3; this can be amplified using standard techniques) if $f \in {\cal C}$, and must output ``no'' with high probability (again, say at least 2/3) if the distance between $f$ and every function $g \in {\cal C}$ is at least $\eps$, where the distance between $f$ and $g$ is measured as
$$ {\frac {|f^{-1}(1) \ \triangle \ g^{-1}(1)|}{2^n}}.
$$

In contrast, a \emph{relative-error} testing algorithm for ${\cal C}$ has oracle access to $\MQ(f)$ and also has access to a $\SAMP(f)$ oracle which, when called, returns a uniform random element $\bx \sim f^{-1}(1)$.
A relative-error testing algorithm for ${\cal C}$ must output ``yes'' with high probability (say at least 2/3; again this can be easily amplified) if $f \in {\cal C}$ and must output ``no'' with high probability (again, say at least 2/3) if $\rd(f,{\cal C}) \geq \eps$, where
$$\rd(f,{\cal C}):=\min_{g \in {\cal C}}\hspace{0.05cm}\rd(f,g)\ \quad\text{and}\ \quad 
\rd(f,g) := {\frac {|f^{-1}(1) \ \triangle \ g^{-1}(1)|}{|f^{-1}(1)|}}.
$$

As mentioned above, a relative-error testing algorithm can call both a $\MQ(f)$ oracle and a $\SAMP(f)$ oracle.  We use the term ``query complexity'' to refer to the total number of calls made to either oracle.

\begin{remark} \label{rem:eps-small}
It is immediate from the definition of relative-error testing that if an algorithm is an $\eps$-relative-error testing algorithm for a class ${\cal C}$, then it is also an $\eps'$-relative-error testing algorithm for ${\cal C}$, for any parameter $\eps' \in [\eps,1/2].$ This observation is useful for us because it allows us to assume, without loss of generality, that we are $\eps$-testing the class of $s$-term DNF for values of $\eps$ that are at most some sufficiently small absolute constant (which in turn means that quantities like $\log(s/\eps)$ are at least some sufficiently large absolute constant).
\end{remark}

We recall the following  ``approximate symmetry'' and ``approximate triangle inequality'' properties of relative distance (see Lemma~9 and~10 of \cite{rel-error-junta} for the straightforward proofs):

\begin{lemma}[Approximate symmetry of relative distance]\label{lem:approx-symetric}
    Let $f,g:\{0,1\}^n \to \{0,1\}$ be functions  such that $\rd(f,g)\leq \epsilon$ where $\epsilon \leq 1/2$. Then $\rd(g,f) \leq 2\epsilon$.
\end{lemma} 

\begin{lemma}[Approximate triangle inequality for relative distance]
\label{lem: approx triangle ineq}
    Let $f,g,h:\{0,1\}^n \to \{0,1\}$ be such that $\rd(f,g)\leq \epsilon$ and $\rd(g,h)\leq \epsilon'$. Then $\rd(f,h) \leq \epsilon+(1+\epsilon)\epsilon'$.
\end{lemma}

\subsection{Oracles and simulators}
Throughout the $\maindnf$ algorithm, we will instantiate different algorithms that are supposed to simulate $\MQ(g)$ or $\SAMP(g)$ for various functions $g$. In particular, an instance of an algorithm that is supposed to simulate the behavior of a $\MQ$ oracle would be denoted $\MQ^*$, and similarly an instance of an algorithm that is supposed to simulate a $\SAMP$ oracle would be denoted $\SAMP^*$. %

A simple example of such a simulator, which we will use later, is the following:
\begin{figure}[h]
\begin{algorithm}[H]
\caption{\protect\hypertarget{Sim-SAMP}{\SampleSub}}\label{alg:Sampfh}
\textbf{Input: } $\SAMP(h)$ of some function $h:\zo^n \to \zo$ and $\emptyset \neq \mathsf{Y} \subseteq [n]$  \\
\textbf{Output: }A point in $\{0,1\}^\mathsf{Y}$.\\
\begin{tikzpicture}
\draw [thick,dash dot] (0,1) -- (16.5,1);
\end{tikzpicture}
    \begin{algorithmic}[1]
    \State Draw $\bz \sim \SAMP(h)$.
    \State Return $\bz_\mathsf{Y}$.
    \end{algorithmic}
\end{algorithm}
\caption[$\SampleSub$]{$\SampleSub$. An example of an algorithm to simulate a $\SAMP$ oracle. {This algorithm will be used later in \Cref{sec:checking}, to simulate a $\SAMP$ oracle for the function $\Gamma$ defined in \ConsCheck. }}
\end{figure}

An algorithm could then execute a statement of the form ``$\SAMP^* \leftarrow \SampleSub({\SAMP(h)},\mathsf{Y})$'' (for appropriately chosen ${\SAMP(h)}, \textsf{Y}$). The algorithm can then call $\SAMP^*$ to get some string in $\bx \in \zo^{\textsf{Y}}$, where the hope is that $\SAMP^*$ behaves exactly like a $\SAMP(g)$ oracle of some function $g:\zo^{\mathsf{Y}} \to \zo$. 

\begin{samepage}
{\begin{definition}
    Let $g:\zo^{\sY} \to \zo$  be an arbitrary function. 
    \begin{flushleft}\begin{enumerate}
        \item We say an algorithm $\textsf{ALG}$ is a \emph{$\delta$-accurate simulator for $\MQ(g)$} if for any $x$, we have that  $\textsf{ALG}$ outputs $g(x)$ with probability at least $1-\delta$. 
        \item Similarly, we say an algorithm $\textsf{ALG}$ is a \emph{$\delta$-accurate simulator for $\SAMP(g)$} if any time it is called,  $\textsf{ALG}$ outputs $z \sim g^{-1}(1)$ with probability at least $1-\delta$.  
    \end{enumerate}\end{flushleft}
If an algorithm is a $0$-accurate simulator for $\MQ(g)$ ($\SAMP(g)$), we say it \emph{perfectly simulates} $\MQ(g)$ ($\SAMP(g)$).  
\end{definition}}
\end{samepage}

\subsection{Notation about strings and functions} \label{sec:prelim-string-function}
We write $[n]$ to denote the set $\{1, \ldots, n\}.$ For a subset $\sS \subseteq [n]$, we write $\overline{\sS}$ to denote $[n] \setminus \sS$. 
We use $\circ$ to denote concatenation.

Given $\sS \subseteq [n]$ and $z \in \{0,1\}^n$, we denote by $z_\sS$ the string in $\{0,1\}^\sS$ that agrees with $z$ on every coordinate in $\sS$. We also use $\overline{z}$ to denote the string with each bit of $z$ flipped, so $\overline{z}_i=1-z_i$ for every $i \in [n]$.

We use standard notation for  restrictions of functions.  For $f: \zo^n \to \zo$ and $u\in \{0,1\}^\sS$ for some $\sS\subseteq [n]$, the function $f{\upharpoonleft_{u}}: \zo^{[n] \setminus \sS}\rightarrow \{0,1\}$ is defined as 
the function 
\[
f{\upharpoonleft_{u}}(x)=f(z),
\quad \text{where} \quad
z_i = 
\begin{cases}
    u_i & \text{if~}i \in \sS,\\
    x_i & \text{if~}i \in [n]\setminus \sS.
\end{cases}
\]
For a finite set $A$ we write ``$\bx \sim A$'' to indicate that $\bx$ is uniform random over $A$.

Given a bijection $\pi:\sX \to \mathsf{Y}$ where $\sX,\sY \subseteq [n]$ and $x \in \zo^\sX$, we denote by $\pi(x)$ the string in $\{0,1\}^{\sY}$ such that $(\pi(x))_j=x_{\pi^{-1}(j)}$ for each $j\in Y$.  Given $f:\zo^{\mathsf{Y}} \to \zo$, we denote by $f_{\pi}: \zo^\sX \to \zo$ the function $f_{\pi}(x)=f(\pi(x))$.

Given $n \geq r \geq 1$ and an injective map $\sigma:[r]\rightarrow [n]$ and an assignment $y \in \{0,1\}^n$, we write $\sigma^{-1}(y)$ to denote the string
  $x \in \{0,1\}^r$
   with $x_i=y_{\sigma(i)}$ for each $i\in [r]$.
Given a $z\in \{0,1\}^r$, we write $\sigma(z)$ to denote
  the string $u\in \{0,1\}^\sS$, where $\sS=\{\sigma(i):i\in [r]\}$ is the image set of $\sigma$ and $u_{\sigma(i)}=z_i$ for each $i\in [r]$.
Given a function $g:\{0,1\}^r\rightarrow \{0,1\}$ over
  $h$ variables and an injective map $\sigma:[r]\rightarrow [n]$,
  we write $g_\sigma$ to denote the Boolean function over
  $\{0,1\}^n$ with
  $g_\sigma(x)=g(\sigma^{-1}(x))$. %

  Finally, we note that simply remapping coordinates does not change relative distance:
\begin{fact}\label{fact:remap} Consider functions $f,g:\zo^\sX \to \zo$ and a bijection $\pi:\sY \to \sX$. We have $${\rd(f_\pi,g_\pi)=\rd(f,g)}.$$ 
\end{fact}

\subsection{Terms, cubes, and distances} \label{sec:TCD}

A \emph{term} $T=\ell_1 \land \cdots \land \ell_k$ is a conjunction of Boolean literals over variables $x_1,\dots,x_n$. The \emph{width} of a term is the number of literals that it contains. We assume without loss of generality that no term in any DNF contains both a variable $x_i$ and its negation $\overline{x_i}$.  

It will often be convenient for us to view a term as a set of literals, $T=\{ \ell_1, \ldots, \ell_k \}$; it will always be clear from context whether we are viewing $T$ as a conjunction or as a set of literals. Hence given two terms $T,T',$ we write ``$T \cap T'$'' to indicate the term which is the conjunction of literals $\land_{\ell \in T \cap T'} \ell$ (note that this is different from the conjunction $T \land T'$ of the two terms, which is the term $\land_{\ell \in T \cup T'} \ell$). Similarly, we write ``$T \setminus T'$'' to denote the term which is the conjunction $\land_{\ell \in T \setminus T'} \ell$.
Finally, for $A$ a set of literals, we write $\var(A)$ to denote the set of variables corresponding to the literals in $A$ (so, for example, if $A=\{x_2,\overline{x_2},x_3\}$ we would have $\var(A)=\{x_2,x_3\}$).

\medskip

It will be useful to have notation for the subcube of $\zo^n$ that is ``spanned'' by two points:

\begin{definition}[Cube]
\label{defn:cube}
Given two $n$-bit strings $a,b\in \{0,1\}^n$ we denote by $\cube(a,b)$ the subcube of $\{0,1\}^n$ that has $a$ and $b$ as antipodal points, i.e.,
$$\cube(a,b):=\big\{z \in \{0,1\}^n: z_i=a_i=b_i\ \text{for all $i$ such that $a_i=b_i$} \big\}$$

\end{definition}

It is easy to see that if $a,b\in \{0,1\}^n$ satisfy the same term $T$ of a DNF $f$, then $\cube(a,b) \subseteq f^{-1}(1)$. The other direction does not necessarily hold, though; as a simple example, consider the DNF $f=(\overline{x_2} \land \overline{x_3}) \vee (\overline {x_1} \land x_3)$ and the points $a=000,b=001.$  We have $\cube(a,b)=\{a,b\} \subseteq f^{-1}(1)$, but $a$ and $b$ do not satisfy the same term of $f$.

We will sometimes measure the distance from a term $T$ to a term $T'$ as the size $|T' \setminus T|$ of the set of all literals that belong to $T'$ but not to $T.$
We will also use the following definition capturing how far a point is from satisfying a term:
\begin{definition}[Distance from a point to a term $T$]
    Given $x \in \zo^n$ and a term $T$, we denote by $$\pointdist(x,T):=\left| \{ \ell : \ell \text{ is a literal in } T \text{~with~} \ell(x)=0 \} \right|.$$
    We say that $x$ is \emph{$k$-far from $T$} if $\pointdist(x,T)\geq k$; an equivalent way of expressing this is to say that ``$x$ disagrees with $T$ on at least $k$ coordinates''.
\end{definition}

\begin{lemma}\label{lem:cubeterm}
    Let $a,b \in \{0,1\}$ and $T$ be a term such that at least one of $a,b$ is $k$-far from $T$.
    Then $$\Prx_{\bz \sim \cube(a, b)}\big[T(\bz)=1\big]\leq 2^{-k}.$$
\end{lemma}
\begin{proof}
    Assume without loss of generality that $a$ disagrees with $T$ on a set $\sS$ of coordinates with $|\sS| \geq k$.
 If $a_i=b_i$ for any coordinate $i \in S$, then no point in $\cube(a, b)$ can satisfy $T$, so the probability in question is zero. So we may assume that  $b$ disagrees with $a$ on all coordinates in $\sS$. In order for a point $z \in \cube(a, b)$ to satisfy $T$, it must set all of the coordinates in $\sS$ correctly. It follows that
    $\Prx_{\bz \in \cube(a, b)}[T(\bz)=1]\leq 1/2^{|\sS|} \leq 2^{-k}.$
\end{proof}

%% file: sections/DNF-clustering.tex
\section{Clustering terms of a DNF}
\label{sec:clustering}

The goal of this section is to establish our  structural result about $s$-term DNFs. This  result shows that given any representation of an $s$-term DNF $f$, and any positive integer $K$, there is a unique partitioning (``$K$-clustering'') of the terms of $f$, denoted $\Clustering(f,K)$,
into disjoint ``clusters'' (sub-DNFs) with useful properties that we will employ later in the ``yes''-case analysis of our relative-error testing algorithm for DNFs.  
 Roughly speaking, the properties we  need about $\Clustering(f,K)$ are
\begin{flushleft}\begin{enumerate}
\item For any cluster $C$ of terms in $\Clustering(f,K)$, all terms
  ``look similar'' to each other (see \Cref{lem:Sstar-size-bound}).
More formally, the sub-DNF of $C$ can be written as
  $T\land (T_1 \lor \cdots \lor T_{|C|})$ such that 
  the number of variables involved in $T_1,\ldots,T_{|C|}$ is bounded
  by $O(s^2K)$. Note that this bound is polynomial in $s$ and $K$ but is independent of $n$.
\item For any $T$ and $T'$ in two \emph{different} clusters of
$\Clustering(f,K)$,  where $T$ is not too wide, we have that $|T'\setminus T|=\Omega(K)$ (see \Cref{cor:K-far-terms-are-in-dif-clusters}).
This implies that it is very unlikely for a uniform random satisfying assignment $\ba \sim \SAMP(f)$ that satisfies a term $T$ to have small distance $\pointdist(\ba,T')$ from any term $T'$ that is not in the same cluster as $T$.  (See \Cref{thm:narrow-term-any-term-outside-cluster-far-from-random-a-OLD} for a precise statement.)
\end{enumerate}\end{flushleft}
Throughout this section we fix $f$ to be an $s$-term DNF, and we fix a particular representation of $f$ as an $s$-term DNF.

\subsection{The $K$-clustering of an $s$-term DNF}

In this subsection we define and explain what is the \emph{$K$-clustering}  of an $s$-term DNF $f$.

\begin{remark}
As mentioned above, given a representation of $f$ as a $s$-term DNF, the $K$-clustering of $f$ is unique: as we describe below, it is obtained by carrying out a deterministic procedure on the set of terms of $f$.  We remark that this deterministic procedure is \emph{not} carried out by our testing algorithm; it is only used for the conceptual purpose of defining and analyzing the $K$-clustering of $f$.
\end{remark}

\begin{definition} \label{def:cluster-label-clustering}
A \emph{cluster $C$ of $f$} is a nonempty subset of the terms of $f$.  Every cluster $C$ has an associated \emph{label} $(T^*,S^*)$,
  where both $T^*$ and $S^*$ are sets of literals,
defined as follows:
\begin{flushleft}\begin{enumerate}
\item $T^*$ is the set of literals shared by all terms in $C$ (equivalently, if we view terms in $C$ as sets of literals, then $T^*$ is their intersection); and 
\item $S^*$ consists of both literals $x_i$ and $\overline{x_i}$ for all $i\in [n]$ such that either (1) $x_i$ appears in some but not all
  terms in $C$ or (2) $\overline{x_i}$ appears in some but not all terms in $C$.
(Connecting the label $(T^*,S^*)$ of a cluster with the first property we need about $\Clustering(f,K)$ mentioned at the beginning of the section, 
  note that the sub-DNF of a cluster $C$ can be written as 
  $T^*\land(T_1\lor \cdots\lor T_{|C|})$ where literals in
  $T_1,\ldots, T_{|C|}$ all come from $S^*$. This is why \Cref{lem:Sstar-size-bound} is about upperbounding $|S^*|$ by $O(s^2K)$ for every cluster in $\Clustering(f,K)$.)
\end{enumerate}\end{flushleft}
From the definition above it is clear that (1) $S^*$ is a set of literals that is closed under negations; and (2) $T^*$ and $S^*$ are disjoint.

Finally, a \emph{clustering of $f$} is a collection of disjoint clusters whose terms partition the $s$ terms of $f$.
As it becomes clear in \Cref{alg:clustering}, the label of each cluster is 
  generated as the cluster is created.
\end{definition}

\Cref{alg:clustering} presents a deterministic procedure that takes as input an $s$-term DNF $f$ as well as a positive integer parameter $K$, and outputs a particular clustering, called the \emph{$K$-clustering of $f$}, denoted $\Clustering(f,K)$. Very roughly speaking, the clustering is formed as follows: A cluster $C$ is formed by starting with the narrowest term as the initial cluster and iteratively adding to it other terms that are ``not too far'' from the cluster (see the while-loop comprising lines 5-9 of \Cref{alg:clustering} for details). 
When no such term exists, the cluster $C$ is finished and is added to $\Clustering(f,K)$, and the procedure repeats (i.e., by starting a new cluster with the narrowest term left) until all terms in $f$ are clustered.
{We further remark that by inspection of lines 4 and 8 (the only lines that set $T^*$ or $S^*$), it is easy to see  that each cluster's label $(T^*,S^*)$ is maintained as stipulated in \Cref{def:cluster-label-clustering}.}

Let $C$ be a cluster in $\Clustering(f,K)$, and let $T$ be the 
  first term added to $C$ (so it must be one of the narrowest terms
  in $C$).
Let $(T^*,S^*)$ be the label of $C$.
Our first lemma about the $K$-clustering of $f$ states that 
  $|T^*\cap S^*|$ cannot be too large compared to $|T|$:

\begin{figure}[t!]
\begin{algorithm}[H]
\label{algo:cluster}
\vspace{0.15cm}\textbf{Input: } An $s$-term DNF $f$ and a positive integer $K$. \\
\textbf{Output: }The $K$-clustering $\Clustering(f,K)$  of the terms of $f$.\\
\begin{tikzpicture}
\draw [thick,dash dot] (0,1) -- (15.9,1);
\end{tikzpicture}
    \begin{algorithmic}[1]
        \State Let $L$ be the set of terms of $f$ and let
           $\Clustering=\emptyset$ initially. 
        \While{$L\ne \emptyset$}

        \State Among terms in $L$ with the smallest width, pick the lexicographically first term $T$.
        \State Remove $T$ from $L$ and create a new cluster $C=\{T\}$ labeled with $(T^*=T, S^*=\emptyset)$. 
        \While{there exists a term $T' \in L$ such that:
        $$\left|\{\ell \in T' \text{ and } \ell \not \in T^* \cup S^* \}\right| \leq 2K $$} %
        \State Pick the lexicographically first such term $T'$ from $L$.
    \State Remove $T'$ from $L$ and add it to $C$.
    \State Update the label of $C$ as follows: $T^* \leftarrow T^* \cap T'$ and
    $S^* \leftarrow  S^* \cup \{\ell, \overline{\ell} : \ell \in (T^* \hspace{0.05cm} \triangle \hspace{0.05cm} T') \}$ 
        \EndWhile
        \State Add $C$ to $\Clustering$.
        \EndWhile
    \State Output $\Clustering$.\vspace{0.1cm}
    \end{algorithmic}
    \caption{Procedure to construct the $K$-clustering of a DNF $f$} \label{alg:clustering}
\end{algorithm}
\caption{The clustering algorithm}
\end{figure}

\begin{lemma}[Cluster label size bound]\label{lem:num-distinct-vars}
Let $C$ be any cluster in $\Clustering(f,K)$ labelled by $(T^*, S^*)$, and let $T$ be the first term added to $C$. Then we have 
$$|T^* \cup S^*|\le |T|+ 4sK.$$ 
\end{lemma}
\begin{proof}
 When the cluster $C$ is first created in line~5 of \Cref{alg:clustering}, the initial $T^* \cup S^*$ contains $|T|$ literals since  $T^*=T$ and $S^*=\emptyset$. 
    Each time the cluster $C$ grows by having a new term $T'$ added to it (in lines 6-8), each new literal that appears in $T^* \cup S^*$ must either belong to $T' \setminus (T^* \cup S^*)$ or be the negation of such a literal. So by virtue of the condition in line~6, each time the cluster grows the number of new literals in $T^* \cup S^*$ can increase by at most $4K$. %
    The lemma follows since the cluster $C$ can grow at most $s$ times.
\end{proof}

\begin{lemma} \label{lem:Sstar-size-bound}
    Let $C$ be any cluster in $\Clustering(f,K)$ labelled by $(T^*, S^*)$.  Then $|S^*| \leq {16s^2K}$.
\end{lemma}
\begin{proof}

Let $T$ be the first term added to $C$ when it was created.
For any additional term $T'$ added to $C$, at the time when it is added,  
  by the condition in line~5 we must have $|T' \setminus (T^* \cup S^*)| \leq 2K$. On the other hand, $|T^* \cup S^*|\le |T|+4sK$ by \Cref{lem:num-distinct-vars}. Since $|T'|\ge |T|$, we have $$|T' \cap (T^* \cup S^*)| \geq |T| - 2K$$ and thus, 
   \[
   |(T^* \cup S^*) \setminus T'| \leq (|T| + 4sK) - (|T| - 2K) \leq 6sK
   \] %
It follows that the number of literals in $(T^* \hspace{0.06cm} \triangle \hspace{0.06cm} T') \setminus S^*$ is at most ${6sK+2K\le 8sK}$; recalling line~8, we get that each time a term $T'$ is added to $C$, the size of $S^*$ grows by at most ${16sK}$.
\end{proof}

\begin{lemma} \label{lem:narrow-dist-same-cluster}
If any two terms $A,B$ of $f$ satisfy $|A|\le |B|+K$ and $|B\setminus A| \leq K$, then they must lie in the same cluster of $\Clustering(f,K).$
\end{lemma}

\begin{proof}
At some point in the execution of \Cref{alg:clustering}, it must be the case that one of $A$, $B$ is added to the current cluster $C$ while the other is still in the list $L$. 

Suppose first that $A$ is added to the current cluster $C$ while $B$ is still in $L$. 
Let $(T^*, S^*)$ be the label of $C$ after $A$ is added.
Then by the definition of labels of clusters, we have 
  $A\subseteq T^*\cup S^*$, and this will continue to hold 
  throughout the execution of the loop spanning lines~5-9.
Since $|B \setminus A| \leq K$, $B$ will eventually be added to $C$ during this loop.

The other possibility is that  $B$ is added to the current cluster $C$
  while $A$ is still in $L$. %
Since $|B \setminus A| \leq K$, we have that $|A \setminus B| \leq K +(|A|-|B|) \leq 2K$.
A similar argument shows that $A$ will eventually be added to $C$ during the while loop.
\end{proof}

Recall that $w_{\min}$ is the minimum width of any term in $f$.

\begin{corollary}\label{cor:K-far-terms-are-in-dif-clusters}
    Let $T,T'$ be two terms of $f$, where $T$ satisfies $|T|\le w_{\min} +K$ and  $T,T'$ are in two different clusters of $\Clustering(f,K)$.
    Then $|T' \setminus T| > K$.
\end{corollary}
\begin{proof}
Given that $|T|\le w_{\min}+K$, we have $|T|\le |T'|+K$.
Then $|T'\setminus T|\le K$ would imply that $T,T'$ lie in the same cluster of 
  $\Clustering(f,K)$ by \Cref{lem:narrow-dist-same-cluster}.
\end{proof}

Looking ahead, we will set $K$ to be $\log^2(s/\eps)$ in the rest of the paper. With this choice of $K$, we can use the following lemma to show that if $\ba\sim f^{-1}(1)$ and $\ba$ satisfies a term $T$ of width $\leq w_{\min}+K$, then for any $T'$ that lies in a different cluster, it is very unlikely (i.e., with probability quasipolynomially small in $s/\eps$) for $\ba$ to be ``close'' to $T'$.

\begin{lemma} \label{thm:narrow-term-any-term-outside-cluster-far-from-random-a}
    Let $g$ be a DNF, $T$ a term of $g$ and $T'$ an arbitrary term such that $|T' \setminus T|=k$. 
    Then, $$\Prx_{\ba \sim g^{-1}(1)}\left[ \pointdist(\ba,T') \leq (k/4) \:  \bigg| \: T(\ba)=1 \right] \leq e^{-k/8} \,.$$
\end{lemma}
\begin{proof}
    Let $\sS$ be the set $\var(T' \setminus T)$. Conditioned on $\ba \sim g^{-1}(1)$ satisfying $T$, we have that {independently} for each $i \in \sS$, either $\ba_i$ is distributed uniformly at random in $\ba_i$ or else $\ba_i$ has the wrong value for $T'$ with probability 1.
For each $i \in \sS$, let $\bold{X}_i \in \zo$ be the binary random variable that takes value $0$ if $\ba_i$ has the right value for $T'$ and takes value $1$ otherwise, and note that the $\bold{X}_i$'s are mutually independent.

    We have that $\pointdist(\ba,T') \geq \sum_{i \in \sS} \bold{X}_i$, and that $\E[\pointdist(\ba,T') ]\geq \E\left[ {\sum_{i \in \sS} \bold{X}_i} \right] \geq |\sS|/2 = k/2$. So by a standard Chernoff bound, we have that
    $$\Prx_{\ba \sim g^{-1}(1)}\left[ \pointdist(a,T') \leq k/4 \:  \bigg| \: T(\ba)=1 \right] \leq e^{-k/8} \,.
    \qedhere$$
\end{proof}

Specializing \Cref{thm:narrow-term-any-term-outside-cluster-far-from-random-a} to our context in which $f$ is an $s$-term DNF, we have the following:
\begin{corollary} \label{thm:narrow-term-any-term-outside-cluster-far-from-random-a-OLD}
    Let $T$ be a term of $f$ with $|T|\le w_{\min}+K$ and let 
    $T'$ be any term of $f$ that is in a different cluster of $T$ in $\Clustering(f,K)$.
    Then, $$\Prx_{\ba \sim f^{-1}(1)}\left[ \pointdist(\ba,T') \leq (K/4) \:  \bigg| \: T(\ba)=1 \right] \leq e^{-K/8} \,.$$
\end{corollary}
\begin{proof}
   Since $T,T'$ are in different clusters and $T$ has width $\leq w_{\min}+K$, by \Cref{cor:K-far-terms-are-in-dif-clusters} we have that $|T' \setminus T| \geq K$. 
   The claim now follows from \Cref{thm:narrow-term-any-term-outside-cluster-far-from-random-a}.
\end{proof}

\subsection{Clusters correspond to Factored-DNFs} \label{sec:clusters-factored-DNFs}

As mentioned in the introduction, this clustering gives rise to a representation of an $s$-term DNF $f$ as a disjunction of {\it factored-DNF}s. We formally define these next:
 
\begin{definition} \label{def:factored-DNF}
  A function $h:\{0,1\}^n\rightarrow \{0,1\}$ is a \emph{$(r,\mu)$-factored-DNF} if it is of the form
  $H\land (T_1\lor \cdots \lor T_{\le r})$ where
\begin{enumerate}
\item $H$ is a term, called the {\it head} of the factored-DNF. (Note that there is no upper bound on the width (number of literals) in $H$.) 
\item $T_1 \lor \ldots \lor T_{\le r}$ is a disjunction of terms, called the {\it tail} of the factored-DNF. 
\item The variables in $H$  are disjoint from the variables in the tail, and the total number of variables appearing in the tail is at most $\mu$.

\end{enumerate}
\end{definition}
In informal intuitive discussions, we sometimes refer to ``factored-DNFs'' without specifying the values of $r$ and $\mu$.
It is easy to check that for each cluster, the disjunction of the terms in the cluster can be equivalently written as an $(r,\mu)$-factored DNF, where the literals in $T^*$ form the head, $r$ is the size of the cluster, and $\mu \leq |S^*|$: 

\begin{claim}\label{claim:subset_of_cluster_is_factored}
    Let $h:\zo^n \to \zo$ be a DNF made of $k$ terms which are all in the same cluster $C$ of $\Clustering(f,K)$. Then $h$ is a $(k,\mu)$-factored-DNF for $\mu=16s^2K$.
\end{claim}
\begin{proof}
    Let $(T^*,S^*)$ be the label of $C$. We have that $h$ can be written as $T^* \land (T_1' \vee \cdots \vee T'_k)$, where $T_1', \ldots, T'_k$ are all written using literals in $S^*$, and by \Cref{lem:Sstar-size-bound} we have $|S^*| \leq \mu$. By \Cref{def:cluster-label-clustering} we have that $\var(T^*) \cap \var(S^*) = \emptyset$ and thus the claim follows. 
\end{proof}

%% file: sections/DNF-main-alg.tex
\def\symd{\hspace{0.05cm}\triangle\hspace{0.05cm}}
\def\ALG{\textsf{ALG}}

\begin{figure}
\begin{algorithm}[H]
\caption{ \protect\hypertarget{Test-DNF}{\maindnf}}%
\label{alg:maintester}
\vspace{0.15cm}\textbf{Input: }$\MQ(f)$ and $\SAMP(f)$ of some function $f: \zo^n \to \zo$, $s$ and $\eps$\\
\textbf{Output: }Accept or reject.\\
\begin{tikzpicture}
\draw [thick,dash dot] (0,1) -- (15.9,1);
\end{tikzpicture}
\begin{algorithmic}[1] 

\State Let $K=\log^2 (s/\eps)$ and $\mu=16s^2 K$.
  
\State Run $\CreateOracles(\MQ(f),\SAMP(f),s,\eps)$.
\State Reject if it rejects;
  otherwise, let $(\MQ^*_i,\SAMP^*_i)$, $i\in [s']$, denote the $s'\le s$ pairs returned. %
\For{each $i\in [s']$}
\For{each $r\in [s]$}
\State Run $\testdnf(\MQ_i^*,\SAMP_i^*,r,\mu,\eps/(2s))$
  for $O(\log s)$ times.
\EndFor
\State Set $r_i$ to be the smallest $r$ such that
  the majority of runs accept; reject
  if no such $r$ exists.
\EndFor
\State Accept if $\sum_{i\in [s']} r_i\le s$; reject otherwise.
\end{algorithmic} 
\end{algorithm}
\caption[\maindnf: Our main algorithm to relative error test DNFs]{\maindnf. An algorithm to test if a function is $\epsilon$ close in relative-distance to an $s$-term DNF}
\end{figure}

\section{Main algorithm and its analysis}

In this section we present the main algorithm, $\maindnf$,
  for testing $s$-term DNFs under relative distance and use it to prove \Cref{thm:main}.
For this purpose, we will state and assume 
  two theorems (\Cref{thm:factoredDNF,thm:createoracles}) that will be established in the rest of the paper.
 
\subsection{Testing factored-DNFs in relative distance}
Looking ahead, the main result of \Cref{sec:factored-tester,sec:learning,sec:checking} combined is the following theorem about testing $(r,\mu)$-factored-DNFs under relative distance.

\begin{theorem}\label{thm:factoredDNF}
There is a randomized algorithm $\testdnf$ that takes inputs 
  (1) $\MQ(h)$ and $\SAMP(h)$ of some function 
  $h:\{0,1\}^n\rightarrow \{0,1\}$,
  (2) positive integers $r$ and $\mu$,
  and (3) a distance parameter $\eps$.
It has query complexity $\poly(r/\eps, {\mu})$  
  and satisfies the following properties:
\begin{flushleft}\begin{enumerate}
\item If $h$ is an $(r,\mu)$-factored-DNF, then it accepts with probability at least $0.9$; and
\item If $\rd(h,h')>\eps$ for every $(r,\mu)$-factored-DNF $h'$, then it rejects with probability at least $0.9$. 
\end{enumerate}\end{flushleft}
\end{theorem}

\subsection{The procedure \CreateOracles}

Assuming \Cref{thm:factoredDNF}, ideally we can use $\testdnf$ to test general $s$-term DNFs under relative distance if we can design a procedure to ``\emph{partition}'' 
  the input function $f:\{0,1\}^n\rightarrow \{0,1\}$ in question as follows:
Given access to $\MQ(f)$ and $\SAMP(f)$,
  we hope the ``partitioning'' procedure, which we will name $\CreateOracles$, can use $\poly(s/\eps)$ queries to create~a list of at most $s$ pairs $(\MQ^*_i,\SAMP^*_i)$, $i\in [s']$ for some $s'\le s$, where 
   $\MQ^*_i$ and $\SAMP^*_i$ are a pair of randomized 
  algorithms  (with access to $\MQ(f)$ and $\SAMP(f)$ and with $\poly(s/\eps)$ query complexity) that implement the 
  membership and sampling oracles of some function $h_i:\{0,1\}^n\rightarrow \{0,1\}$,\footnote{As the reader may suspect, the randomized algorithms
  behind $\MQ_i^*$ and $\SAMP_i^*$ will not be able to implement the desired oracles $\MQ(h_i)$ and $\SAMP(h_i)$ perfectly; we will discuss this issue  in detail later, and we always highlight imperfect oracles simulated by randomized algorithms using $*$ as in $\MQ^*_i$ and $\SAMP_i^*$. In the current overview, the reader may assume that $\MQ_i^*$ and $\SAMP_i^*$ are indeed faithful oracles for $h_i$.} respectively, such that
\begin{flushleft}\begin{enumerate}
\item When $f$ is an $s$-term DNF, the number $s'$ of pairs $(\MQ^*_i,\SAMP^*_i)$ matches the number of clusters in $\Clustering(f,K)$ with $K=\log^2(s/\eps)$, and each 
  cluster $C$  corresponds to a pair $(\MQ^*_i,\SAMP^*_i)$ such that the function $h_i$
  is the sub-DNF formed by terms in cluster $C$.

\item When $f$ is $\eps$-far from $s$-term DNFs in relative distance, the functions $h_i$ satisfy that $h_1^{-1}(1),\dots,h_{s'}^{-1}(1)\subseteq f^{-1}(1)$ and their union is $f^{-1}(1)$.
\end{enumerate}\end{flushleft}
Under these ideal performance guarantees of $\CreateOracles$, we explain below that \hyperlink{Test-DNF}{\textcolor{violet}{\textbf{Test-}\textbf{DNF}}}, presented in \Cref{alg:maintester}, is a $\poly(s/\eps)$-complexity tester for $s$-term DNFs under relative distance.

To see this is the case, we first note that the query complexity 
  of $\maindnf$ is $\poly(s/\eps)$ given that the query complexity of $\testdnf$ is $\poly(s/\eps)$ and 
  each call to $\MQ_i^*$ and $\SAMP_i^*$ makes $\poly(s/\eps)$ queries to $\MQ(f)$ and $\SAMP(f)$.
For the correctness of $\maindnf$, we consider the following two cases:
\begin{flushleft}\begin{itemize}
\item [I.] When $f$ is an $s$-term DNF, 
  it follows from \Cref{lem:Sstar-size-bound} and performance guarantees of $\CreateOracles$ that
  each $h_i$ is a $(|C_i|,\mu)$-factored-DNF, where we write $C_i$ to denote the cluster in $\Clustering(f,K)$ that corresponds to $h_i$. 
It follows from \Cref{thm:factoredDNF}, a standard
  Chernoff bound and a union bound that with high probability,
  $r_i\le |C_i|$ for all $i\in [s']$.
When this happens, $\maindnf$ accepts because
  $\sum_i r_i\le \sum_i |C_i|\le s$.
  
\item [II.] When $f$ is $\eps$-far from every $s$-term DNF in relative distance, let $t_i$ be the smallest integer such that $h_i$ is $(\eps/(2s))$-close to a $(t_i,\mu)$-factored-DNF in relative distance. 
It follows from \Cref{thm:factoredDNF}, a standard
  Chernoff bound and a union bound that with high probability, we have
  $r_i\ge t_i$ for all $i\in [s']$ and thus, $\sum_i r_i\ge \sum_i t_i$.
On the other hand, the following simple lemma shows that, 
  given that $f$ is $\eps$-far from $s$-term DNFs in relative distance, we must have $\sum_i t_i>s$.
As a result, $\maindnf$ rejects with high probability.
\end{itemize}\end{flushleft}

\begin{lemma}\label{lem:partition}
Let $h_1,\dots,h_{s'}$ be functions such that $h^{-1}_1(1),\ldots,h^{-1}_{s'}(1)\subseteq f^{-1}(1)$ and their union is $f^{-1}(1)$. If $h_i$ is $(\eps/(2s))$-close to an $t_i$-term DNF in relative distance for each $i\in [s']$ and $\sum_i t_i\le s$, then $f$ must be $\eps$-close to an $s$-term DNF in relative distance.
\end{lemma}
\begin{proof}
Let $h_i'$ be a $t_i$-term DNF such that 
  $\rd(h_i,h_i')\le \eps/(2s)$.
Let $g$ be the disjunction of $h_i'$, $i\in [s']$.
Then $g$ is an $s$-term DNF given that $\sum_i t_i\le s$, and we have 
$$
\rd(f,g)=\frac{|f^{-1}(1)\symd g^{-1}(1)|}{|f^{-1}(1)|}
\le \frac{\sum_i |h^{-1}_i(1)\symd {h'}_i^{-1}(1)|}{|f^{-1}(1)|}\le 
\sum_i \frac{|h^{-1}_i(1)\symd {h'}_i^{-1}(1)|}{|h^{-1}_i(1)|}\le s'\cdot \frac{\eps}{2s}< \eps,
$$
using $s'\le s$. This finishes the proof of the lemma.
\end{proof}

Unfortunately we were not able to obtain a procedure $\CreateOracles$ that achieves exactly the ideal performance guarantees presented above as items (1.) and (2.)
Before presenting the theorem for $\CreateOracles$,
  we discuss the differences between the real versus ideal guarantees on the $h_i$ functions:
\begin{flushleft}\begin{itemize}
\item [$1'$.] When $f$ is an $s$-term DNF, $h_i$ (i.e., the function implemented by the $i$th pair
  $(\MQ^*_i,\SAMP^*_i)$) is still a DNF but is not exactly the sub-DNF of a cluster in $\Clustering(f,K)$. 
Instead, letting $D_i$ be the set of terms comprising $h_i$,
  we have that $D_i\subseteq C$ for some cluster $C$ in $\Clustering(f,K)$ and we also have that the $D_i$'s are pairwise disjoint, from which we have $\sum_{i} |D_i|\le s$.
(Note that the union of the $D_i$'s does not necessarily 
  cover all terms in $f$. Intuitively, this is because $f$ may contain some ``very wide'' terms, that have very few satisfying assignments, which we are unable to find.) 

\item [$2'$.] When $f$ is $\eps$-far from
  $s$-term DNFs in relative distance, we still have $h_i^{-1}(1)\subseteq f^{-1}(1)$ but\\
  the union of $h^{-1}_i(1)$'s only covers $f^{-1}(1)$ approximately, i.e., 
 $|f^{-1}(1)\setminus (\cup_i h^{-1}_i(1))|$ is small   
 compared to $|f^{-1}(1)|$.
\end{itemize}\end{flushleft}
It is easy to check that if these were the only differences, then the correctness of $\maindnf$ would follow from similar arguments as above.
(In particular, the proof of \Cref{lem:partition} is robust enough to tolerate the approximation in item 2 above.)

The main challenge, however, is that the randomized algorithms behind $(\MQ^*_i,\SAMP^*_i)$ are not perfect implementations of oracles for $h_i$. 
To this end, we introduce the following definition:

\begin{definition}\label{def:adequate}
Let $\ALG$ be an algorithm that has access to 
  $\MQ(h)$ and $\SAMP(h)$ of 
  some function $h:\{0,1\}^n\rightarrow \{0,1\}$, which 
  either accepts or rejects at the end.
We say $(\MQ^*,\SAMP^*)$ is an \emph{adequate implementation of oracles for  $h:\{0,1\}^n\rightarrow \{0,1\}$ with respect to} $\ALG$ if
\begin{flushleft}\begin{enumerate}
\item $\MQ^*$ is a randomized algorithm that, given
  any $x\in \{0,1\}^n$, returns either $0$ or $1$; 
\item $\SAMP^*$ is a randomized algorithm that always returns a point $x\in \{0,1\}^n$; and
  
\item We have 
$$
\Big|\Pr\left[\ALG\big(\MQ(h),\SAMP(h)\big)\ \text{accepts} \right]-\Pr\left[\ALG\big(\MQ^*,\SAMP^*\big)\ \text{accepts}\right]\Big|\le 0.1.$$
\end{enumerate}
\end{flushleft}
\end{definition}

Let $K=\log^2(s/\eps)$ and $\mu=16s^2K$. 
With \Cref{def:adequate} in place, we state the main theorem about $\CreateOracles$, which we prove in
  \Cref{sec:pooling}.

\begin{theorem}\label{thm:createoracles}
$\CreateOracles$ takes as inputs (1) access to $\MQ(f)$ and $\SAMP(f)$ of some function $f:\{0,1\}^n\rightarrow \{0,1\}$ and (2) two parameters $s$ and $\eps$. It makes~$\poly(s/\eps)$ queries to $\MQ(f),\SAMP(f)$ and at the end, either rejects or returns a list of pairs $(\MQ_1^*,\SAMP_1^*),$ $\ldots, (\MQ_{s'}^*,\SAMP_{s'}^*)$ 
such that the following conditions hold:

\begin{flushleft}\begin{enumerate}
\item For each $i\in [s']$, $\MQ_i^*$ and $\SAMP_i^*$ are two
  randomized algorithms that have access to $\MQ(f)$ and $\SAMP(f)$ and have query 
  complexity $\poly(s/\eps)$;
\item 
When $f$ is an $s$-term DNF, with probability at least $0.9$, $\CreateOracles$ does not reject and 
there exist $s' \leq s$ Boolean functions $\smash{h_1,\ldots,h_{s'}:\{0,1\}^n\rightarrow \{0,1\}}$ such that (a) Each $h_i$ is a $k_i$-term DNF such that $\sum_{i=1}^{s'}k_i \leq s$;
(b) Each $h_i$ is a $(k_i,\mu)$-factored-DNF; and 
(c) For each $i\in [s']$, $(\MQ^*_i,\SAMP^*_i)$ is an adequate implementation of oracles for $h_i$ with respect to $\testdnf(\cdot,\cdot, {k_i},\mu,\eps/(2s))$.

\item When $f$ is $\eps$-far from every $s$-term DNF in relative distance, with probability at least $0.9$, 
  $\CreateOracles$ either rejects or there exist $s' \leq s$
  functions $h_1,\ldots,h_{s'}:\{0,1\}^n\rightarrow \{0,1\}$ such that
  (a$'$) For each $i\in [s']$, we have $h^{-1}_i(1)\subseteq f^{-1}(1)$ and
\begin{equation}\label{eq:hehe1}
\frac{|f^{-1}(1)\setminus (\cup_{i} h^{-1}_i(1))|}{|f^{-1}(1)|}\le \frac{\eps}{2}; 
\end{equation}
and (b$'$) For each $i\in [s']$ and for every $r \in [s]$, $(\MQ^*_i,\SAMP_i^*)$ is an adequate  implementation of oracles for $h_i$ with respect to $\testdnf(\cdot,\cdot,r,\mu,\eps/(2s))$.
\end{enumerate}\end{flushleft}
\end{theorem}

\subsection{Proof of \Cref{thm:main} assuming \Cref{thm:factoredDNF} and \Cref{thm:createoracles}}

First we note that the query complexity 
  of $\maindnf$ is clearly $\poly(s/\eps)$ given that  by \Cref{thm:factoredDNF} (and the choice of $\mu=16s^2K$ and $K=\log^2 (s/\eps)$) the query complexity of $\testdnf$ is $\poly(s/\eps)$ and 
  each call to $\MQ_i^*$ and $\SAMP_i^*$ makes $\poly(s/\eps)$ queries to $\MQ(f)$ and $\SAMP(f)$~by \Cref{thm:createoracles}.
For the correctness of $\maindnf$, we consider the following two cases:
\begin{itemize}
    \item When $f$ is an $s$-term DNF, 
  it follows from \Cref{thm:createoracles} and \Cref{lem:Sstar-size-bound} that
  each $h_i$ is a $(k,\mu)$-factored-DNF. It then follows from \Cref{thm:factoredDNF} and \Cref{thm:createoracles} (the adequate implementation part), a  
  Chernoff bound and a union bound that with high probability,
  $r_i\le k_i$ for all $i\in [s']$.
When this happens, $\maindnf$ accepts because
  $\sum_i r_i\le \sum_i k_i \le s$, using the guarantee from \Cref{thm:createoracles}
  that $\sum_{i=1}^{s'} k_i \leq s$.

\item When $f$ is $\eps$-far from every $s$-term DNF in relative distance, by \Cref{thm:createoracles}, with probability~at least $0.9$ $\CreateOracles$ either rejects (in which case $\maindnf$ rejects) or returns a list of pairs of oracles that satisfy the conditions detailed in \Cref{thm:createoracles} for this case with corresponding functions $h_1,\dots,h_{s'}$. Let $t_i$ be the smallest integer such that $h_i$ is $(\eps/(2s))$-close to a $(t_i,\mu)$-factored-DNF in relative distance.~It~then follows from \Cref{thm:factoredDNF}, a  
  Chernoff bound and a union bound that with high probability, we have
  $r_i\ge t_i$ for all $i\in [s']$ and thus, $\sum_i r_i\ge \sum_i t_i$.
On the other hand, the following generalization of \Cref{lem:partition}
  shows that, 
  given that $f$ is $\eps$-far from $s$-term DNFs in relative distance, 
  under the conditions about $h^{-1}_i(1)$ vs $f^{-1}(1)$ given in  
  \Cref{thm:createoracles},
  we must have $\sum_i t_i>s$.
As a result, $\maindnf$ rejects with high probability.

\end{itemize}

\begin{lemma} 
\label{lem:partition2}
Let $h_1,\dots,h_{s'}$ be functions that satisfy $h_i^{-1}(1)\subseteq f^{-1}(1)$ for each $i$ and \Cref{eq:hehe1}. 
If $h_i$ is $(\eps/(2s))$-close to a $t_i$-term DNF in relative distance for each $i\in [s']$ and $\sum_i t_i\le s$, then $f$ must be $\eps$-close to an $s$-term DNF in relative distance.
\end{lemma}
\begin{proof}
For each $i \in [s']$ let $h_i'$ be a $t_i$-term DNF such that 
  $\rd(h_i,h_i')\le \eps/(2s)$. Let $g= h_1' \vee \cdots \vee h_{s'}$. 
Then $g$ is an $s$-term DNF given that $\sum_i t_i\le s$, and we have 
$$
\rd(f,g)=\frac{|f^{-1}(1)\symd g^{-1}(1)|}{|f^{-1}(1)|}
\le  \frac{ |(\cup_i h_i^{-1}(1))\symd g^{-1}(1)|}{ |f^{-1}(1)| }+\frac{\eps}{2},%
$$
using the first and then the second part of \Cref{eq:hehe1}, respectively. 
On the other hand, we have 
$$
\frac{|(\cup_i h_i^{-1}(1))\symd g^{-1}(1) |}{ |f^{-1}(1)|} \le \frac{\sum_i |h^{-1}_i(1)\symd {h'}^{-1}_i(1)|}{|h^{-1}_i(1)|}\le \frac{\eps}{2} 
$$
using $h_i^{-1}(1)\subseteq f^{-1}(1)$ and $s'\le s$. 
As a result we have $\rd(f,g)\le  (\eps/2)+(\eps/2)= \eps$. 
\end{proof}

%% file: sections/DNF-Create-Oracles.tex
\def\calS{\mathcal{S}}
\def\calP{\mathcal{P}}
\def\TS{\textsf{TS}}
\def\calG{\mathcal{G}}

\def\XX{\mathbf{X}}
\def\str{\textbf{str}}
\def\covered{\textsf{terms}}
\def\regularstr{\text{str}}

\section{Creating Oracles: 
 $\CreateOracles$ and the proof of 
\Cref{thm:createoracles}} \label{sec:pooling}

In this section we present the procedure $\CreateOracles$ and prove 
  \Cref{thm:createoracles}.

Let $f:\{0,1\}^n\rightarrow \{0,1\}$ be the unknown function that we would like to test, with relative error parameter $\eps$, for being an $s$-term DNF.
As in earlier sections, we may assume that 
  $s/\eps$ is at least some sufficiently large absolute constant.
  
We use the following parameters in the rest of this section:

\begin{itemize}

    \item As defined earlier, we have $K:=\log^2(s/\eps)$ and $\mu:=16s^2 K$.
    
    \item Let $Q=\poly(s/\eps)$ be an upper bound on the query complexity (total number of $\MQ(f)$ and $\SAMP(f)$ calls) 
  of $\testdnf(\cdot,\cdot,r\le s, \mu,\eps/(2s))$.
    \item Let $N$ be a sufficiently large polynomial~of $s/\eps$ such that
\begin{equation}\label{eq:choiceN}
N\ge Q^{10}\quad\text{and}\quad N\ge (s/\eps)^{10}.
\end{equation}

    \item Finally, let $\alpha=\log N=O(\log (s/\eps))$.

\end{itemize}

\subsection{Overview of the key aspects of $\testdnf$: 
\red{Type 1} and \cyan{Type 2} queries}\label{sec:testersummary}

To simplify the connection between $\testdnf$ and $\CreateOracles$, let us summarize the properties that we need about $\testdnf$ to prove the adequateness of the oracles that $\CreateOracles$ outputs for $\testdnf$. In our later arguments (in particular, to establish part (c) of item~2 and part (b$'$) of item~3 of \Cref{thm:createoracles}) it will suffice to consider these properties.

As mentioned above, $Q=\poly(s/\eps)$ is an upper bound on the maximum total query complexity of $\testdnf(\MQ(h),\SAMP(h),r\le s,\mu,\eps/(2s))$ for any function $h:\{0,1\}^n\rightarrow \{0,1\}$.
 By inspection of $\testdnf$ and the subroutines that comprise it, it can be verified that throughout its execution, $\testdnf$ makes two different kinds of queries to $\MQ(h)$, which we call \red{Type 1} and \cyan{Type 2} queries.  In the code of each subroutine of $\testdnf$, each occurrence of one of these types of queries is indicated by the corresponding color (red for \red{Type 1}, blue for \cyan{Type 2}).  
 
 Before making any queries either of \red{Type 1} or \cyan{Type 2}, $\testdnf$ randomly partitions $[n]$ into $\tau:=\Theta(\mu^2)$ many blocks $\bX_1,\ldots,\bX_m$; this is done in Phase~1 of $\FRB$.  Now we can describe each of the two types of queries:

\begin{flushleft}\begin{enumerate} 

\item[] {\bf \red{Type 1} queries:}  A \red{Type 1} query has the following structure:  Once and for all (before making any of the \red{Type 1} queries), $\testdnf$ draws a collection of at most $Q$ points $\by\sim \SAMP(h)$ and at most $Q$ points $\bw$ uniformly from $\{0,1\}^n$.
Let $\mathbf{Y}$ and $\mathbf{W}$ denote these two sets of points, respectively. A \red{Type 1} query is a query to $\MQ(h)$ on a string $z$ that has the following structure:  
There exist $\by\in \mathbf{Y}$, $\bw\in \mathbf{W}$ and one block $\bX_\ell$ such that for any $i\in [n]\setminus \bX_\ell$, we have $z_i\in \{\by_i,\bw_i\}$.  The bits of $z$ that lie in $\bX_\ell$ may be set arbitrarily.
\red{Type 1} queries are made in $\FRB$ (a subroutine of $\DNFLearner$). Furthermore, $\FRB$ returns a list of oracles which will be used throughout \hyperlink{Test-Factored-DNF}{\textcolor{violet}{\textbf{Test-}}\textcolor{violet}{\textbf{Factored-}}\textcolor{violet}{\textbf{DNF}}}; these oracles use \red{Type 1} queries only. \footnote{The reader might notice that $\FRB$ only partitions $\sR \subseteq [n]$ into blocks, and not all of $[n]$. Indeed, for a coordinate $i \in \sU$, we will always have $z_i=\by_i$. So it's easy to see that this indeed is a special case of \red{Type 1} queries.}

\item[] {\bf \cyan{Type 2} queries:}
A \cyan{Type 2} query is a query to 
$\MQ(h)$ on an arbitrary string in $\zo^n$.  The crucial property of \cyan{Type 2} queries is that %
$\testdnf$ rejects if $h(z)=0$ for any  of the \cyan{Type 2} queries.  
\cyan{Type 2} queries are made in $\ConsCheck$.

\item[] More specifically, after $\FRB$ has been run, the algorithm only queries $\MQ(h)$ through the oracles returned by $\FRB$ or using \cyan{Type 2} queries.

\end{enumerate}\end{flushleft}

To prove that our simulators are an adequate implementation for $h_1, \ldots, h_{s'}$ with respect to $\testdnf$ we will set things up so that in the ``no case'', each $\MQ_i^*$  perfectly implements $\MQ(h_i)$ (See \Cref{sec:part-1-no-case}). 
However, for the ``yes case'' we have to analyze both types of queries separately. We will show that if ``things go well'' during the execution of $\CreateOracles$, then $\MQ_i^*$ answers $1$ for all \cyan{Type 2} queries where $h(z)=1$ and with high probability (over the random partition $\bX_1, \ldots, \bX_m$ and the draw of $\mathbf{Y}$ and $\mathbf{W}$) $\MQ^*_i$ answers correctly all \red{Type 1} queries (See \Cref{sec:queries}). 

\subsection{Two useful subroutines and their analysis when $f$ is a DNF}

In this subsection we present and analyze two  subroutines that are heavily used by the \hyperlink{Find-Factored-DNFs}{\textcolor{violet}{\textbf{Find-}\textbf{Factored-}\textbf{DNFs}}} algorithm.  Before presenting these two subroutines, let us provide some intuition.  At an intuitive level, given a draw $\ba \sim \SAMP(f)$, in the yes-case (when $f$ is an $s$-term DNF) it would be helpful to be able to determine the terms of $f$ that are satisfied by $\ba$. This goal is unfortunately out of reach, since we should not expect to be able to identify any of the terms of $f$ with an $O_{s,\eps}(1)$-complexity algorithm.  However, as discussed earlier (see \Cref{defn:cube} and subsequent discussion),
if $a,b\in \{0,1\}^n$ satisfy the same term of $f$ then  $\cube(a,b)$ is guaranteed to be a subset of $f^{-1}(1)$. Building on this,  we define and analyze two simple subroutines, $\testeq$ and $\inpool$, which give us some partial information about ``how close'' $a$ and $b$ are to satisfying the same term of $f$. 

First we need the following simple definition:

\begin{definition}
\label{def:dense} Given two points $a,b\in \{0,1\}^n$, we say that $\cube(a,b)$ is \emph{dense} (in $f$) if at least $(1/N^2)$-fraction of points in $\cube(a,b)$ are satisfying assignments of $f$.
\end{definition}

\begin{figure}[h]

\begin{algorithm}[H]
\caption{\protect\hypertarget{Test-Equivalence}{\testeq}}%
\label{alg:test-weak-equivalence}
\vspace{0.15cm}\textbf{Input: }$a,b \in \zo^n$, $\MQ$ access to $f: \zo^n \to \zo$.\\
\textbf{Output: }True or False.\\
\begin{tikzpicture}
\draw [thick,dash dot] (0,1) -- (15.9,1);
\end{tikzpicture}
\begin{algorithmic}[1] 
\State Repeat %
  ${N^3}$ times: Sample $\bz \sim \cube(a,b)$ and query $f(\bz)$.
\State If $f(\bz)=1$ for some $\bz$ sampled, return True; otherwise, 
  return False.
\end{algorithmic} 
\end{algorithm}
\caption[$\testeq$]{$\testeq$. The algorithm takes as input two points $a,b \in \zo^n$. If $\cube(a,b)$ is dense, the algorithm returns True with high probability. When $f$ is a DNF, if there exists a term $T$ of $f$ such that $T(a)=T(b)=1$ then the algorithm always returns True. }
\end{figure}

\begin{figure}[h]
\begin{algorithm}[H]
\caption{\protect\hypertarget{In-Pool}{\inpool}} \label{alg:in-pool}
\vspace{0.15cm}\textbf{Input: }$a \in \zo^n$, a pool $\mathcal{P}\subseteq f^{-1}(1)$, $\MQ$ access to $f: \zo^n \to \zo$. \\
\textbf{Output: }True or False.\\
\begin{tikzpicture}
\draw [thick,dash dot] (0,1) -- (15.9,1);
\end{tikzpicture}
\begin{algorithmic}[1]

\For{each $b \in \mathcal{P}$}
\State Repeat ${10\alpha}$ times: 
 Sample $\bz \sim \cube(a,b)$ and query $f(\bz)$.  
\State If $f(\bz)=1$ for all $\bz$ sampled, return True.%
\EndFor
\State Return False.

\end{algorithmic} 
\end{algorithm}
\caption[$\inpool$]{$\inpool$. The algorithm takes as input a point $a \in \zo^n$ and a pool $\calP$. When $f$ is a DNF, if there exists a term $T$ of $f$ and $b \in \calP$ such that $T(a)=T(b)=1$, the algorithm returns True. If for every $b \in \calP$ $\cube(a,b)$ is not dense for every $b \in \calP$, the algorithm returns False with high probability.}
\end{figure}

 We will use the following simple lemma about $\testeq$
  when $f$ is an $s$-term DNF:
\begin{lemma} \label{lem:testeq}
Let $f$ be an $s$-term DNF and $a,b \in \{0,1\}^n$
  be two points.
\begin{flushleft}\begin{enumerate}
\item $\testeq(a,b,\MQ(f))$ always 
  returns True if $a,b$ satisfy $T(a)=T(b)=1$ for\\ some term $T$ of~$f$. %

\item If   
  $\cube(a,b)$ is dense in $f$, 
   then $\testeq(a,b,\MQ(f))$ returns True with probability at least $1-2^{-\Omega(N)}$.%
\item If $\max\{\pointdist(a,T) ,\pointdist(b,T)\} \geq K/4$ for every term $T$ of $f$, then,  
 $$\Prx \big[ {\testeq(a, b,\MQ(f) )}\text{ returns True}  \big] \leq 2^{-K/{10}}.$$ 
\end{enumerate}\end{flushleft}
\end{lemma}
\begin{proof}
Item (1) holds because every point $z \in \cube(a,b)$ satisfies $T$ and thus, $f(z)=1$. 

Item (2) follows trivially from the fact that \testeq\ returns True as long as one of the $N^3$ samples lies in $f^{-1}(1)$, and $(1-1/N^2)^{N^3}\le 2^{-\Omega(N)}$.

 For item (3), by \Cref{lem:cubeterm} and a union bound over the $s$ terms of $f$, we have that $$\Prx_{\bz \sim \cube(a,b)}\big[f(\bz)=1\big] \leq s 2^{-K/4}\le 2^{-K/5}$$
using $K\gg \log(s)$. Item (3) follows by a union bound over the $N^3$ samples that are drawn by $\testeq$ draws and using $K\gg \log(N)$.
\end{proof}

We will be interested in certain nonempty subsets $\mathcal{P} \subseteq f^{-1}(1)$ of satisfying assignments, which we will refer to as \emph{pools}. Looking ahead, a useful intuition is to think of a pool $\mathcal{P}$ as consisting of a collection of points $\bz \sim \SAMP(f)$ such that the terms they satisfy are ``close'' to each other, in the sense that they all fall in the same cluster of $\Clustering(f,K)$ (see  \Cref{sec:clustering}).

\begin{definition}
\label{def:covered}
Let $f$ be an $s$-term DNF. Given $x\in f^{-1}(1)$, we write $\covered(x)$ to denote the set of terms of $f$ that $x$ satisfies; we sometimes say that these terms are \emph{covered} by $x$. 
Given a pool $\mathcal{P}\subseteq f^{-1}(1)$, we write
  $\covered(\calP)$ to denote the set of terms of $f$ that are satisfied by at least one $x\in \calP$, i.e., it is the union of $\covered(x)$, $x\in \calP$; we sometimes say that these terms are \emph{covered} by $\calP$.

\end{definition}

We prove the following simple lemma about the subroutine $\inpool$:

\begin{lemma}\label{lem:inpool}
Let $f$ be an $s$-term DNF.
Let $\calP\subseteq f^{-1}(1)$ be a pool and $a \in \zo^n$.
\begin{enumerate}
    \item If $T(a)=1$ for some term $T\in \covered(\calP)$, 
      then $\inpool(a, \mathcal{P}, \MQ(f))$ always returns True.
    \item If for every $b\in \calP$, $\cube(a,b)$ is not dense in $f$, then 
    $$\Prx \big[ \inpool(a,\calP,\MQ(f))\ \text{returns True}  \big] \leq \frac{|\calP|}{N^{20 \alpha}}.$$
\end{enumerate}
\end{lemma}
\begin{proof}
Item (1) is trivial, as
$T(a)=1$ for some $T\in \covered(\calP)$ implies that $T(b)=1$ for some~$b\in \calP$ as well and thus,
  $f(z)=1$ for all $z\in \cube(a,b)$. 
Item (2) follows from \Cref{def:dense}, the fact that $\inpool$ draws $10 \alpha$ independent points, and a union bound over the elements of $\calP$.
\end{proof}

\subsection{The \CreateOracles~procedure}

We are now ready to present the \CreateOracles~algorithm, \Cref{algo:s-term-dnf-part-1}. This algorithm has five main phases; in this subsection we give a high-level overview of these phases, before the detailed presentation of the algorithm, and give the simple proof of the algorithm's query complexity.

The first phase, ``Building Pools,'' draws a collection ${\cal S}$ of ${N^3}$ many assignments, each drawn uniformly at random from the set $f^{-1}(1)$ of satisfying assignments. ${\cal S}$ is then partitioned into a collection of pools ${\cal G}$ based on the outcome of running  \textsf{Test-Weak-Equivalence}$(a,b,\MQ(f))$ on each pair $a,b \in {\cal S}.$ (We note that this property, that ${\cal G}$ is a partition of ${\cal S}$, is maintained throughout the \CreateOracles~algorithm.)  Intuitively, if $f$ is an $s$-term DNF then the different pools resulting from this phase will consist of (some of) the satisfying assignments of different factored-DNFs in a representation of $f$ as a disjunction over a family of factored-DNFs.

The second phase, ``Sanity Check,'' halts and rejects if strictly more than $s$ pools were formed in the first phase (as suggested above, this should not occur if $f$ is actually an $s$-term DNF).

The third phase, ``Merging Pools,'' uses \red{\testeq}~and a small number of fresh samples from $f^{-1}(1)$ to merge some pairs of pools in ${\cal G}$. As was discussed in \Cref{sec:hurdle1} this step is carried out to ensure that property (iii) of \hyperref[fig:decomposition]{Figure 1} holds. 

The fourth phase,
``Removing low weight pools,'' uses \textsf{In-Pool} and another small collection of fresh samples from $f^{-1}(1)$ to identify a subset of the pools in ${\cal G}$ that are ``heavy'' pools.
The high-level intuition here is that heavy pools are the ones for which we can have high confidence that random draws from $\SAMP(f)$ in the rest of the algorithm will satisfy some term covered by that pool.  
We remark that in Phase 4, the ``sufficiently long'' random string $\str_{\calP}$ should be viewed as a huge random string, with enough randomness to make all calls to $\inpool$ during and after this phase fully deterministic. In particular, for every $z \in f^{-1}(1)$, $\str_{\calP}$ contains the random bits required to simulate $\inpool$  deterministically for every possible input $z$.
(Having a global random string will be useful in the analysis when we prove \Cref{thm:createoracles}.)

Finally, in the fifth phase, for each ``heavy'' pool ${\cal P}_i$ we run 
\Cref{algo:simulator-for-MQ} and \Cref{algo:simulator-for-samp}
 to construct $\MQ^*$ and $\SAMP^*$ simulators for a function $h_i$ corresponding to that ``heavy'' pool. Recall that these simulators will be subsequently used in line~6 of the main DNF testing algorithm $\maindnf$.

  \medskip
  
 The remainder of \Cref{sec:pooling} is devoted to the proof of \Cref{thm:createoracles}.
 The following easy claim establishes Item~1 of the theorem, which upper bounds the query complexity:

\begin{figure}[t!]
\begin{algorithm}[H] 
\caption{\protect\hypertarget{Find-Factored-DNFs}{\CreateOracles}}\label{algo:s-term-dnf-part-1} 
\textbf{Input: }$\MQ(f), \SAMP(f)$, a positive integer $s$ and $\eps > 0$. \\
\textbf{Output:} Either reject or a list of at most $s$ pairs of oracle simulators $(\MQ_i^*,\SAMP_i^*)$. \\
\begin{tikzpicture}
\draw [thick,dash dot] (0,1) -- (16.5,1);
\end{tikzpicture}
\begin{algorithmic}[1]
    \Algphase{Phase 1:  Building Pools}
    \State Draw a set $\mathcal{S}$ of $N^3$ points $\bz \sim \SAMP(f)$.
    \State Create an undirected graph $\calG$ with vertex set $\calS$ as follows. %
    \For{each pair of distinct samples $\ba, \bb \in {\cal S}$} 
    \State Run $\testeq(\ba,\bb,\MQ(f))$; add edge $\{\ba,\bb\}$ to $\calG$ if it returns true.
    \EndFor 
    \State Let $\{\calP_i\}$ be the connected
      components of $\calG$, which we refer to as \emph{pools} of $\calG$.\vspace{0.1cm}
    \vspace{-0.1cm}
  \Algphase{Phase 2:  Sanity Check }    
   \State If ${\cal G}$ contains more than $s$ pools, then halt and reject.\vspace{0.1cm}
 \vspace{-0.1cm}
\Algphase{Phase 3:  Merging pools}
\State Draw $N^5$ many samples $\bz \sim \SAMP(f)$.
\State Run $\testeq (\ba, \bz,\MQ(f))$ for every $\ba\in \calS$ and every sample $\bz$.
\If{there exists a sample $\bz$ and $\ba\in \calP$, $\ba'\in\calP'$ from two different pools $\calP,\calP'$ of $\calG$ such that \newline \ \ \ \hskip2em\hspace{0.5cm}\testeq\ returns True on both $(\ba,\bz,\MQ(f))$ and $(\ba',\bz,\MQ(f))$} %
\State Remove $\mathcal{P}, \mathcal{P}'$ and replace them by $\calP\cup \calP'$.
\State Go back to the begining of this phase (line 8).
\EndIf\vspace{0.1cm}

    \vspace{-0.2cm}
   \Algphase{Phase 4:  Removing low weight pools}
\State For each pool $\mathcal{P}$, draw a sufficiently long random string $\str_\calP$ and let $\textsf{counter}_{\cal{P}}=0$.
\State Draw $\omega := \tilde{O}(s/\eps)$ many $\bz \sim \SAMP(f)$.
\For{each sample $\bz$ drawn} 
\For{each pool $\mathcal{P}$ in $\calG$}
\State Run \inpool($\bz, \mathcal{P},\MQ(f))$ with randomness from  $\str_{\calP}(\bz)$.
\State Increment $\textsf{counter}_\calP$ if $\inpool$  returns True.
\EndFor 
\State Halt and reject $f$ if all calls to $\inpool$ on $\bz$ return False. %
\EndFor
\State Let $\He$ be the set of those $\calP$ for which
$\textsf{counter}_{\cal{P}}$ is at least  
  ${\frac \epsilon {20s}} \cdot \omega$.\vspace{0.1cm}
 \vspace{-0.1cm}
\Algphase{Phase 5: Returning Oracles}
\State Return the list containing the following pairs for $\calP\in \He$: %
$$\Big(\MQ^*\big(\cdot,\calP,\MQ(f),\str_\calP\big), \hspace{0.08cm}\SAMP^*\big(\calP,\MQ(f),\SAMP(f),\str_\calP\big)\Big),$$
where $\MQ^*$ and $\SAMP^*$ are given in \Cref{algo:simulator-for-MQ} and \Cref{algo:simulator-for-samp}, respectively.
\end{algorithmic} 
\end{algorithm}
\caption[$\CreateOracles$: An algorithm to generate a list of at most $s$ pairs of oracle simulators $(\MQ_i^*,\SAMP_i^*)$, corresponding to factored-DNFs]{$\CreateOracles$.
The algorithm takes as input $\MQ$ and $\SAMP$ oracles for $f$, a size parameter $s$, and $\eps>0$.  The algorithm attempts to construct ``heavy pools'' of draws from $\SAMP$. It either rejects, or else for each heavy pool it returns a pair of simulated oracles $\MQ^*$, $\SAMP^*$ for a subfunction of $f$ induced by that pool.}
\end{figure}
\clearpage

\begin{figure}[t!]
\begin{algorithm}[H]
    \caption{$\MQ^*(x,\calP,\MQ(f),\regularstr)$}\hypertarget{algo:simulator-for-MQ}{algo:simulator-for-MQ}\label{algo:simulator-for-MQ}
    \textbf{Input: }$x\in \{0,1\}^n$, a pool $\calP\subseteq f^{-1}(1)$, $\MQ$ access to $f$, and a sufficiently long string $\regularstr$.  \\
    \textbf{Output: }Either $0$ or $1$.\\
\begin{tikzpicture}
\draw [thick,dash dot] (0,1) -- (16.5,1);
\end{tikzpicture}
    \begin{algorithmic}[1]
    \State Return $0$ if $f(x)=0$; otherwise,
 run $\inpool(x,\calP,\MQ(f))$ with randomness from $\regularstr(x)$. %
    \State Return $1$ if $\inpool$ returns True; return $0$ otherwise.
    \end{algorithmic}
\end{algorithm}
\caption[$\MQ^*$]{{$\MQ^*$. The algorithm takes as input $x$, a pool $\calP$ and a long string str of random bits.
The goal of the algorithm is to simulate a $\MQ$ oracle for $h_\calP$ (as defined in \Cref{eq:hPdef}). }}
\end{figure}

\begin{figure}[t!]
\begin{algorithm}[H]
    \caption{$\SAMP^*(\calP,\MQ(f),\SAMP(f),\regularstr)$}\hypertarget{algo:simulator-for-samp}{algo:simulator-for-samp}\label{algo:simulator-for-samp}
    \textbf{Input: }A pool $\calP\subseteq f^{-1}(1)$, $\MQ$ and $\SAMP$ access to $f$, and a sufficiently long string $\regularstr$.  \\
    \textbf{Output: }A point in $\{0,1\}^n$.\\
\begin{tikzpicture}
\draw [thick,dash dot] (0,1) -- (16.5,1);
\end{tikzpicture}
    \begin{algorithmic}[1]
    \RepeatN{{$100(s/\eps)\alpha=\tilde{O}(s/\eps)$}}
    \State Draw $\bz\sim \SAMP(f)$; halt and return 
      $\bz$ if $\MQ^*(\bz,\calP,\MQ(f),\regularstr)=1$.

        \End
        \State Return $0^n$. \Comment{Default string: The algorithm failed to get a ``good'' $\bz$.}
    \end{algorithmic}
\end{algorithm}
\caption[$\SAMP^*$]{{$\SAMP^*$. The algorithm takes as input a pool $\calP$ and a long string str of random bits. %
The goal of the algorithm is to simulate a $\SAMP$ oracle for $h_\calP$ (as defined in \Cref{eq:hPdef}). }}
\end{figure}

\begin{claim} [Complexity of $\CreateOracles$] \label{rem:pooling-efficiency}
The total number of calls to $\SAMP(f)$ and $\MQ(f)$ that $\CreateOracles$ makes is at most $\poly(s/\eps)$, and the oracle simulators that it returns each have $\poly(s/\eps)$ query complexity.
\end{claim}

\begin{proof}
 First note that since Phase 3 is only reached if $|{\calG}| \leq s$, the loop in Phase 3 can have no more than $s$ iterations. So, the total number of samples drawn on lines 1 and 8 together is no more than $O(sN^5+s^2/\epsilon)$, which is $\poly(s/\eps)$ by our setting of parameters.
Second, the total number of $\MQ$ calls is at most {$N^3$ times} the number of calls of $\testeq$ in Phase 1 and 3, which is at most $O(sN^6)$.
For Phase 4, it is easy to verify that the total query complexity is $\poly(s/\eps)$ at first we get $\poly(s/\epsilon)$ samples from $\SAMP(h)$, then the outer loop is repeated $\poly(s/\eps)$ times, and the inner loop is repeated at most $s$ times, where each iteration of the inner loop uses $\poly(s/\eps)$ calls to $\MQ(h)$.
Finally, the claim about the query complexity of the oracle simulators is easily verified by direct inspection of \Cref{algo:simulator-for-MQ} and \Cref{algo:simulator-for-samp}.
\end{proof}

\subsection{Analyzing $\CreateOracles$ in the yes case}

Throughout this subsection, we assume that $f:\{0,1\}^n\rightarrow \{0,1\}$ is an $s$-term DNF and we fix some representation of $f=T_1 \lor \ldots \lor T_{\leq s}$ --- in particular, the same representation of $f$ as in the clustering procedure (\Cref{algo:cluster}). Our goal in this subsection is to prove the second item of \Cref{thm:createoracles}, restated here for convenience:
\begin{flushleft}\begin{enumerate}
\item[] \textbf{Item 2 of \Cref{thm:createoracles} ($\CreateOracles$, yes case):}

When $f$ is an $s$-term DNF, with probability at least $0.9$, $\CreateOracles$ does not reject and there exist $s' \leq s$ Boolean functions $\smash{h_1,\ldots,h_{s'}:\{0,1\}^n\rightarrow \{0,1\}}$ such that (a) Each $h_i$ is a $k_i$-term DNF such that $\sum_{i=1}^{s'}k_i \leq s$;
(b) Each $h_i$ is a $(k_i,\mu)$-factored-DNF; and 
(c) For each $i\in [s']$, $(\MQ^*_i,\SAMP^*_i)$ is an adequate implementation of oracles for $h_i$ with respect to $\testdnf(\cdot,\cdot, {k_i},\mu,\eps/(2s))$.
\end{enumerate}\end{flushleft}

\subsubsection{High level idea} Given a pool $\calP$, let 
\begin{equation}
    \label{eq:hPdef}
h_{\calP} := \bigvee_{a \in \calP} \bigvee_{T \in \covered(a)} T.
\end{equation}
In other words, recalling \Cref{def:covered}, $h_{\calP}$ is the DNF consisting of all  terms in $\covered(\calP)$.
The functions $h_1, \ldots, h_{s'}$ in Item 2 of \Cref{thm:createoracles} will correspond to the functions $h_\calP$ when \hyperlink{Find-Factored-DNFs}{\textcolor{violet}{\textbf{Find-}\textbf{Factored-}\textbf{DNFs}}} reaches line 24. With this in mind, observe that when $h_i=h_\calP$, then we have $k_i=|\covered(\calP)|$.
To prove part (a), we show that each term of $f$ is covered by at most one assignment in $\calS$. In turn this will imply that the {terms of $h_{\cal P}$, ${\cal P} \in \He$,} form a partition of (a subset of) the terms in our DNF representation of $f$, and therefore the terms in each $h_{\calP}$ are distinct and the total number of  terms summed over all $h_{\calP}$, ${\cal P} \in \He$, is at most $s$.
To prove part (b), we show that each $h_{\calP}$ is a factored-DNF by showing that all its terms must belong to some single cluster $C$ of the clustering $\Clustering(f,K)$. Then since each cluster $C$ corresponds to a factored-DNF (recall \Cref{sec:clusters-factored-DNFs}), it follows that each $h_{\calP}$ is also a factored-DNF. Finally to prove (c), we show that our oracles simulate functions that are very close to the functions $h_{\calP}$. 

\subsubsection{Preliminary definitions and results}
Before turning to the detailed proof of Item~2 of \Cref{thm:createoracles},
we begin with some preliminary definitions and basic lemmas.

\begin{lemma}
\label{fact1-onepoolterms}
    Every term $T \in f$ is covered by at most one pool. 
\end{lemma}

\begin{proof}
Fix a term $T$ of $f$, and assume we have $\ba,\bb \in \calS$ with $T(\ba)=T(\bb)=1$. By part (1) of \Cref{lem:testeq}, on line~4 of $\CreateOracles$ an edge is added between $\ba$ and $\bb$. Hence, by line~6, we have that  all $\ba \in \calS$ with $T(\ba)=1$ are in the same connected component of $\calG$. This implies that on line~6, $T$ can only be covered by a single pool $\mathcal{P}$.
After line~6, the only way pools are changed is by merging them during Phase~3, so it still holds that $T$ can only be covered by one pool.
\end{proof}

\begin{definition} \label{def:mediumnarrow}
Let $w_{\min}$ be the minimum width of any term in our DNF representation of $f$.
We use $w_{\min}$ and $\alpha$ to classify terms of $f$ according to their width as follows: We say a term of $f$ is \emph{medium} if its width is between $w_{\min}$ and $w_{\min}+2\alpha$; and we say that it is \emph{narrow} if its width is between  $w_{\min}$ and $w_{\min}+\alpha$. (Note that a term of length between $w_{\min}$ and $w_{\min}+\alpha$ is both medium and narrow.)
Finally, if a term of $f$ has width at least $w_{\min}+K$, we say that it is \emph{wide}.
\end{definition}

\begin{lemma}
    \label{oldlemma74}
    With probability at least $0.99$, every medium term of $f$ is covered by at least one $a \in {\cal S}$.
\end{lemma}

\begin{proof}
Fix a medium term $T$ of $f$.  A draw $\bz\sim f^{-1}(1)$ satisfies $T$ with probability at least 
$$
\frac{2^{n-(w_{\min}+2\alpha)}}{s\cdot 2^{n-w_{\min}}}
\ge \frac{1}{sN^2},
$$
recalling that $\alpha = \log N$.
As $\calS$ has $N^3$ samples, the probability that no draw in $\calS$ satisfies $T$ is at most $(1-{\frac 1 {sN^2}})^{N^3} \leq e^{-N/s}$; since $N \gg s$, the lemma follows by a union bound over all (at most $s$) possibilities for the medium term $T$.
\end{proof}

\begin{definition} \label{def:good}
We say that a set $A \subseteq f^{-1}(1)$ is {\it good} if every  $a \in A$ satisfies: 
\begin{itemize}
    \item[(i)] $\covered(a)$ does not contain any wide terms;
    \item[(ii)] Let $T \in \covered(a)$ be such that $|T| \leq w_{\min} +K$. Then for all $T' \in f$ such that $|T' \setminus  T| \geq K$, it holds that $\pointdist(a,T') \geq K/4$. 
\end{itemize}
\end{definition}

\begin{lemma}
    \label{S-is-good}
With probability at least $0.99$, the set ${\cal S}$ drawn on line~1 of $\CreateOracles$ is good.
\end{lemma}

\begin{proof}
The probability that $\ba \sim f^{-1}(1)$ satisfies a fixed wide term $T$ is at most 
$$
\frac{|T^{-1}(1)|}{|f^{-1}(1)|}\le \frac{1}{2^K}.
$$
Therefore, by a union bound, the probability that any $a \in {\cal S}$ satisfies any wide term is at most $s \cdot N^3 2^{-K}< .005$.
For part (ii), assume that every $a \in {\cal S}$ satisfies property (i).  Let $T \in \covered(a)$, and let $T' \in f$ satisfy $|T' \setminus T| \geq K$. 
By \Cref{thm:narrow-term-any-term-outside-cluster-far-from-random-a},   $\pointdist(\ba,T') \leq K/{4}$ for a uniform random $\ba \in f^{-1}(1)$ with probability at most $e^{-K/{8}}$.
Therefore by a union bound over all
$N^3$ elements of ${\cal S}$ and all $s^2$ pairs of terms $T,T' \in f$ such that $T$ is not wide and $|T' \setminus T| \geq K$,%
the set ${\cal S}$ falsifies point (ii) with probability at most $s^2 {N^3} e^{-K/{8}} < .005$. Combining the two $0.005$ failure probabilities, we get that the overall probability that ${\cal S}$ is not good is at most $.01$. 
\end{proof}

\subsubsection{$\CreateOracles$ accepts with high probability}
With the above preliminary definitions and lemmas out of the way, we turn to the proof of Item~2 of \Cref{thm:createoracles}. We first show, in this subsubsection, that
$\CreateOracles$ accepts with high probability.
To upper bound the probability that $\CreateOracles$ rejects, we separately upper bound the two places where it could potentially reject: Phase 2 (line 7) and Phase 4 (line 21).
First, by \Cref{fact1-onepoolterms}, each term $T \in f$ is covered by at most one pool; thus ${\cal G}$ contains at most $s$ pools, so Phase 2 never rejects. 
For Phase 4, a necessary condition for line 21 to reject is that one of the samples $\bz$ drawn in line~15 causes $\inpool(\bz,\calP,\MQ(f))$ to return False for every pool $\calP$ of $\calG$. Given that (with probability at least 0.99) every medium term is covered by a pool (\Cref{oldlemma74}), recalling \Cref{lem:inpool}, a necessary condition for this to happen is that $T(\bz)=1$ for some  $T$ of $f$ that is not medium.
But~the fraction of points in $f^{-1}(1)$ that satisfy a term that is not medium is at most
$$
\frac{s\cdot 2^{n-w_{\min}-2\alpha}}{2^{n-w_{\min}}}
= \frac{s}{N^2}.
$$
Therefore it follows from a union bound over the $\tilde{O}(s/\eps)$ many draws in line~15 that \hyperlink{Find-Factored-DNFs}{\textcolor{violet}{\textbf{Find-}\textbf{Factored-}\textbf{DNFs}}} rejects with probability at most {$.02$}. %

\medskip

It remains to establish parts (a), (b) and (c) of Item~2; we do that in the following subsubsections. 

\subsubsection{Proof of Item 2, part (a)}
    Our aim is to show that if $\CreateOracles$  reaches Phase~5, then it must be the case that $\sum_{\calP \in \He}|\covered(\calP)| \leq s$.
    By \Cref{fact1-onepoolterms}, each term is covered by at most one pool, and since there are at most $s$ terms, the sum of $|\covered(\calP)|$ over all pools $\calP$ created in Phase 1 is at most $s$.  Then since Phase 3 merges pools, and Phase 4 only deals with a subset of the pools, this property is preserved throughout the run of the algorithm. 

\subsubsection{Proof of  Item 2, part (b)}
To prove this, we will show that after Phase~3, for every pool $\mathcal{P}$ we have that  $\covered(\calP) \subseteq C$ for some cluster $C$ in $\Clustering(f,K)$. This suffices to prove Item 2 (b), since each $h_\calP$ is made of $|\covered(\calP)|$ terms and by \Cref{claim:subset_of_cluster_is_factored}, a DNF made of $k$ terms from the same cluster $C$ in $\Clustering(f,K)$ is a $(k,\mu)$-factored-DNF. 

Towards this end, assume that ${\cal S}$ satisfies all of the previous properties from \Cref{oldlemma74,S-is-good}; in particular ${\cal S}$ is good and every medium term is covered by at least one $a \in {\cal S}$.
We first show that for each $a \in {\cal S}$, there exists a cluster $C$ such that $\covered(a) \subseteq C$. 
Let $a \in {\cal S}$, and suppose there exists $T,T' \in \covered(a)$ such that $T,T'$ are in different clusters. Since ${\cal S}$ is good, by property (i) of \Cref{def:good}, $|T|,|T'| \leq w_{\min} + K$, so by \Cref{cor:K-far-terms-are-in-dif-clusters}, we have $|T'\setminus T| > K$. Then by property (ii) of \Cref{def:good}, we have $\pointdist(a,T') \geq K/4$. But this contradicts the fact that for every $T'$ that satisfies $a$, $\pointdist(a,T')=0$. Therefore all terms in $\covered(a)$ lie in the same cluster.

We say that a point $\ba\in \calS$ lies in cluster $C$ if all terms in $\covered(a)$ lie in $C$. 
 Next, we argue that at the end of Phase 1, for every pool $\calP$, all points in $\calP$ lie in the same cluster. 
 The next lemma shows that with high probability we never add an edge between points that lie in different clusters:

\begin{lemma}\label{lem2:goodphase2}
With probability at least $0.99$,
for every $a,b\in \calS$ that lie in different clusters of $\Clustering(f,K)$, $\testeq(a,b,\MQ(f))$ returns False and $(a,b)$ is not an edge in $\calG$.
\end{lemma}
\begin{proof}
Let $C$ be the cluster of $\Clustering(f,K)$ such that $\covered(a)\subseteq C$ and let $C'$ be the cluster with $\covered(b) \subseteq C'$.
Given that $\calS$ is good and $a,b \in \calS$ lie in different clusters, either $\pointdist(a,T)\ge K/4$ or $\pointdist(b,T)\ge K/4$ for any term $T$ of $f$. (This is because the former holds, as argued two paragraphs above, whenever $T\notin C$ and the latter holds whenever $T\notin C'$.) The lemma then follows from item (3) of \Cref{lem:testeq} and a union bound over the at most $N^6$ pairs of $a,b$ belonging to $\calS$.
\end{proof}

From now on we fix $\calS$ and $\calG$ (and thus the pools of $\calG$ at the end of Phase 1) and assume that they satisfy all conditions so far; in particular ${\cal S}$ is good, and the conclusion of 
\Cref{lem2:goodphase2} holds.  %
Under these assumptions, by the end of Phase~1, for each pool $\calP$, all points in $\calP$ lie in the same cluster $C_{\calP}$ in $\Clustering(f,K)$, so $\covered(\calP) \subseteq C$. 
It is left to show that when we merge pools during Phase~3, we still have that for every pool $\calP$, $\covered(\calP)$ is a subset of some cluster $C$ in $\Clustering(f,K)$.

We start with a lemma that is similar to %
\Cref{lem2:goodphase2}:

\begin{lemma}\label{lem2:goodphase4}
With probability at least $0.99$, the set of all points $\bz\sim \SAMP(f)$ sampled in Phase 3 is a good set, and every call to $\testeq(a,z,\MQ(f))$ in Phase 3 returns False when $a,z$ lie in different clusters.
\end{lemma}
\begin{proof}
This follows from arguments similar to those in the proof of %
\Cref{S-is-good,lem2:goodphase2}.
\end{proof}

Consider an iteration of Phase~3, where on reaching line~8 we have that for every pool $\calP$, all its points lie in the same cluster. Suppose that $\ba\in \calP$, $\ba' \in \calP'$ for two different pools $\calP,\calP'$ are such that $\testeq(\ba,\bz, \MQ(f))$ and $\testeq(\ba',\bz, \MQ(f))$ both return True in line~10. By \Cref{lem2:goodphase4} it must be that $\ba, \ba'$ and $\bz$ all lie in the same cluster $C$. By assumption, all points in $\calP$ and $\calP'$ also lie in $C$. So, when we create the pool $\calP \cup \calP'$, every point in this new pool also lies in cluster $C$. So, when we go back to line~8, our assumption about the pools still holds.
Thus, when the algorithm reaches Phase~4, it must be the case that for every pool $\calP$ of $\calG$, $\covered(\calP) \subseteq C$ for some cluster $C$ in $\Clustering(f,K)$. Hence by \Cref{claim:subset_of_cluster_is_factored}, $h_\calP$ is a $(|\covered(\calP)|,\mu)$-factored-DNF, so Item 2, part (b) holds with probability at least 0.97.

\subsubsection{Proof of Item 2, part (c)}\label{sec:queries}

This is the part of Item~2 that requires the most involved argument.
Given a pool $\calP$ and a term $T$ not in $\covered(\calP)$, we say that \emph{$T$ is attracted by $\calP$} if at least a $(1/N^2)$-fraction of points $b\in T^{-1}(1)$ are such that $\cube(a,b)$ is dense for some $a\in \calP$. (Recall that $\cube(a,b)$ is dense if at least $(1/N^2)$-fraction of its points are in $f^{-1}(1)$.)

We begin with two preliminary lemmas giving some useful properties of $\calS$ and $\calG$:

\begin{lemma}\label{lem2:goodphase3}
Conditioned on $\calS$ being good, with probability at least ${0.97}$, the following does not hold  at the end of Phase 3:
There exist a pool $\calP$ and a medium term $T\notin \covered(\calP)$ of $f$ such that $T$ is attracted by $\calP$.
\end{lemma}
\begin{proof}
We first recall that by \Cref{S-is-good}, $\calS$ is good with probability at least 0.99; for the rest of the argument, we condition on $\calS$ being good.
By \Cref{oldlemma74}, with probability at least $1 - 0.01 \cdot {\frac {100}{99}}$ every medium term is covered by a unique pool at the beginning of Phase 3 (the probability bound is as stated because we have conditioned on the at-least-0.99-probability event that $\calS$ is good). This still holds after any number of rounds of merging pools, so we suppose that this holds throughout the rest of the argument below.

Now consider the following event: At the beginning of the $i$-th round of executing lines~8-12, $\calP$ is a current pool that attracts some medium term $T\notin \covered(\calP)$ but $\calP$ does not get merged with the pool that covers $T$ during this round. We show that this happens with probability at most $e^{-\Omega(N/s)}$; given this, the lemma follows by a union bound over all (at most $s$) choices of $i,$ all (at most $s$) choices of $\calP$ and all (at most $s$) choices of $T$.

To this end, since every medium term is covered let $\calP'$ be the pool at the beginning of round $i$ with $T\in \covered(\calP')$  and let $b$ be any point in $\calP'$ with $T\in \covered(b)$.
Given that $T$ is medium and is attracted by $\calP$, we have that the number of points $b\in T^{-1}(1)$ such that $\cube(a,b)$ is dense for some $a\in \calP$ is at least 
$$
\frac{1}{N^2}\cdot 2^{n-w_{\min}-2\alpha},
$$
which is at least {$1/(sN^4)$}-fraction of $|f^{-1}(1)|$. 
Given that we draw $N^5$ samples $\bz$ in this round,
  at least one $\bz$ satisfies that $T(\bz)=1$ and $\cube(a,\bz)$ is dense for some $a\in \calP$ with probability at least $1- e^{-\Omega(N/s)}$.
For this $\bz$, we have from \Cref{lem:testeq} that $\testeq(b,\bz,\MQ(f))$ always returns True and 
  $\testeq(a,\bz,\MQ(f))$ returns True with probability at least $1-2^{-\Omega(N)}$.
As a result, $\calP$ does not get merged with $\calP'$ in this round with probability at most $e^{-\Omega(N/s)}$. 
\end{proof}

\begin{lemma}\label{lem2:farfromnarrow}
    Fix $\calS, \calG$ such that all the {conclusions of} the previous lemmas in this subsection hold.
    Fix a pool $\calP$ of $\calG$ at the beginning of Phase~4 of \CreateOracles. For any narrow $T\in \covered(\calP)$ and any term $T'\notin \covered(\calP)$, we have $|T'\setminus T|\ge \alpha$. 
\end{lemma} 
\begin{proof}
 This is trivial if $T'$ is not medium given that $|T'\setminus T|\ge |T'|-|T|\ge \alpha$. When $T'$ is medium, we claim that we have $|T\setminus T'|\ge 2\alpha$ because otherwise, $T'$ would be attracted by $\calP$, which contradicts the conclusion of \Cref{lem2:goodphase3}.
To see why this is the case, note that for any $a$ with $T(a)=1$ and $b$ with $T'(b)=1$, the assumption that $|T\setminus T'| < 2 \alpha$ implies that the fraction of points in $\cube(a,b)$ that satisfy $T$ is at least $1/2^{2\alpha}=1/N^2$.
This implies that, letting $a\in \covered(\calP)$ be a point with $T(a)=1$, every $b\in {T'}^{-1}(1)$ 
  satisfies that $\cube(a,b)$ is dense, so indeed $T'$ is attracted by $\calP$.
  
Now, using $|T\setminus T'|\ge 2\alpha$ and that $T$ is narrow, we have
$$
|T'\setminus T|=|T'|-|T\cap T'|=|T'|-(|T|-|T\setminus T'|)\ge \alpha,
$$
and the lemma is proved.
\end{proof}

In what follows, we assume that all of the {conclusions}
of the previous lemmas hold.
First note  that after drawing $\str_\calP$ at the beginning of Phase 4, $\MQ^*(\cdot,\calP,\MQ(f),\str_\calP)$ (which we will denote by $\MQ^*_\calP$ for short, and similarly we denote $\SAMP^*(\calP,\MQ(f),\SAMP(f),\str_\calP)$ by $\SAMP^*_\calP$) becomes a deterministic algorithm and corresponds to a Boolean function over $\{0,1\}^n$. We write $g_\calP:\{0,1\}^n\rightarrow \{0,1\}$ (after $\str_\calP$ has been drawn) to denote this function underlying $\MQ^*_\calP$.
{By inspection of \Cref{algo:simulator-for-MQ},} it is clear that we always have $$h_\calP^{-1}(1)\subseteq g_\calP^{-1}(1)\subseteq f^{-1}(1).$$
In particular we have that $\MQ^*_\calP$ is exactly the membership oracle for $g_\calP$ so it always answers $1$ when queried on a point with $h_\calP(x)=1$.
\begin{observation}\label{obs:type2}
      $\MQ^*_\calP$ always answers $1$ on a \cyan{Type 2} query if $h_\calP(z)=1$. 
\end{observation}
The next lemma shows that, with high probability over $\str_\calP$, $g_\calP^{-1}(1)$ is very close to 
  $h_\calP^{-1}(1)$: 

\begin{lemma}\label{lem2:littleextra}
With probability at least $.99$ over $\str_\calP$,  for every pool $\calP$ the following inequality holds:  
$$
\frac{|g_\calP^{-1}(1)\setminus h_\calP^{-1}(1)|}{|f^{-1}(1)|}\le O\left(\frac{s^2}{N}\right).
$$
\end{lemma}
\begin{proof}
We consider a fixed pool $\calP$ and apply a union bound over the at most $s$ pools at the end of the argument.

Recalling \Cref{eq:hPdef} and \Cref{algo:simulator-for-MQ}, for a point $x\in \{0,1\}^n$ to have $g_\calP(x)=1$ but $h_\calP(x)=0$, it must be the case that $T(x)=1$  for some term $T$ of $f$ with $T\notin\covered(\calP)$\footnote{Since $g_{\calP}(x)=1$ we must have $f(x)=1$, so $T(x)=1$ for some term $T$ of $f$. But since $h_\calP(x)=0$, we must have $T \notin\covered(\calP)$. }.  
So we have
\begin{equation}\label{hehe100}
|g^{-1}_\calP(1)\setminus h^{-1}_\calP(1)|\le 
\sum_{T\notin \covered(\calP)} \big|g_\calP^{-1}(1) \cap T^{-1} (1)\big|.
\end{equation}
The sum over terms $T\notin \covered(\calP)$ that are not medium can be easily upper bounded by
$$
\frac{s}{N^2}\cdot |f^{-1}(1)|.
$$
given that $|T^{-1}(1)|\le |f^{-1}(1)|/N^2$ for each term that is not medium.

On the other hand, consider any medium term $T\notin \covered(\calP)$.
By \Cref{lem2:goodphase3}, we have that $T$ is not attracted to $\calP$, so at most a $1/N^2$ fraction of points in $b\in T^{-1}(1)$ are such that $\cube(a,b)$ is dense in $f$ for some $a\in \calP$. Hence, using \Cref{lem:inpool}, we have that the expected size of $|g_\calP^{-1}(1)\cap T^{-1}(1)|$ is at most 
$$
|T^{-1}(1)|\cdot \left(\frac{1}{N^2}+\left(1-\frac{1}{N^2}\right)\frac{|\calP|}{N^{{20 \alpha}}}\right).
$$
Using $|\calP|\le N^3$ and summing over all medium terms $T \notin \covered(\calP)$, the sum over medium terms $T \notin \covered(\calP)$ is at most \[
{\frac 2 {N^2}} \cdot \sum_{\text{medium terms~}T \notin \covered(\calP)}|T^{-1}(1)| \leq
\frac{2{s}}{N^2}\cdot |f^{-1}(1)|.
\]
So the expectation on the RHS of \Cref{hehe100} is at most $$O(s/N^2)\cdot |f^{-1}(1)|.$$
By Markov's inequality, we have that $|g^{-1}_\calP(1)\setminus h_\calP^{-1}(1)|> O(s^2/N)\cdot |f^{-1}(1)|$ with probability at most $1/(100s)$. The lemma follows from a union bound over the at most $s$ pools.
\end{proof}

The next lemma shows that with high probability, 
for every pool $\calP$ that is in $\He$ at the end of Phase 4,  $|g^{-1}_\calP(1)|$ is fairly large.

\begin{lemma} \label{lem2:light-heavy}
With probability at least $0.99$, every pool $\mathcal{P} \in \He$ at the end of Phase 4 satisfies 
     \begin{equation} 
     \label{eq:inpool-lb}\Prx_{\bz \sim f^{-1}(1)}\big[g_\calP(\bz)=1\big] \geq \epsilon/(30s). %
     \end{equation}
\end{lemma}
\begin{proof}
Let ${\cal P}$ be a pool for which the ``true frequency'' $\Prx_{\bz \sim f^{-1}(1)}\big[g_\calP(\bz)=1\big]$ is at most $\epsilon/(30s).$ By a standard Chernoff multiplicative bound, the probability that the empirical frequency %
is at least $\epsilon/(20s)$ after $O((s/\eps)\log(s/\eps))$ samples is at most $1/(100s)$.
The lemma follows from a union bound.
\end{proof}

Let us assume that all the  {conclusions} hold from all of the previous lemmas,
which happens with probability at least $0.9$. 
Thus, $\CreateOracles$~reaches Phase 5 with a collection $\He$ of at most $s$ pools and it returns oracle simulators $(\MQ^*_\calP, \SAMP^*_\calP)$ for each $\calP \in \He$.
Since $\testdnf(\SAMP(h_{\calP}),\MQ(h_{\calP}))$ accepts with probability at least 0.9,

to finish the proof of Part 2, item (c) 
it suffices to show that
\begin{align}\label{eq:step0}
&\Pr\big[\testdnf\big(\SAMP_{\calP_i}^*,\MQ_{\calP_i}^*\big)\ \text{accepts}\big] \nonumber \\[0.5ex]
&\quad\quad\quad\ge  \Pr\big[\testdnf\big(\SAMP(h_{\calP_i}),\MQ(h_{\calP_i})\big)\ \text{accepts}\big]-0.1 
\end{align}
for every pool $\calP_i\in \He,$ 
where throughout this sentence we omitted the parameters ${k_i}\le s,\mu$
  and $\eps/(2s)$ in $\testdnf$ for typographical convenience, {where $k_i = |\covered(\calP_i)|$}.
 We will break the proof of \Cref{eq:step0} into two steps. First we replace $\SAMP_\calP^*$ by $\SAMP(h_\calP)$ and then $\MQ^*_\calP$ by $\MQ(h_\calP)$; in each of the two steps we work to bound the change in the probability that $\testdnf$ accepts.

Recall that $\MQ^*_\calP$ is exactly the membership oracle for $g_\calP$ by the definition of $g_\calP$. On the other hand we show below that $\SAMP^*_\calP$ is very close to the sampling oracle for $g_{\calP}$ which in turn is very close to the sampling oracle $\SAMP(h_\calP)$ for $h_\calP$:

\begin{lemma}
With probability $1-O(s^3/(\eps N))$, $\SAMP^*_\calP$ returns a uniform draw from $h_\calP^{-1}(1)$.
\end{lemma}
\begin{proof}
By inspection of \Cref{algo:simulator-for-samp} we can see that $\SAMP^*_\calP$ returns the first sample $\bz \in g_\calP^{-1}(1)$ it receives, and if no such sample exists it returns the default value $0^n$. Since $\calP \in \He$ by  \Cref{lem2:light-heavy}, the probability that $\SAMP^*_{\calP}$  fails to draw a ``good'' $\bz$ and returns $0^n$ by default is at most
$$
(1-\eps/(30s))^{100(s/\eps)\alpha}\le 1/N.
$$
On the other hand by a union bound over the $100(s/\eps)\alpha$ repetitions in \Cref{algo:simulator-for-samp} and using \Cref{lem2:littleextra}, the probability the algorithm receives a point  $\bz \in \left(g_\calP^{-1}(1) \setminus  h_\calP^{-1}(1)\right)$ is at most 
$$
{O}\big((s/\eps)\alpha\big)\cdot {O(s^2/N)}=O\left(\frac{s^3}{\eps N}\right).
$$
When neither of these two low-probability events occurs, $\SAMP_\calP^*$ must return a uniform draw from $h_\calP^{-1}(1)$.
\end{proof}

Now recall that $Q$ is an upper bound on the number of samples that $\testdnf$ draws from the sampling oracle it is given. As a result, we have 
\begin{align}\label{eq:step1}
&\Pr\big[\testdnf\big(\SAMP_\calP^*,\MQ_\calP^*\big)\ \text{accepts}\big] \nonumber \\ 
&\quad\quad\ge  \left(1-O\left(\frac{s^5Q}{\eps^3 N}\right)\right)\cdot \Pr\big[\testdnf\big(\SAMP(h_\calP),\MQ_\calP^*\big)\ \text{accepts}\big] \nonumber\\
&\quad\quad\ge \Pr\big[\testdnf\big(\SAMP(h_\calP),\MQ_\calP^*\big)\ \text{accepts}\big]-0.01,
\end{align}
using the choice of $N$ that $N\gg Q(s/\eps)^5$.

It remains for us to replace $\MQ^*_{\calP}$ by $\MQ(h_{\calP})$. We will need to handle \red{Type 1} queries, to do this, the following new notation will be helpful:

\begin{definition} \label{def:cubeX}
Given $a,b\in \{0,1\}^n$ and $\sX\subseteq [n]$, we define $\cube_\sX(a,b)$ to be the set of points $y\in \{0,1\}^n$ such that $y_i\in \{a_i,b_i\}$ for all $i\notin \sX$.
\end{definition}
(Note that  unlike the original definition of $\cube(a,b)$, there is no condition on $y_i$ for $i\in \sX$.)

Next, inspired by the summary of how $\testdnf$ works in \Cref{sec:testersummary}, we prove the following lemma. This lemma roughly says that with high probability (over $\by, \bw$ and the random partition of $[n]$), when the algorithm makes a \red{Type 1} query,  regardless of how bits $i \in \bX_\ell$ are set, for the queried point $z$ we don't have $h_i(z)=0$ but $f(z)=1$.  %

\begin{lemma} \label{lem:cubex}
Let $\bX_1,\ldots,\bX_\tau$ be a random $\tau$-way partition of $[n]$ with ${\tau=\Theta(\mu^2)}$.
Let $\by\sim h_\calP^{-1}(1)$ and $\bw\sim \{0,1\}^n$.
Then with probability at least $1-e^{-\Omega(\alpha)}$, it is the case that for every block $\bX_\ell$ we have 
$$
\cube_{\bX_\ell}(\by,\bw)\cap \big(f^{-1}(1)\setminus h_\calP^{-1}(1)\big)=\emptyset.
$$
\end{lemma}
\begin{proof}
Fix an $\ell\in [\tau]$ and a term $T'\notin \covered(\calP)$.
We upperbound the probability of 
$$
\cube_{\bX_\ell}(\by,\bw)\cap {T'}^{-1}(1)\ne \emptyset
$$
by $e^{-\Omega(\alpha)}$ and then apply a union bound over the at most $\tau s$ many choices of $\ell$ and $T'$.
First, we show that with probability at least $1-{O(s^2/(\eps N))}-{e^{-\alpha/8}}$, a draw of $\by\sim h_\calP^{-1}(1)$ satisfies the following two conditions:
\begin{flushleft}\begin{enumerate}
\item $\by$ satisfies a narrow term $T$ in $\covered(\calP)$;
\item $\pointdist(\by,T')\ge \alpha/4$.
\end{enumerate}\end{flushleft}

Since $\calP \in \He$, we have by \Cref{lem2:light-heavy} that $\Pr_{z \sim f^{-1}(1)}[g_\calP(\bz)=1]\geq \epsilon/30s$. On the other hand, \Cref{lem2:littleextra} gives us that 
$$
\Prx_{\bz \sim f^{-1}(1)}[g_\calP(\bz)=1 \land h_\calP(\bz)=0] \leq O(s^2/N).
$$
Since $O(s^2/N) \ll \eps/(60s)$, combining the two equations above we have that $$\Prx_{z \sim f^{-1}(1)}[h_\calP(\bz)=1]\geq \epsilon/60s.$$
Now, for $\bz \sim f^{-1}(1)$, the probability that $\bz$ doesn't satisfy any narrow terms of $f$ is at most $s \times 2^{-\alpha}=s/N$. 
Combining this with the preceding inequality, we get that $$\Prx_{\bz \sim f^{-1}(1)}[T(\bz)=1 \text{ for some narrow $T \in \covered(\calP)$}] \geq \Prx_{z \sim f^{-1}(1)}[h_\calP(\bz)=1]-s/N.$$
Now combining the previous two inequalities, we get that
$$\Prx_{\by \sim h_\calP^{-1}(1)}[T(\by)=1 \text{ for some narrow $T \in \covered(\calP)$}] \geq 1-  60s^2/(\eps N),$$
as desired in item (1.) above.

So let us assume that $\by \sim h_\calP^{-1}(1)$ satisfies some narrow term $T$ in $h_\calP$. By \Cref{lem2:farfromnarrow} we have that $|T'\setminus T|\geq \alpha$. So using \Cref{thm:narrow-term-any-term-outside-cluster-far-from-random-a} with probability at least $1-e^{-\alpha/8}$ we have $\pointdist(\by,T') \geq \alpha/4$, giving (2.)

Given condition (2.) above, by a simple Chernoff bound (now over the randomness in the partition of $\bX_1,\ldots,\bX_\tau$),  taking $\tau$ to be sufficiently large\footnote{Since $\tau$ depends on $\epsilon$, this is without loss of generality as we discussed earlier.}, we have that with probability at least $1-e^{-\alpha/8}$ the following holds:

\begin{itemize}
    \item [3.] the piece $\bX_\ell$ of the random partition contains no more than half of the of indices contributing to $\pointdist(\by,T')$. Call these indices $\sY$.
\end{itemize}
When all of (1.), (2.) and (3.) occur, we have 
$$
\cube_{\bX_\ell}(\by,\bw)\cap {T'}^{-1}(1)= \emptyset
$$
as long as $\ell(\bw)=0$ for some literal $\ell \in T'$ with $\ell(\by)=0$ and the variable for $\ell$ is outside $\bX_\ell$. Hence, this happens with probability at least $1-2^{-\alpha/8}$ over the draw of $\bw$.

Putting these bounds together, the probability that $\cube_{\bX_\ell}(\by,\bw)\cap {T'}^{-1}(1)\ne \emptyset$ is at most 
  $O(s^2/(\eps N))+2e^{-\alpha/8} \leq e^{-\Omega(\alpha)}$.
\end{proof}

We now explain the connection between \Cref{lem:cubex} and \red{Type 1} queries.
Recall how \red{Type 1} queries work: the algorithm first randomly partitions $[n]$ into blocks $\bX_1, \ldots, \bX_{\tau}$ and draws a set $\mathbf{Y}$ of at most $Q$ many strings $\by \sim \SAMP(h)$ and a set $\bW$ of at most $Q$ strings $\bw \sim \zo^n$. A \red{Type 1} query then correspond to a point $z$ where 
$z \in \cube_{\bX_\ell}(\by, \bw)$ for some $\by \in \mathbf{Y}$ and $\bw \in \bW$. Hence, if  $\cube_{\bX_\ell}(\by, \bw) \cap (f^{-1}(1) \setminus h_\calP^{-1}(1))=\emptyset$ we have the following: either $f(z)=h_\calP(z)=0$ or $h_\calP(z)=1$. In the first case, $\MQ^*_\calP$ will always answer $0$, in the later case $\MQ^*_\calP$ will always answer $1$. In particular, this \red{Type 1} query will be answered correctly by $\MQ^*_\calP$. 

It follows that by \Cref{lem:cubex} and a union bound over all $\tau Q^2$ pairs $(\by,\bw) \in [\tau] \times \mathbf{Y} \times \bW$, every \red{Type 1} query is answered correctly with probability at least $1-\tau Q^2e^{-\Omega(\alpha)}$. %
We have already argued that \cyan{Type 2} queries are always answered with $1$ if $h_\calP(z)=1$ (See \Cref{obs:type2}). Since $\testdnf$ immediately rejects when a \cyan{Type 2} query is answered by $0$, it follows that using $\MQ^*_\calP$ instead of $\MQ_\calP$ to answer \cyan{Type 2} queries cannot decrease the probability $\testdnf$ accepts. %
\begin{align}\label{eq:step2}
&\Pr\big[\testdnf\big(\SAMP(h_\calP),\MQ_\calP^*\big)\ \text{accepts}\big] \nonumber \\ 
&\quad\quad\ge    \Pr\big[\testdnf\big(\SAMP(h_\calP),\MQ(\calP)\big)\ \text{accepts}\big] -\tau Q^2\cdot e^{-\Omega(\alpha)}.%
\end{align}
By making $N$ (which is $2^\alpha$) a sufficiently large 
  polynomial of $Q$,

  we can make $\tau Q^2\cdot e^{-\Omega(\alpha)}\le 0.01$. 
\Cref{eq:step0} then follows from \Cref{eq:step1,eq:step2}.
This completes the proof of \Cref{thm:createoracles}, Item 2, part (c), and thus the proof of Item 2 (the yes case of $\CreateOracles$).

\subsection{Analyzing \CreateOracles\ in the no case} \label{sec:part-1-no-case}

Throughout this subsection 
  we take $f$ to be any function $f:\{0,1\}^n\rightarrow \{0,1\}$ that is $\eps$-far from $s$-term DNFs in relative distance, and we prove Item~3
   of \Cref{thm:createoracles}. 
  We restate Item 3 here for convenience:

\begin{flushleft}\begin{enumerate}
\item[] \textbf{Item 3 of \Cref{thm:createoracles} ($\CreateOracles$, no case):}

When $f$ is $\eps$-far from every $s$-term DNF in relative distance, with probability at least $0.9$, $\CreateOracles$ either rejects or there exist $s' \leq s$
  functions $h_1,\ldots,h_{s'}:\{0,1\}^n\rightarrow \{0,1\}$ such that
  (a$'$) For each $i\in [s']$, we have $h^{-1}_i(1)\subseteq f^{-1}(1)$ and
\begin{equation}
\frac{|f^{-1}(1)\setminus (\cup_{i} h^{-1}_i(1))|}{|f^{-1}(1)|}\le \frac{\eps}{2}; \tag{1}%
\end{equation}%
and (b$'$) For each $i\in [s']$ {and for every $r \in [s]$,} $(\MQ^*_i,\SAMP_i^*)$ is an adequate  implementation of oracles for $h_i$ with respect to $\testdnf(\cdot,\cdot,r,\mu,\eps/(2s))$.
\end{enumerate}\end{flushleft}

First we observe that after drawing $\str_\calP$, $\MQ^*(\cdot,\calP,\MQ(f),\str_\calP)$ becomes a deterministic function which we denote by $h_\calP:\{0,1\}^n\rightarrow \{0,1\}$. (So in the no case that we are considering here,
  $\MQ^*(\cdot,\calP,\MQ(f),\str_\calP)$ is exactly the 
  membership oracle $\MQ(h_\calP)$ of $h_\calP$.)
Trivially (by how $\MQ^*$ works) we have that $h^{-1}_\calP(1)\subseteq f^{-1}(1)$ for every pool.

The following lemma shows that it is very unlikely for 
  $\CreateOracles$ to reach Phase 5 (not rejecting) unless the union of
  $h^{-1}_\calP(1)$, $\calP\in \He$, covers almost all $f^{-1}(1)$:

\begin{lemma}
With probability at least $0.9$, $\CreateOracles$ either rejects or else reaches Phase 5 with $\He$ such that both
$$
\frac{|f^{-1}(1)\setminus (\cup_{\calP\in \He} h^{-1}_\calP(1))|}{|f^{-1}(1)|}\le \frac{\eps}{2},
$$ 
and moreover, every $\calP\in \He$ satisfies
\begin{equation}\label{hehe44}
\Pr_{\bz\sim f^{-1}(1)} \big[h_\calP(\bz)=1\big]\ge \frac{\eps}{30s}.
\end{equation}
\end{lemma}
\begin{proof}
We assume that the collection of pools at the beginning of Phase 4 satisfies
\begin{equation}\label{hehe222}
\frac{|f^{-1}(1)\setminus (\cup_{\calP\in \cal} h^{-1}_\calP(1))|}{|f^{-1}(1)|}\le \frac{\eps}{10s};
\end{equation}
otherwise, given that we draw $\tilde{O}(s/\eps)$ samples, Phase 4 rejects with probability at least $0.99$ so 
we are done. 

On the other hand, by a Chernoff bound and a union bound over the randomness of the samples $\bz$ drawn in Phase~4, with probability at least $0.99$ we have 
$$
\Pr_{\bz\sim f^{-1}(1)} \big[h_\calP(\bz)=1\big]\ge \frac{\eps}{30s}.
$$ 
for all $\calP\in \He$ and 
\begin{equation}\label{hehe33}
\Pr_{\bz\sim f^{-1}(1)} \big[h_\calP(\bz)=1\big]\le  \frac{\eps}{10s}.
\end{equation}
for all $\calP$ not picked in $\He$.

Combining \Cref{hehe222} and \Cref{hehe33}, we have 
$$
\frac{|f^{-1}(1)\setminus (\cup_{\calP\in \He} h^{-1}_\calP(1))|}{|f^{-1}(1)|}\le \frac{\eps}{10s}+s\cdot \frac{\eps}{10s}\le \frac{\eps}{2}.
$$
This finishes the proof of the lemma.
\end{proof}

It remains to show that, if the algorithm completes Phase 5 in the no case, then for every $\calP\in \He$, we have %
\begin{align}
&\Big|\Pr\big[\testdnf(\MQ^*_\calP,\SAMP^*_\calP))\ \text{accepts}\Big] \nonumber \\ &\quad\quad-\Pr\big[\testdnf(\MQ(h_\calP),\SAMP(h_\calP))\ \text{accepts}\big]\Big|\le 0.1,\label{eq:step0000}
\end{align}
where we omit the parameters $r,\mu$
  and $\eps/(2s)$ in $\testdnf$ for typographical convenience.
  We recall that $\MQ^*_\calP$ is exactly the same as $\MQ(h_\calP)$ so we just need to show that $\SAMP^*_\calP$ behaves similarly to $\SAMP(h_\calP)$. This is given by the following lemma:

\begin{lemma}
With probability at least $1-1/N$, $\SAMP^*_\calP$ returns a uniform draw from $\SAMP(h_\calP)$.
\end{lemma}
\begin{proof}
Given that every $\calP\in \He$ satisfies \Cref{hehe44}, $\SAMP^*_\calP$ returns the default $0^n$ with probability at most $(1-\eps/(30s))^{100(s/\eps)\alpha}\le 1/N$. When this does not happen, $\SAMP_\calP^*$ returns a uniform draw from $\SAMP(h_\calP)$.
\end{proof}

Given that $\testdnf$ draws at most $Q$ samples from $\SAMP^*_\calP$, by a union bound it gets $Q$ uniform draws from $\SAMP(h_\calP)$ with probability at least $1-Q/N$.
As a result, for every $\calP\in \He$ (for convenience letting $A_\calP$ denote the probability from $\MQ^*_\calP$ and $\SAMP^*_\calP$ and $B_\calP$ denote the probability from $\MQ(h_\calP)$ and $\SAMP(h_\calP)$), we have that 
\begin{align}
 (1-Q/N) B_\calP\le A_\calP\le (Q/N)+(1-Q/N)B_\calP\le (Q/N)+B_\calP,
\end{align}
which implies that $|A_\calP-B_\calP|\le Q/N$ so \Cref{eq:step0000} holds as long as $N\gg Q$.
This completes the proof of \Cref{thm:createoracles}.

%% file: sections/DNF-bounded-tester.tex
\section{Testing factored-DNFs}%
\label{sec:factored-tester}

In this section, we present our tester for $(r,\mu)$-factored-DNFs and prove \Cref{thm:factoredDNF}.
The tester, $\testdnf$, is given in 
  \Cref{alg:testfactoredDNF}.

\begin{figure}[t!]
\begin{algorithm}[H]
\caption{\protect\hypertarget{Test-Factored-DNF}{\testdnf}}
\label{alg:testfactoredDNF}
\vspace{0.15cm}\textbf{Input: }$\MQ(h)$ and $\SAMP(h)$ of some function $h: \zo^n \to \zo$, $r$, $\mu$ and $\epsilon$. \\
\textbf{Output: }Accept or reject.\\
\begin{tikzpicture}
\draw [thick,dash dot] (0,1) -- (15.9,1);
\end{tikzpicture}
\begin{algorithmic}[1] 
    \State Draw $O((\mu\log \mu)/\xi)$ samples $\bz \sim \SAMP(h)$, and call the set of samples $\textsf{Samples}$. Let 
    $$\hspace{-1cm}\bR=\big\{i\in [n]  \mid \exists \bz, \bz' \in \textsf{Samples} \text{ with } \bz_i=1, \bz_i'=0 \big\}$$
 Accept if $\bR=\emptyset$. 

\State Run $\DNFLearner({\bR},\MQ(h), \SAMP(h))$. 
\State Rejects if it rejects;
 otherwise, let $(\MQ^*, \SAMP^*)$ be what $\DNFLearner$ returned.
\State Run $\ConsCheck(\bR,\MQ(h),\SAMP(h), \MQ^*, \SAMP^*)$ and return the same output.
\end{algorithmic} 
\end{algorithm}
\caption[$\testdnf$: An algorithm to relative-error test if a function is an $r$-term factored-DNF.]{{$\testdnf$. The algorithm takes as input oracle access to a function $h$, $r,\mu \in \mathbb{N}$ and $\epsilon>0$, and tests if $h$ is $\epsilon$-close in relative distance to some $(r,\mu)$-factored DNF. 
}}
\end{figure}

$\testdnf$ starts with an initialization stage
  (lines 1--2).
After its initialization, the two main steps that $\testdnf$ carries out are $\DNFLearner$ and $\ConsCheck$,
  which we present and analyze in \Cref{sec:learning} and \Cref{sec:checking}, respectively.
In this section, we prove a simple lemma about what the initialization stage achieves in the yes-case (\Cref{thm:learner_first_phase} below).
After that, 
  we explain what the $\DNFLearner$ and $\ConsCheck$ procedures achieve and formally state their performance guarantees.
Delaying their proofs to \Cref{sec:learning} and \Cref{sec:checking}, we use their performance guarantees to prove \Cref{thm:factoredDNF} at the end of this section.
\subsection{Parameters}

For the sake of clarity, we assume that all three input parameters $r,\mu$ and $\eps$ of \hyperlink{Test-Factored-DNF}{\textcolor{violet}{\textbf{Test-}\textbf{Factored-}\textbf{DNF}}} are shared by the $\DNFLearner$ and $\ConsCheck$ procedures as well as their subroutines. So we don't need to pass them down as input parameters. 
We will also use the following two parameters $\kappa$ and $\xi$ in 
  this section as well as 
  \Cref{sec:learning,sec:checking}:
\begin{equation}
\label{eq:xi}
{\kappa:=\eps^2}\quad\text{and}\quad
\xi:=\frac{\kappa}{4000}  \cdot \frac{1} {\log(100  \cdot 2^{2\mu r})}.
\end{equation}
Intuitively, $\xi$ is a quantity which is ``quite small'' but $1/\xi$ remains $\poly(r,\mu,1/\eps)$.
We assume that both $\kappa$ and $\xi$ are shared among all algorithms as well.

\subsection{Initialization}
\label{sec:initialization}

During its initialization, $\testdnf$ draws
  $O((\mu\log\mu)/\xi)$ samples from $\SAMP(h)$
and let $\bR$ be the set of coordinates $i\in [n]$ that these samples do not agree on. 
Note that when 
  $|h^{-1}(1)|=1$ (in which case $h$ is a trivial $(r,\mu)$-factored-DNF), the algorithm always accepts correctly on line 2;
  when $|h^{-1}(1)|>1$, it is extremely unlikely to have $\bR=\emptyset$ at the end of the initialization.
Without loss of generality, we always assume that 
  $|h^{-1}(1)|>1$ and $\bR\ne \emptyset$ in the 
  rest of the proof.
  
We prove the following lemma about $\bR$ when $h$ is a $(r,\mu)$-factored-DNF.
Recall that when $h$ is a $(r,\mu)$-factored-DNF, we can write it as $h=H \land (T_1 \lor \ldots \lor T_{\leq r})$ where:
\begin{itemize}
    \item $H, T_1, \ldots, T_{\leq r}$ are all terms.
    \item There are at most $\mu$ variables in $T_1, \ldots, T_{\leq r}$. We will refer to this set of variables as $\sS$. 
    \item We have $\var(H) \cap \sS=\emptyset$.
\end{itemize}
Given a set $\sR\subseteq [n]$, we say ${\sR}$ is \emph{stable} (with respect to $h$) if $\var(H)\subseteq \overline{\sR}$ and there exists a string $u\in \{0,1\}^{{\sV}}$, 
  where $\sV:= \overline{\sR} \cap (\sS \cup \var(H) )$, such that 
\begin{equation}\label{eq:hehe100}
    \Prx_{\bz \sim h^{-1}(1)}\big[\bz_{\sV} \neq u\big] \leq {{{\xi}}}.
\end{equation}
The string $u\in \{0,1\}^{\sV}$ is clearly 
  unique when ${\sR}$ is stable.

\begin{lemma}\label{thm:learner_first_phase}
When $h$ is an $(r,\mu)$-factored-DNF,  $\bR$ is stable with probability at least $0.99$.
\end{lemma}
\begin{proof}
First, note that every $\bz \in h^{-1}(1)$ must have $H(\bz)=1$. Hence all the samples must agree on  coordinates in $\var(H)$, meaning that $\var(H) \subseteq \bU$. 

For the second part of the lemma, 
  let 
$\smash{\bu\in \{0,1\}^{\bU}}$ be the unique~string
  such that $\bu_i=\bz_i$ for all $i\in \bU$ and $\bz\in \textsf{Samples}$.
Fix any coordinate $i \in \sS$ and bit $b\in \{0,1\}$ such that $$\Prx_{\bz \sim h^{-1}(1)}\big[\bz_i=b\big] \geq \xi/\mu.$$
The probability of having $i \in \bU$ with $\bu_i=1-b$ is at most
$(1-\xi/\mu)^{\Delta} \leq 1/200\mu,$
  where $\Delta=O((\mu\log \mu)/\xi)$ is the number of
  samples drawn.
Therefore, by a union bound over all at most $\mu$ coordinates in $\sS$ and $b\in \{0,1\}$, we have with probability at least $0.99$, every $i\in \sS\cap \overline{\bR}$ satisfies 
$\Prx_{\bz \sim h^{-1}(1)}[\bz_i=\bu_i] \geq 1-\xi/\mu$. Given that samples always agree on coordinates in $\var(H)$, by a union bound over $\smash{i\in \sS\cap \overline{\bR}}$ we can conclude that
 $$\Prx_{\bz \sim h^{-1}(1)}\big[\bz_{\sV} \neq \bu_{\sV}\big] \leq |\sS| \cdot\frac{\xi}{\mu} \leq \xi.$$
This finishes the proof of the lemma.
\end{proof}

\subsection{Performance guarantees of $\DNFLearner$}

Next we explain what $\DNFLearner$ aims to achieve and state its performance guarantees. 

For this purpose, consider the yes case where we  assume that $h$ is an $(r,\mu)$-factored-DNF
  and~$\sR$ is stable.
Let $u\in \{0,1\}^{\sV}$ be the unique string that satisfies \Cref{eq:hehe100}.
By the definition of $\sV$, every variable~in $\overline{\sR}\setminus {\sV}$ is irrelevant. It follows directly from \Cref{eq:hehe100} that there exists a unique function $\fh:\{0,1\}^{\sR}\rightarrow \{0,1\}$ such that 
$$
\Pr_{\bz\sim h^{-1}(1)} \left[
h{\upharpoonleft_{\bz_{\overline{\bR}}}}\not\equiv \fh
\right]\le \xi.
$$
This function  $\fh$ can  be obtained from $h$ by setting variables in $\sV$ according to $u$ and variables in $\overline{\sR}\setminus \sV$ arbitrarily. 
We will refer to $\fh$ as the \emph{dominating}
  function of $h$ over $\sR$, which is well defined in the yes case when $R$ is stable.
Note that since $h$ is an $(r,\mu)$-factored-DNF, $\fh$ must be an
  $r$-term DNF over at most $\mu$ variables 
  in $\sR$.

In particular, given an $r$-term DNF $D$ which depend on at most $\mu$ variables, we will refer to $D$ as a $r$-term, $\mu$-junta DNF.  

The goal~of $\DNFLearner$~is to \emph{learn the dominating function $\fh$}.
It takes $\sR$, oracles $(\MQ(h),$ $\SAMP(h))$ of $h$ and parameters $r,\mu$ and $\eps$ as inputs. %
Upon termination, it either rejects or outputs a pair of oracle simulators $(\MQ^*,\SAMP^*)$~for some function $\fJ:\{0,1\}^\sR\rightarrow \{0,1\}$, with $\MQ^*$ and $\SAMP^*$ both being randomized algorithms with low query complexity on $\MQ(h)$. 

We summarize the performance guarantees of $\DNFLearner$ (which we will establish later) below:

\begin{restatable}{lemma}{learnerComplexity}\label{lem: complexity of Learner}
$\DNFLearner(\sR,\MQ(h), \SAMP(h))$ makes  
$\poly(r,\mu,1/\epsilon)$ calls to $\MQ(h)$, $\SAMP(h)$.
It either rejects or returns a pair $(\MQ^*,\SAMP^*)$, where $\MQ^*$ and $\SAMP^*$ are randomized algorithms that make $O(\mu\log(\mu/\eps))$ queries to $\MQ(h)$ each time they are called.
\end{restatable}

\begin{restatable}{theorem}{learnerYes}(Performance of $\DNFLearner$ in the yes-case) 
\label{thm: learner yes case}
Assume that $h$ is an $(r,\mu)$-factored-DNF and $\sR\subseteq [n]$ is stable (with respect to $h$).
Then with probability at least $19/20$, 
 $\DNFLearner$ does not reject and the pair $(\MQ^*,\SAMP^*)$ returned are 
    perfect oracle simulators of some function $\fJ:\{0,1\}^\sR\rightarrow \{0,1\}$ with %
    $\rd(\fh,\fJ)\le \kappa,$ where $\fh$ is the dominating function of $h$ over $\sR$.
\end{restatable}

\begin{restatable}{theorem}{learnerNo}
(Performance of $\DNFLearner$ in the no-case) \label{thm: learner no case}
Assume that $h$ is $\eps$-far from $(r,\mu)$-factored-DNFs in relative distance and $\sR\subseteq [n]$.
Then with probability at least $19/20$,
  $\DNFLearner$ either rejects or returns 
    $(\MQ^*,\SAMP^*)$ that are $\kappa$-accurate oracle simulators for some function $\fJ:$ $\{0,1\}^\sR\rightarrow \{0,1\}$ with $\rd(\fJ,\fD)\le \kappa$ for some $r$-term, $\mu$-junta DNF $\fD$.
\end{restatable}

\subsection{Performance guarantees of $\ConsCheck$}

Roughly speaking, when $h$ is a $(r,\mu)$-factored-DNF, $\sR$ is stable,  
  and $\smash{\fJ}$ satisfies $\rd(\fh,\fJ)\le \kappa$,
  then $h$ can be shown to be close to 
  $\smash{C(x_{\overline{\sR}})\cdot \fJ(x_{\sR})}$,
  where $C$ is some conjunction over $\smash{\{0,1\}^{\overline{\sR}}}$.
On the other hand, in the no case, if $\fJ$
  is close to some $r$-term, $\mu$-junta DNF in relative distance,
  then $h$ must be far from $\smash{C(x_{\overline{\sR}})\cdot \fJ(x_{\sR})}$  for any conjunction $C$ over $\smash{\{0,1\}^{\overline{\sR}}}$ because in the no case $h$ is far from $(r,\mu)$-factored-DNFs in relative distance.

Given the above intuition, the goal of $\ConsCheck$ is to use the oracle simulators $(\MQ^*,\SAMP^*)$ for $\fJ$ that are passed down by $\DNFLearner$ to distinguish these two cases.
We summarize its performance guarantees below:

\begin{restatable}{lemma}{constestComplexity}\label{lem: complexity of ConsCheck}
$\ConsCheck(\sR, \MQ(h),\SAMP(h), \MQ^*, \SAMP^*)$
makes no more than $\tilde{O}(1/\eps)$ many calls to $\MQ(h),\SAMP(h),\MQ^*$ and 
  $\SAMP^*$. At the end it either accepts or rejects.
\end{restatable}

\begin{restatable}{theorem}{conscheckYes}(Performance of $\ConsCheck$ in the ``yes'' case)
\label{thm: ConsCheck: yes case}
Assume that $h$ is an $(r,\mu)$-factored-DNF, $\sR$ is stable, and   $\smash{(\MQ^*,\SAMP^*)}$ are perfect simulators for some function $\fJ:\{0,1\}^\sR\rightarrow \{0,1\}$ with $\rd(\fh,\fJ)\le \kappa.$
Then $\ConsCheck$ accepts with probability at least ${39/40}$.

\end{restatable}

\begin{restatable}{theorem}{conscheckNo}(Performance of $\ConsCheck$ in the ``no'' case) \label{thm: ConsCheck: no case}
Assume that $h$ is $\eps$-far from $(r,\mu)$-factored-DNFs in relative distance, $\sR\subseteq [n]$, and 
  $(\MQ^*,\SAMP^*)$ are $\kappa$-accurate oracle simulators for some function $\fJ:\{0,1\}^\sR\rightarrow \{0,1\}$ with $\rd(\fJ,\fD)\le \kappa$ for some $r$-term, $\mu$-junta DNF
  $\fD$.
Then $\ConsCheck$ rejects with probability at least $19/20$.
\end{restatable}

\subsection{Proof of \Cref{thm:factoredDNF}}

With the above pieces in place, the proof of \Cref{thm:factoredDNF} is simple.

\begin{proof}[Proof of \Cref{thm:factoredDNF}]
To bound the complexity of $\testdnf$, we note that (1) it makes $\poly(r,\mu,1/\eps)$ calls to $\SAMP(h)$
  during the initialization;
(2) $\DNFLearner$ makes $\poly(r,\mu,1/\eps)$ calls to 
  $\MQ(h)$ and $\SAMP(h)$; and 
(3) $\ConsCheck$ makes $\tilde{O}(1/\eps)$ calls to $\MQ(h),\SAMP(h),\MQ^*$ and $\SAMP^*$. 
Given that each call to $\MQ^*$ and $\SAMP^*$ makes
  at most $O(\mu\log(\mu/\eps))$ calls to $\MQ(h)$,
  the overall complexity of $\testdnf$ is $\poly(r,\mu,1/\eps)$.

For the correctness, when $h$ is a $(r,\mu)$-factored-DNF, it follows from \Cref{thm:learner_first_phase},
\Cref{thm: learner yes case} and \Cref{thm: ConsCheck: yes case} and a union bound that $\testdnf$ accepts with probability at least $0.9$.
When $h$ is $\eps$-far from $(r,\mu)$-factored-DNFs in relative distance, it follows from 
\Cref{thm: learner no case,thm: ConsCheck: no case} and a union bound that $\testdnf$ rejects with probability at least $0.9$.
\end{proof}

%% file: sections/DNF-learning.tex
\section{The $\DNFLearner$ algorithm} \label{sec:learning}

The $\DNFLearner$ is presented in \Cref{algo:DNF-learner}. 
We start with some high-level intuition of its three phases.  Our high-level discussion focuses on the yes-case, in which $h$ is an $(r,\mu)$-factored-DNF and $\sR\subseteq [n]$ is a stable set (with respect to $h$).
Recall that in this case,
  we write $\fh:\{0,1\}^\sR\rightarrow \{0,1\}$ to denote the dominating function, which is a $r$-term, $\mu$-junta DNF and satisfies 
$$
\Prx_{\bz\sim h^{-1}(1)} \left[
h\upharpoonleft_{\bz_{\overline{\sR}}}\not\equiv \fh
\right]\le \xi.
$$

We first give some notation and setup, then analyze each of the three phases of the algorithm in turn, and finally prove \Cref{thm: learner yes case}, \Cref{thm: learner no case}, and \Cref{lem: complexity of Learner}.
{As will become clear, most of our effort in this section is on the yes-case.
  Arguments for the no-case (\Cref{thm: learner no case}) only occur in the analysis of Phase 3 in \Cref{sec:learner-phase-3}.}

\begin{figure}[t!]
\begin{algorithm}[H]
\caption{\protect\hypertarget{DNFLearner} {$\DNFLearner$}}\label{algo:DNF-learner}
\vspace{0.15cm}
\textbf{Input: }$\emptyset\ne \sR\subseteq [n]$, $\MQ(h)$ and $\SAMP(h)$ of some function $h: \zo^n \to \zo$.\\ %
\textbf{Output: }Reject or oracle simulators $(\SAMP^*, \MQ^*)$ for some function $\fJ: \zo^\sR \to \zo$. \\
\begin{tikzpicture}
\draw [thick,dash dot] (0,1) -- (16.5,1);
\end{tikzpicture}
    \begin{algorithmic}[1]

\Algphase{Phase $1$: Preprocessing}
\State Run $\FRB(\MQ(h), \SAMP(h),\sR)$.
\State Reject if it rejects; otherwise let $\bL$ be the list of local oracles returned. %
\State If $|\bL|=0$ return $(\MQ(\mathbf{1}_\sR), \SAMP(\mathbf{1}_\sR))$. \Comment{ $\mathbf{1}_\sR$ is the all $1$ function from $\zo^\sR$ to $\zo$} \vspace{0.1cm}%
\Algphase{Phase $2$: Finding a good approximator}
\State Run $\TrimCan( \SAMP(h), \bL,\sR) $.
\State  Reject if it rejects; otherwise let $\fJ$ be the function it returned (as a truth table).\vspace{0.1cm}
\Algphase{Phase $3$: Build oracles for the approximator}
\State Run $\CheckLit(\bL)$; reject if it rejects. 
\State Let $\SAMP^* \leftarrow \SimSAMPA(\sR,\bL,\fJ)$ and $\MQ^* \leftarrow \SimMQA(\cdot,\sR,\bL,\fJ)$
\State Return $(\SAMP^*, \MQ^*)$.
\vspace{0.1cm}  

\end{algorithmic}
\end{algorithm}
\caption[\DNFLearner: Given a factored DNF $f:=H \land t$, the algorithm returns $\MQ$ and $\SAMP$ oracles for a function $g$ which is relative-error close to $t$]{$\DNFLearner.$ The algorithm takes as input oracle access to $h$ and a set $\sR \subseteq [n]$. The goal of the algorithm is to return simulators for the $\MQ$ and $\SAMP$ oracles of a function $\fJ:\zo^\sR \to \zo$, where $\fJ$ is close to some $r$-term, $\mu$-junta DNF $\fD$. In the yes case, recall that $\fh$ denotes the dominating function over $\sR$, with high probability we have that $\fh$ is very close to $\fJ$ in relative-distance. }
\end{figure}

\subsection{Notation and Setup}

We start by defining a useful notion of a list of variable oracles:
 
\begin{definition} \label{def:listoforacles}
We say an ordered tuple $L=((\MQ(g^1),\sX_1),\ldots,(\MQ(g^\ell),\sX_\ell))$, for some $\ell\ge 0$, is a \emph{domain-disjoint list of oracles} if each $\MQ(g^i)$ is the membership oracle of some function $g^i:\{0,1\}^{\sX_i}\rightarrow \{0,1\}$, $\sX_i\subseteq \sR$, and the $\sX_i$'s are pairwise disjoint.

If in addition, 
  every function $g^{i}$ is $(1/30)$-close to a literal $x_{\tau(i)}$ or  $\overline{x_{\tau(i)}}$ for some $\tau(i)\in \sX_i$ \emph{under the uniform distribution}, we say $L$ is a \emph{list of variable oracles}. 
Furthermore, if every $g^{i}$ is exactly a literal, we say $L$ is a \emph{perfect list of variable oracles}.
(Note that an empty list  $L$ is trivially a perfect list of variable oracles.)

Given a list $L$ of variable oracles (so that the $\tau(i)$'s are uniquely defined), we denote by $\sigma_{L}$ the injective map  $\sigma_L:[\ell]\to \sR$ defined by $\sigma_L(i)=\tau(i)$ for each $i\in [\ell]$.
\end{definition}

Looking ahead, when building a list $L$ of variable oracles in Phase 1,
  each function $g^i$ is chosen to be
$h{\upharpoonleft_{v}}$ for some string $v\in \{0,1\}^{\overline{\sX_i}}$ found by the algorithm along the way and thus, $\MQ(g^i)$ can be easily simulated exactly, query by query, using $\MQ(h)$.

Recall that in the yes case, the dominating function $\fh$ that we aim to learn is always an $r$-term DNF over at most $\mu$ variables in $\sR$ (so a $\mu$-junta).
Motivated by the fact that the algorithm won't be able to guarantee to find all the relevant variables (as some of them may have low influence), we introduce 
a collection $\Approximator(\ell)$  of functions over $\zo^{\ell}$ for each $\ell\in [\mu]$.
Roughly speaking, functions in $\Approximator(\ell)$ approximate $r$-term, $\mu$-junta DNFs within relative error $\kappa/2$. 

Formally, given any $r$-term DNF $\fD : \{0,1\}^\mu \to \{0,1\}$, let ${\fD'}: \zo^{\mu} \to \zo$ be defined as: $${\fD'}(x)= \text{arg}\max_{b \in \{0,1\}}\left\{ \Prx_{\bw \sim \{0,1\}^{\mu-\ell}} \left[\fD\left(x_{[\ell]} \circ \bw\right)=b\right]\right\}.$$
Clearly $\fD'$ only depends on the first $\ell$
  variables. Let $\fJ_\fD:\{0,1\}^{\ell} \rightarrow \{0,1\}$ be the function 
  with $$\fJ_{\fD}(z)=\fD'\left(z\circ 0^{r-\ell}\right).$$
Then the function $\fJ_{\fD}$ is added to $\Approximator(\ell)$ iff $\rd(\fD,\fD')\le \kappa/2$, i.e.~$\Approximator(\ell)$ contains all the functions $\fJ_{\fD}$ as $\fD$ ranges over all $r$-term DNFs over $\zo^\mu$  for which $\rd(\fD,\fD') \leq \kappa/2$.

Note that an algorithm can build the sets
  $\Approximator(\ell)$ without making any queries.

Looking ahead, the pair $(\MQ^*,\SAMP^*)$ that $\DNFLearner$ outputs at the end will always~be oracle simulators of $\fJ_\sigma$\footnote{See the end of \Cref{sec:prelim-string-function} for a reminder of the meaning of the notation $\fJ_\sigma.$}  for some $\ell\in [\mu]$, $\fJ\in \Approximator(\ell)$ and injective map $\sigma:[\ell]\rightarrow \sR$ 
(except the trivial case on line 3 of $\DNFLearner$).
We record two lemmas about $\Approximator(\ell)$:

\begin{lemma}\label{lem:Approx-Size}
For any $\ell\in [\mu]$, we have $|\Approximator(\ell)| \leq 2^{2\mu r}$.
\end{lemma}
\begin{proof}
By definition $|\Approximator( \ell)|$ 
can be bounded from above by the number of $r$-term DNFs $\fD$ over $\{0,1\}^\mu$. 
The latter can be bounded by $(2^{2\mu})^r$.
\end{proof}

\begin{lemma}\label{lem: dnf approx remapping}
    Given any function $\fJ \in \Approximator(\ell)$ and any injective map $\sigma:[\ell] \to \sR$, we have $\rd( \fJ_\sigma, \fD) \leq \kappa$ for some $r$ term, $\mu$-junta DNF $\fD:\{0,1\}^\sR \to \zo$.  
\end{lemma}
\begin{proof}
    By the definition of $\Approximator(\ell)$, we have $\fJ=\fJ_\fD$ for some $r$-term DNF $\fD:\zo^{\mu} \to \zo$. 
  Consider any injective map $\pi : [\mu] \to \sR$ 
  with $\pi(i)=\sigma(i)$ for each $i\in [\ell]$.
 Then, we have %
 $$\rd\left( \fD_{\pi}, \fJ_\sigma\right )=\rd(\fD, \fD') \leq \kappa/2.$$ 
    By symmetry (see \Cref{lem:approx-symetric}), we have $\rd(\fJ_\sigma, \fD_{\pi})\leq \kappa$. And obviously $\fD_{\pi}$ is a $r$-term, $\mu$-junta DNF over $\zo^\sR$. 
\end{proof}

Finally, given a set $\sR$ we write $\mathbf{1}_\sR$ to denote the all $1$ function from $\zo^\sR \to \zo$. Note that an algorithm can trivially build $\MQ(\mathbf{1}_\sR)$ (always return $1$) and $\SAMP(\mathbf{1}_\sR)$ (return $\bz \sim \zo^{\sR}$).

\subsection{Phase 1: $\FRB$}
\label{sec:Approximate}
Consider the yes case. 
$\FRB$ takes as input $\MQ(h)$, $\SAMP(h)$ as well as  a set $\emptyset \neq \sR \subseteq [n]$ which we will assume to be stable. 
Let $\fh:\{0,1\}^\sR\rightarrow \{0,1\}$ be the 
  dominating function and let $\sS'$ be the set of 
  relevant variables of $h$ in $\sR$.
Given that $\var(H)\subseteq {\overline{\sR}}$ (as $\sR$ is stable), we have $|\sS'|\le \mu$.
We also have that all relevant variables of $\fh$ are in $\sS'$.

In the yes case, the goal of $\FRB$ is to find as many as possible of the at most $\mu$ relevant variables of   $\fh$. Note that we cannot afford to explicitly find any of these variables using~$O_{r,\mu, \epsilon}(1)$ queries. Instead, $\FRB$ starts by randomly partitioning $\sR$ into $\tau=O(\mu^2)$ disjoint blocks $\bX_1,\ldots,\bX_\tau$. By the birthday paradox, {at the cost of a small constant failure probability} we may assume that $|\sS'\cap \bX_i|\le 1$ for all $i\in [\tau]$.
$\FRB$ then repeatedly performs a binary search procedure over blocks to find (i) a new ``relevant block,'' i.e., one of those $\bX_i$ with $|\sS\cap \bX_i|=1$ and (ii) for each new ``relevant block'' $\bX_i$ found, a restriction $g^i:\{0,1\}^{\bX_i}\rightarrow \{0,1\}$ of the function $h$ that is exactly a literal $x_{\tau(i)}$ or $\overline{x_{\tau(i)}}$ for the unique 
  $x_{\tau(i)}$ in $\sS'\cap \bX_i$.
The pair $(\MQ(g^i),\bX_i)$ is then added to a list to obtain a \emph{perfect} list of variable oracles that $\FRB$ returns at the end.

As mentioned above a key subroutine our algorithm uses is binary search over ``blocks'' of variables to find new relevant blocks. In more detail, let $\bX_{1}, \ldots \bX_{\tau}$ be a partition of $\sR$ and let us fix some $I \subseteq [\tau]$ with $\bX=\cup_{\ell \in I} \bX_\ell$. Given $u \in \zo^n$ and $w \in \zo^{\sR \setminus \bX}$ such that $h(u) \neq h(z_\sU \circ z_\bX \circ w)$, it is straightforward, using binary search, to identify an $\bX_\ell, \ell \not \in I$, and a string $v \in \{0,1\}^n$ with $$h(v) \neq h\left(w_{\bX_\ell} \circ v_{\overline{\bX_\ell}}\right),$$ using $O(\log \tau)=O(\log \mu)$ queries to $h$.
The (simple and standard) idea is as follows: Let $\textsf{U}=[\tau] \setminus I$ be indices of blocks outside $\bX$, and let $a=z, b=z_\sU \circ z_\bX \circ w$. We repeat the following until $|\textsf{U}|=1$:
\begin{flushleft}\begin{quote}
Let $\textsf{U}'$ be the first half of $\textsf{U}$. 
Query $c$, which is $a$ with blocks in $\textsf{U}'$ set according to $b$. If $f(a) \neq f(c)$, we set $\textsf{U}=\textsf{U}'$ and $b=c$; otherwise we set $\textsf{U}=\textsf{U} \setminus \textsf{U}'$ and $a=c$. 
\end{quote}\end{flushleft}

Our goal will be to prove the following two lemmas for $\FRB$:

\begin{lemma}\label{lem:FBR complexity}
     $\FRB(\MQ(h), \SAMP(h),\sR)$ makes  $\poly(\mu,1/\xi)$ calls to $\SAMP(h)$, $\MQ(h)$.
\end{lemma}
\begin{proof}
    We repeat lines 3--9 until either $|\bL| > \mu$~or $t=t^*$.
In one iteration of the loop $\FRB$ draws one sample from $h$ on line~4.  If the ``if'' condition on line {5} is true, then line 6 is executed which makes $O(\log \tau)=O(\log \mu)$ queries, and $|\bL|$ increases by one and $t$ is set to $0$.  
And if the ``if'' condition on line {5} is false, then line {5} only makes one query and $|{\bL}|$
does not increase, but $t$ increases by one.
Therefore, the query complexity of $\FRB$ is 
$$O\pbra{\mu t^* + \mu \log \mu}=\poly(\mu/\xi).\qedhere$$
\end{proof}

\begin{lemma}\label{thm:FRB}
  Assume $h$ is a $(r,\mu)$-factored-DNF and $\sR$ is stable.
With probability at least $99/100$, $\FRB$ returns a perfect list $\bL$ of local oracles with $|\bL|\le \mu$ and the following properties: 
There exists a function $\fg \in \Approximator(|\bL|)$  such that %
$\rd(\fh,\fg_{\sigma_{\bL}})\leq \xi$;
for the special case when $\bL$ is empty, 
  we have $\rd(\fh,\mathbf{1}_\bR) \leq \xi$,
  where $\mathbf{1}_\bR$ is the all-$1$ function over $\{0,1\}^\sR$.
\end{lemma}

\begin{figure}[t!]
\begin{algorithm}[H]
\caption{\protect\hypertarget{Approximate}{\FRB}} \label{algo:FRB}
\vspace{0.15cm}\textbf{Input: }$\MQ(h)$ and $\SAMP(h)$ of some function $h: \zo^n \to \zo$ and $ \emptyset \neq \sR \subseteq[n]$.\\
\smash{\textbf{Output: }Either Reject or a list $\bL$ {of domain-disjoint oracles}.}

\begin{tikzpicture}
\draw [thick,dash dot] (0,1) -- (15.9,1);
\end{tikzpicture}
\begin{algorithmic}[1]
\Algphase{Phase~1: Partition $\sR$ into blocks}
    \State Let $\tau=\theta(\mu^2)$; randomly partition $\sR$ into blocks $\bX_1, \ldots, \bX_{\tau}$.\vspace{0.1cm} %
\Algphase{Phase~2: Find relevant blocks}
\State Set {$\bX=\emptyset$}; $\bL$ to be an empty list; $t=0$ and $t^*={5}\ln(200\mu)/\xi$.
\While{$t\le  t^*$}

\State  Draw $\smash{\bz \sim \SAMP(h)}$ and $\smash{\bw \sim \{0,1\}^{\sR \setminus \sX}}$, and set $t \leftarrow t+1.$
\If{$\smash{h(\bz_{\overline{\sR}\cup \bX} \circ \bw)= 0}$} \red{\Comment{Type~1 Query}} %
\State Binary Search to find a new relevant block $\bX_i$ and a string $\bv \in \{0,1\}^n$ such that $$1=h(\bv) \neq h\left(\bw_{\bX_i} \circ \bv_{\overline{\bX}_i}\right)=0.$$
\Statex \red{\Comment{Binary Search only uses Type~1 Queries}}
\Statex \hskip3em For notational convenience, we will assume without loss of generality that the 
$i$th new 
\Statex \hskip3em block found is indexed by $\bX_i$. Denote by $g^i : \{0,1\}^{\bX_i} \rightarrow \{0,1\}$ the following function $$g^i(z)=h\left(z \circ \bv_{\overline{\bX}_i}\right).$$ \red{\Comment{Every query to $g^i$ is a Type 1 query}}
\State Set $t\leftarrow 0$, $\bX \leftarrow \bX \sqcup \bX_i$ and add $(\MQ(g^i),\bX_i))$ to the list $\bL$.
\State If $|\bL|>\mu$ then output ``Reject'' and halt.
\EndIf
\EndWhile
\State Output $\bL$. 
\end{algorithmic} 
\end{algorithm}
\caption[$\FRB$]{$\FRB.$ The algorithm takes as input oracle access to $h$ and a set $\sR \subseteq [n]$. In the yes case, recall that $\fh$ is the dominating function over $\sR$ and it depends on at most $\mu$ variables. We cannot find these variables explicitly using $O_{r,\mu, \epsilon}(1)$ queries. Hence, the goal of the algorithm is to produce a perfect list of variable oracles $\bL$ which contains ``the most relevant'' variables $\fh$ depends on. In particular, when the algorithm returns, with high probability there is some function $\fg \in \Approximator(|\bL|)$ such that $\fh$ is close to $\fg$ under the remapping of the variables given by $\sigma_{\bL}$ (see \Cref{def:listoforacles} for the definition of $\sigma_{\bL}$).}
\end{figure}
We prove \Cref{thm:FRB} in the rest of this subsection. We always assume that $h$ is an $(r,\mu)$-factored-DNF and $\sR$ is stable.
We start with a simple lemma about the random partition:

\begin{lemma}\label{lem: FRB 1.1}
 We have $|\bX_i\cap \sS'|\le 1$ for all $i\in [\tau]$ with probability at least $199/200$. 
\end{lemma}
\begin{proof}
    Fix any $j_1,j_2\in \sS'$. The probability that they lie in the same block is $1/\tau$ if they are both in $\sR$ and $0$ otherwise. By a union bound and the assumption that $|\sS'|\le \mu$, the probability of $|\bX_i\cap \sS'|>1$ for some $i$ is at most $${\mu \choose 2} \cdot \frac{1}{\tau} \leq 1/200. \qedhere$$ 
\end{proof}

From now on we fix the partition $\sX_1, \ldots , \sX_{\tau}$ of $\sR$ and assume that $|\sX_i \cap \sS'| \leq 1$ for every $i \in [\tau]$.
We write $\bL=((\MQ(g^1),\sX_1),\ldots,(\MQ(g^{\ell}),\sX_{\ell}))$ with $\ell=|\bL|$ to denote the list returned.

Next we prove a lemma about $\bL$ and $\bX$ when $\FRB$ terminates:

\begin{lemma}\label{lem: FRB 1.1.1}
Assume that $|\sX_i\cap \sS'|\le 1$ for all $i\in [\tau]$. Then $\FRB$ never rejects and always returns a perfect list $\bL$ of variable oracles with $|\bL|\le \mu$. %
Moreover, with probability at least $199/200$, 
when $\FRB$ terminates we have
    \begin{equation}\label{eq: xi/10}
\Prx_{\substack{\bz \sim h^{-1}(1)\\ \bw \sim\{0,1\}^{\sR \setminus \bX}}}\Big[h(\bz_{\sU\cup \bX} \circ \bw)=0\Big] \le \frac{\xi}{5}.
    \end{equation}
\end{lemma}
\begin{proof}
Clearly the $\sX_i$'s are disjoint in $\bL$.
    For each pair $(\MQ(g^i),\sX_i)$  added to $\bL$, $\FRB$~has found a restriction $g^i$ of $h$  such that $g^i:\{0,1\}^{\sX_i} \to \zo$ is not a constant function. Since within $\sR$, $h$ only depends on variables in $\sS'$ and $|\sS'\cap \sX_i|\le 1$,
    we must have $|\sS'\cap \sX_i|=1$ and 
$g^i$ is a literal of that variable in $\sS'\cap \sX_i$. As a result, $\bL$ is a perfect list of variable oracles.
Given that $|\sS'|\le \mu$, 
$\bL$ can never be
   larger than $\mu$ and thus, $\FRB$ never rejects. 

To violate \Cref{eq: xi/10} when
  returning, %
 it must be the case that $\FRB$ has drawn $\smash{\bz\sim h^{-1}(1)}$ and $\smash{\bw\sim \{0,1\}^{\sR \setminus \bX}}$ for $t^*$ many times, and  $h(\bz_{\sU\cup \sX} \circ \bw)=1$ every time.
However,~given the violation of \Cref{eq: xi/10}, the probability that this happens is at most 

    $$
    \big(1-\xi/5\big)^{t^*}\le e^{-\ln(200\mu)}\le 1/200\mu.
    $$

By a union bound (over the at most $\mu$ 
  different $\bX$'s that are in play across all executions of line {5})  
  the probability of  $\FRB$
  returning while violating \Cref{eq: xi/10} is at most $1/200$. 
\end{proof}

The rest of the subsection is all about connecting 
\Cref{eq: xi/10} with our goal: There exists a function $\fg \in \Approximator(|\bL|)$ such that %
$\rd(\fh,\fg_{\sigma_{\bL}})\leq \xi$
  (and the special, easier~case when $\bL$ is empty).
To this end, given that $\sR$ is stable and  $\fh:\{0,1\}^{\sR}\rightarrow \{0,1\}$ is the dominating function,
  we first use \Cref{eq: xi/10} to obtain the following inequality about $\fh$.

\begin{lemma}\label{lem:fhtoh}
\Cref{eq: xi/10} implies that
    \begin{equation}\label{eq:bla11}
\Prx_{\substack{\by \sim \fh^{-1}(1) \\ \bw \sim \{0,1\}^{\sR \setminus \bX}}}\Big[\fh(\by_{\bX} \circ \bw)=0\Big] \le \xi/4,
\end{equation}
\end{lemma}
\begin{proof}
Using that $\sR$ is stable, we have
\begin{align*}
    \Prx_{\substack{\bz \sim h^{-1}(1),\\ \bw \sim \zo^{\sR \setminus \bX}}}\Big[h(\bz_{\sU\cup \bX} \circ \bw)=0\Big]
    &\geq \Prx_{\bz,\bw}\Big[h(\bz_{\sU\cup \bX} \circ \bw)=0\hspace{0.06cm} \Big| \hspace{0.06cm}h{\upharpoonleft_{\bz_\sU}} \equiv \fh\Big] \times \Prx_{\bz \sim h^{-1}(1)}\big[ h{\upharpoonleft_{\bz_\sU}} \equiv \fh\big]\\
    &\geq (1-\xi)\cdot \Prx_{\bz,\bw}\Big[\fh(\bz_\bX \circ \bw)=0 \hspace{0.06cm}\Big| \hspace{0.06cm}h{\upharpoonleft_{\bz_\sU}} = \fh\Big] .
\end{align*}
But conditioned on $h{\upharpoonleft_{\bz_\sU}} \equiv \fh$, we have that $\bz_{\sR} \sim \fh^{-1}(1)$. Hence, the above becomes:
$$
\xi/5\ge \Prx_{\substack{\bz \sim h^{-1}(1),\\ \bw \sim \zo^{\sR \setminus \bX}}}\Big[h(\bz_{\sU\cup \bX} \circ \bw)=0\Big]\geq (1-\xi)\cdot \Prx_{\substack{\by \sim \fh^{-1}(1)\\ \bw \sim \zo^{\sR \setminus \bX}}}\Big[\fh(\by \circ \bw)=0\Big],$$
from which \Cref{lem:fhtoh} follows.
\end{proof}

 We note that the special case when $\bL$ is empty will follow from 
  \Cref{lem:fhtoh} so we assume that 
the list $\bL=((\MQ(g^1),\sX_1),\ldots,(\MQ(g^\ell),\sX_\ell))$ is not empty and $\ell=|\bL|>0$.
For each $i\in [\ell]$, let $\tau(i)$ be the unique variable in $\sX_i\cap \sS'$ such that $g^i$ is  $x_{\tau(i)}$ or $\overline{x_{\tau(i)}}$ and 
let $\sigma=\sigma_{\bL}$ be the injective map 
  with $\sigma(i)=\tau(i)$ for each $i\in [\ell]$.
Let $$\Sigma:=\big\{\sigma(i) :i\in [\ell]\big\}
\quad\text{with}\quad
\Sigma=\sS'\cap \bX$$ given that $\bX=\sX_1\cup \cdots\cup \sX_\ell$. Let $\fg:\{0,1\}^{\ell}\rightarrow \{0,1\}$ be the following $\ell$-variable function:
\begin{equation}\label{eq:defg}
\fg(z):= \text{arg}\max_{b \in \{0,1\}}\left\{ \Prx_{\bw \sim \{0,1\}^{\sR \setminus \Sigma}} \Big[\fh\left(\sigma(z) \circ \bw)\right)=b\Big]\right\}.
\end{equation} 

\begin{lemma}\label[lemma]{lem:apple}
\Cref{eq:bla11} implies that $\rd(\fh,\fg_{\sigma})\le \xi$.
\end{lemma} 
\begin{proof}
Observe that we have:
$$
\fg_{\sigma}(z)= \text{arg}\max_{b \in \{0,1\}}\left\{ \Prx_{\bw \sim \{0,1\}^{\sR \setminus \Sigma}} \Big[\fh\left(z_{\Sigma} \circ \bw\right)=b\Big]\right\}.
$$

Assume for a contradiction that $\rd(\fh,\fg_\sigma)>\xi$. 
Let $F=\fh^{-1}(1)$ and $G=\fg_{\sigma}^{-1}(1)$. 
Given that $|F \ \triangle \  G|\ge \xi |F|$, we
  consider the two cases of $|F\setminus G|>\xi |F|/2$
    or $|G\setminus F|>\xi |F|/2.$
To this end, we consider the following equivalent way of drawing $\by_\bX\circ \bw$ with $\smash{\by\sim \fh^{-1}(1)}$ and $\smash{\bw\sim \{0,1\}^{\sR \setminus \bX}}$: draw $\by\sim \fh^{-1}(1)$ and $\smash{\bw\sim \{0,1\}^{\sR \setminus \Sigma}}$ and return $\by_\Sigma\circ \bw$.
The two distributions are the same because 
  $\Sigma=\bX\cap \sS'$ are the only relevant 
  variables in $\bX$ for $\fh$.%
 \begin{equation}\label{eq:temp404}
  \Prx_{\substack{\by \sim \fh^{-1}(1) \\ \bw \sim \{0,1\}^{\sR \setminus \bX}}}
  \Big[\fh\left(\by_{\sX} \circ \bw\right)=0\Big] =      \Prx_{\substack{\by \sim \fh^{-1}(1) \\ \bw \sim \{0,1\}^{\sR \setminus \Sigma}}}\Big[\fh\left(\by_{\Sigma} \circ \bw\right)=0\Big]. 
    \end{equation}

For the case when $|F\setminus G|>\xi |F|/2$, we have
$$\Prx_{\substack{\by \sim \fh^{-1}(1)\\ \bw \sim \zo^{\sR \setminus \Sigma} }}\Big[\fh\left(\by_{\Sigma} \circ \bw\right)=0\Big] \geq 
\Prx_{\by \sim \fh^{-1}(1) } \big[\by \in F\setminus G\big] \cdot \Prx_{\substack{\by \sim \fh^{-1}(1)\\ \bw \sim \zo^{\sR \setminus \Sigma}} }\Big[\fh\left(\by_{\Sigma} \circ \bw\right)=0 \hspace{0.08cm}\big|\hspace{0.08cm} \by \in F\setminus G \Big].$$
The first probability on the RHS is at least $\xi/2$;
  the second probability on the RHS is at least $1/2$ given the definition of
  $\fg_{\sigma}$ from $\fh$ and the condition that $\by\in F\setminus G$ and thus, $\fg_{\sigma}(\by)=0$.
This leads to a contradiction with 
  \Cref{eq:bla11} when combined with \Cref{eq:temp404}.

For the second case when $|G \setminus F| \geq \xi |F|/2$, consider any fixed $a \in G \setminus F$. Since $a \in G$, we have 
$$
\Prx_{\bw \sim \zo^{\sR \setminus \Sigma}}\Big[f\left(a_{\Sigma} \circ \bw\right)=1\Big] \geq 1/2
$$ and thus there must exist at least $2^{|\sR \setminus \Sigma|}/2$ many $y \in F$ with $y_{\Sigma}=a_\Sigma$. So 
$$\Prx_{\substack{\by \sim \fh^{-1}(1)\\ \bw \sim \{0,1\}^{\sR \setminus \Sigma}}}\big[\by_{\Sigma} \circ \bw = a\big] \geq \Prx_{\by \sim \fh^{-1}(1)}\big[\by_{\Sigma} = a_{\Sigma}\big] \cdot \Prx_{\bw \sim \zo^{\sR \setminus \Sigma}}\big[\bw=a_{\sR \setminus \Sigma}\big] \geq \frac{{}2^{|\sR \setminus \Sigma|}}{2|F|}\times 2^{-|\sR \setminus \Sigma|} =\frac{1}{2|F|}.$$
As a result, we have $$\Prx_{\substack{\by \sim \fh^{-1}(1)\\ \bw \sim \{0,1\}^{\sR \setminus \Sigma}}}\Big[\fh\left(\by_{\Sigma} \circ \bw\right)=0\Big] \geq \sum_{a \in G \setminus F} \left(\Prx_{\substack{\by \sim \fh^{-1}(1)\\ \bw \sim \{0,1\}^{\sR \setminus \Sigma}}}\big[\by_{\Sigma} \circ \bw=a\big]\right) \geq \frac{{\xi |F|}}{2} \cdot \frac{1}{2|F|} = \frac{\xi}{4}.$$
This finishes the proof of the lemma.
\end{proof}
\begin{lemma}\label{lem:meansinApprox}
    If $\rd(\fh, \fg_{\sigma}) \leq \xi$ then $\fg \in \Approximator(|\bL|)$.
\end{lemma}
\begin{proof}
This follows directly from how $\Approximator(|\bL|)$ is defined. Given $\fh$, which is a $r$-term, $\mu$-junta DNF, we can permute its relevant variables to obtain $\fD$ as in the definition of $\Approximator$ so that $\fg$ is the $\fD'$ there. 
The lemma then follows from our choices of $\kappa$ and $\xi$ that $\xi\ll \kappa$.
\end{proof}

We finally combine all the lemmas to prove \Cref{thm:FRB}. 
\begin{proof}[Proof of \Cref{thm:FRB}]
    By a union bound over \Cref{lem: FRB 1.1} and  \Cref{lem: FRB 1.1.1} we have that $\FRB$ doesn't reject, and returns a perfect list $\bL$ of variable oracles with $|\bL|\le \mu$ such that
    \Cref{eq: xi/10} holds.
    The special case when $\bL$ is empty follows from 
    \Cref{lem:fhtoh}, $\bL$ being empty implies $\bX=\emptyset$, so we have: 
    $$\rd(\fh,\mathbf{1}_\sR)=\Prx_{\bw \sim \zo^{\sR}}[\fh(\bw)=0] \leq \xi/4.$$
    The general case when $\bL$ is not empty follows from \Cref{lem:fhtoh,lem:apple,lem:meansinApprox}.  
\end{proof}

\subsection{Learner phase 2: $\TrimCan$}

We now turn our attention to Phase~2. 
Before we can present $\TrimCan$, we present a simple helper subroutine $\Extract$. 
(This subroutine will play a crucial role in the analysis of the no case later in Phase 3, where $L$ may only be a list of variable oracles but not  perfect. So below we analyze its performance in both cases, even though we may assume $L$ to be perfect in the yes case.) 

\begin{figure}[t!]
\begin{algorithm}[H]
\caption{\protect\hypertarget{Extract}{$\Extract$}}
\vspace{0.15cm}
\textbf{Input: }A domain-disjoint list of  oracles $L=\smash{((\MQ(g^{1}),\sX_1), \ldots, (\MQ(g^{\ell}),\sX_{\ell}))}$, $w \in \zo^{\sR}$ and $\delta$. \\
 \textbf{Output:} A string  $\bz\in \zo^{\ell}$.\\
  \begin{tikzpicture}
\draw [thick,dash dot] (0,1) -- (16.5,1);
\end{tikzpicture}
\begin{algorithmic}[1]
\For{$i$ from $1$ to $\ell$}
\State Let $Y_{b}=\{j \in \sX_i: w_j = b\}$ for both $b \in \zo$ and set counters  $\bG_{0}=
\bG_{1}=0$.
\RepeatN{$O(\log(\mu/\delta))$}
 \State Draw $\ba^b \sim \{0,1\}^{Y_{b}}$ for both $b \in \zo$.
 \State If {$\smash{{g^i(\ba^0 \circ \ba^1 ) \neq g^i(\overline{\ba^0} \circ \ba^1)}}$}, 
$\bG_{0} \leftarrow \bG_{0}+1$; otherwise, $\bG_{1} \leftarrow \bG_{1}+1$.
\red{\Comment{Type~1 Query}} \End
  \State If $\bG_0>\bG_1$ set $\bz_i=0$; otherwise, set $\bz_i=1$.
\EndFor
\State Return $\bz$.
\end{algorithmic} 
\end{algorithm}
\caption[$\Extract$]{$\Extract$. The algorithm takes as input a domain-disjoint list of oracles $L$, a string $w \in \zo^{\sR}$, and a failure probability $\delta$. 
Its goal is to return  $\sigma_{L}^{-1}(w)\in \{0,1\}^\ell$ when $L$ is a list of variable oracles. (See \Cref{def:listoforacles} to recall the definition of $\sigma_L$). }
\end{figure}

{$\Extract$ takes as input a domain-disjoint list of oracles $L=((\MQ(g^1),\sX_1),\ldots,(\MQ(g^\ell),\sX_\ell))$ with $\ell\in [\mu]$, a string $w \in \zo^{\sR}$, and a failure probability $\delta$. 
Its goal is to return  $\sigma_{L}^{-1}(w)\in \{0,1\}^\ell$ when $L$ is a list of variable oracles.
(See \Cref{def:listoforacles} to recall the definition of $\sigma_L$, and recall that $\sigma_L$ is well defined when $L$ is a list of variable oracles.)}%

\begin{lemma}\label[lemma]{lem: Extract}
$\Extract(L,w, \delta)$ makes $O(\mu\log(\mu/\delta))$ queries on oracles in $L$ and satisfies: 
\begin{flushleft}\begin{enumerate}
  \item When  $L$ is a list of variable oracles, 
  $\Extract$ returns $\sigma_{L}^{-1}(w)$ with probability at least $1-\delta$;
\item When $L$ is perfect, $\Extract$ always returns $\sigma_{L}^{-1}(w)$.
\end{enumerate}\end{flushleft}
\end{lemma}
\begin{proof}
The query complexity of $\Extract$ follows immediately by inspection and $\ell\le \mu$. 

When $L$ is perfect, every $g^i$ is a literal $x_{\tau(i)}$ or $\overline{x_{\tau(i)}}$ with $\sigma_L(i)=\tau(i)$. During each round~of lines~3--6, its $\smash{\bG_{w_{\tau(i)}}}$ that always gets incremented. Hence we always have $\bz_i=w_{\tau(i)}$, and  $\smash{\bz=\sigma_{L}^{-1}(w)}$. 

\def\aa{\ba}

Now, consider the case when $L$ is only promised to be a list of variable oracles. Fix $g^i$ which is $(1/30)$-close
  to either $x_{\tau(i)}$ or $\overline{x_{\tau(i)}}$ with $w_{\tau(i)}=0$ (meaning $\tau(i) \in Y_0$;
  the case when $w_{\tau(i)}=1$ is symmetric).
Then
$$
\Prx\big[g^i(\aa^0\circ\aa^1)= g^i(\overline{\aa^0}\circ\aa^1)\big]
\le \Prx\big[g^i(\aa^0\circ\aa^1)\ne \aa^0_i\big]
+\Prx\big[g^i(\overline{\aa^0}\circ\aa^1)\ne \overline{\aa^0_i}\big]\le 1/15,
$$
where the last inequality follows from the fact that 
  both the marginal distributions of $\aa^0\circ \aa^1$ and $\smash{\overline{\aa^0}\circ\aa^1}$ are uniform.
As a result, in each round, $\bG_0$ goes up by
  $1$ with probability at least $14/15$ and $\bG_1$ goes
  up by $1$ with probability at most $1/15$.
By a standard Chernoff bound in this case and by making the hidden constant large enough, we have that $\bz_i=0$ with probability at least $1-\delta/\mu$.
The lemma then follows from a union bound.
\end{proof}

\begin{figure}[t!]
\begin{algorithm}[H]
\caption{\protect\hypertarget{FindCandidate}{\TrimCan}} \label{alg:TrimCan}
\vspace{0.15cm}
\textbf{Input: }\smash{$\SAMP(h)$, a domain-disjoint list of oracles $L$ of length $\ell$ and $\sR\subseteq [n]$.} \\
\textbf{Output: }Either reject, or the truth table of a function $\fJ \in \Approximator(\ell)$. \\ 
\begin{tikzpicture}
\draw [thick,dash dot] (0,1) -- (16.5,1);
\end{tikzpicture}

\begin{algorithmic}[1]
\State Let $\bAA=\Approximator(\ell)$ 
\RepeatN{$\eta= 1/(200\xi)$} %
\State Draw $\bw \sim \SAMP(h)$ and let $\bz= 
\Extract(L,\bw_{\sR},0.5)$.
\State  Remove from $\bAA$ any $\fJ$ for which  $\fJ(\bz)=0$. \Comment{This does not query $\MQ(h)$}
\End
\State If $\bAA=\emptyset$, reject; otherwise,
  return the $\fJ\in \bAA$ with the smallest $|\fJ^{-1}(1)|$. 
\end{algorithmic} 
\end{algorithm}
\caption[$\TrimCan$]{$\TrimCan$. The algorithm takes as input sample access to $h$ a domain-disjoint list  of oracles $L$ and a set $\sR \subseteq [n]$. In the yes case, recall that $\fh$ is the dominating function over $\sR$. The goal of the algorithm is to find a function $\fJ \in \Approximator(|L|)$ such that $\fJ$ is close in relative-distance to $\fh$.}
\end{figure}
$\TrimCan$ takes as input $\SAMP(h)$,  a domain-disjoint list of oracles, and $\sR$. 
Again we focus on the yes case and show that, when $\sR$ is stable and $L$ is a perfect list that satisfies the conditions of \Cref{thm:FRB}, $\TrimCan$ returns a $\fJ\in \Approximator(\ell)$ satisfying $\rd(\fh, \fJ_{\sigma_L})\le \kappa$.

Formally, our goal is to prove the following:
\begin{lemma}\label{thm: a good approx is found}
Assume that $h$ is an $(r,\mu)$-factored DNF, $\sR$ is stable, $L$ is a perfect list of variable oracles with $\ell=|L|>0$, and there exists $\fg\in \Approximator(\ell)$ such that $\rd(\fh,\fg_{\sigma_L})\le \xi$. 
With probability at least $98/100$, $\TrimCan\hspace{0.02cm}(\SAMP(h),L,\sR)$ returns a function $\fJ\in \Approximator(\ell)$~that satisfies $\rd(\fh, \fJ_{\sigma_{L}}) \leq \kappa$.
\end{lemma} 
$\TrimCan\hspace{0.02cm}(\SAMP(h),L,\sR)$ starts by setting $\bAA$ to the set $\Approximator(\ell)$ of all candidate functions. It then draws samples from $\SAMP(h)$, remaps them using $\Extract$ and discards every candidate function  for which the remapped sample is not a satisfying assignment. 
$\TrimCan$ rejects if $\bAA$ becomes empty; otherwise, crucially, it returns (the truth table of) the candidate left in $\bAA$ that has the fewest satisfying assignments. 

We break the proof into two lemmas. First we show 
that most likely, $\fg$ survives in $\bAA$ at the end so that $\TrimCan$ does not reject:
\begin{lemma}\label{lem: perm_in_S}
Under the assumptions of \Cref{thm: a good approx is found}, $\fg$ remains in $\bAA$ with probability at least $99/100$.
\end{lemma}
\begin{proof}
Given that $\rd(\fh,\fg_{\sigma_L})\le \xi$, we have
$$
\Prx_{\by\sim \fh^{-1}(1)}\Big[\fg\left(\sigma_L^{-1}(\by)\right)=0\Big]=
\Prx_{\by\sim \fh^{-1}(1)}\Big[\fg_{\sigma_L}(\by)=0\Big]\le \xi.
$$
Given that $\sR$ is stable, we have 
$$
\Prx_{\bx\sim h^{-1}(1)}\Big[\fg\left(\sigma_L^{-1}(\bx_{\sR})\right)
=0\Big]\le\xi + (1-\xi)\cdot \Prx_{\by\sim \fh^{-1}(1)}\Big[\fg\left(\sigma_L^{-1}(\by)\right)=0\Big]\le 2\xi.
$$
Given that $\Extract$ always works correctly when $L$ is perfect,  
each iteration of the loop removes $\fg$ from $\bAA$ with probability at most $2\xi$.
    Hence, by a union bound over all $\eta$ loops, the probability that $\fg$ is removed is at most $2\eta\xi \leq 1/100$.
\end{proof}

We then show that most likely, every function 
  $\fJ$ that survives in $\bAA$ satisfies the following condition:

\begin{lemma}\label{lem: no_small_in_S}
Under the assumptions of \Cref{thm: a good approx is found}, with probability at least $99/100$, every function $\fJ$ that survives in $\bAA$  satisfies
    $$\Prx_{\by \sim \fh^{-1}(1)}\Big[\fJ_{\sigma_{L}}(\by)=0\Big] \le \frac{\kappa}{10}.$$
\end{lemma}
\begin{proof}
Take any $\fJ\in \Approximator(\ell)$ that violates
  the inequality above. 
Given that $\sR$ is stable, %
$$
\Prx_{\bx\sim h^{-1}(1)}\Big[\fJ\left(\sigma_L^{-1}(\bx_\sR)\right)=0\Big]\ge (1-\xi)\cdot 
\Prx_{\by\sim \fh^{-1}(1)} \Big[\fJ_{\sigma_{L}}(\by)=0\Big]\ge \frac{\kappa}{20}.
$$
Given that $\Extract$ always works correctly, 
each iteration of the loop removes $\fJ$ with probability at least $\kappa/20$. Thus,
  $\fJ$ survives with probability at most
$$ \left(1-\frac{\kappa}{20}\right)^{\eta} \leq \frac{1}{100\cdot  2^{2\mu r}},$$
using $\eta=1/(200\xi)$ and our choice of $\xi$. 
The lemma follows from a union bound over at most $2^{2\mu r}$ many  functions $\fJ$ in $\Approximator(\ell)$ (see \Cref{lem:Approx-Size}).
\end{proof}

Combining the above, we can now prove 
\Cref{thm: a good approx is found}: %

\begin{proof}[Proof of \Cref{thm: a good approx is found}]
With probability at least $98/100$, the conclusions of both \Cref{lem: perm_in_S,lem: no_small_in_S} hold. Assume this occurs and let $\fJ$ be the function returned. As $\fg \in \bAA$, we have $|\fJ^{-1}(1)|\leq |\fg^{-1}(1)|$. So  
\begin{equation} \label{eq:likeable-equation}
\left|\fJ_{\sigma_{L}}^{-1}(1)\right|\le \left|\fg_{\sigma_{L}}^{-1}(1)\right|\quad\text{and}\quad
\frac{|\fh^{-1}(1)\setminus \fJ_{\sigma_{L}}^{-1}(1)|}{|\fh^{-1}(1)|}\le \frac{\kappa}{10}.
\end{equation} 
 Given that $\rd(\fh,\fg_{\sigma_{L}}) \leq \xi$, we have $$\left|\fJ_{\sigma_{L}}^{-1}(1)\right|\leq \left|\fg_{\sigma_{L}}^{-1}(1)\right| \leq (1+\xi)|\fh^{-1}(1)|.$$
Combining the above with the first inequality of \Cref{eq:likeable-equation} we have  \begin{equation}\label{eq: another likeable equation}
    \left|\fJ_{\sigma_{L}}^{-1}(1) \setminus \fh^{-1}(1)\right|=\left|\fJ_{\sigma_{L}}^{-1}(1)\right|
    -\left|\fh^{-1}(1)\right|+
    \left|\fh^{-1}(1)\setminus \fJ_{\sigma_{L}}^{-1}(1)\right|
    \le \left(\xi+\frac{\kappa}{10}\right)|\fh^{-1}(1)|.
\end{equation}

Combining the second inequality of \Cref{eq:likeable-equation} and \Cref{eq: another likeable equation}, we get that
$$\rd(\fh,\fJ_{\sigma_{L}}) \leq \kappa/10+ \xi + \kappa/10 \leq \kappa$$
using $\xi\ll \kappa$ from our choices of these two parameters.
This finishes the proof of the lemma.
\end{proof}

Finally, we also need the following simple lemma about the query complexity of $\TrimCan$.
\begin{lemma}\label{lem:TrimCan_complexity}
 $\TrimCan(\SAMP(h), L,\sR)$ makes   $\poly((\mu/\xi)\log \mu)$ calls to 
    $\SAMP(h),\MQ(h)$.
\end{lemma}
\begin{proof}
  $\TrimCan$ draws $\eta=O(1/\xi)$ samples $\SAMP(h)$ and calls $\Extract$ on each of them. By \Cref{lem: Extract} each call to $\Extract$ uses $O(\mu\log(\mu))$ calls to oracles in $L$. Note that
each such call is simulated by a single call to $\MQ(h)$.
\end{proof}

\subsection{Learner phase $3$: Building oracles} \label{sec:learner-phase-3}
We now finally look at Phase 3 of $\DNFLearner$.
Phase 3 first runs a simple subroutine $\CheckLit$ to check that $L$ is a list of variable oracles (i.e., each $g^i$ in $L$ is $(1/30)$-close to a literal).
Note~that the yes case will pass this test trivially since we assume that the list $L$ passed down to this phase is perfect. 
As will become clear soon, 
  running $\CheckLit$ is really a service for the no case because we don't have any guarantees about what  the list $L$ contains in that case.
(Recall that all analysis of $\DNFLearner$ so far is about the yes case; for the no case, all we know is that either $\DNFLearner$ has already rejected or a domain-disjoint list of oracles $L$ has been passed down to Phase 3.)
Assuming that $L$ passes the test, Phase 3 builds a pair of oracles $(\MQ^*,\SAMP^*)$ and returns them to end $\DNFLearner$.
As will become clear, building these two oracles does not cost any queries. 
It suffices to bound the number of queries to $h$ per call to these two oracles, and show that they indeed simulate functions with desired properties in the two cases.

\subsubsection{The $\CheckLit$ subroutine}

Before presenting $\CheckLit$, we recall the following result, due to \cite{Blaisstoc09}, on uniform-distribution junta testing:  

\begin{theorem}[Theorem~1.1 of \cite{Blaisstoc09}] \label{thm: uniform junta blais}\hypertarget{UniformJunta}{}
    There is a one-sided adaptive algorithm, $\UniformJunta$ which $\epsilon$-tests $k$-juntas in the standard model.  $\UniformJunta(\MQ(f),k, \epsilon,\delta)$
makes $O((k/\epsilon+k\log(k)) \log(1/\delta))$ queries, always accepts if $f$ is a $k$-junta, and rejects $f$ with probability at least $1-\delta$ if $f$ is $\epsilon$-far from
every $k$-junta with respect to the uniform distribution.
\end{theorem}

\begin{figure}[t!]
\begin{algorithm}[H]
\caption{\protect\hypertarget{CheckLit}{$\CheckLit$}}
\textbf{Input: }A domain-disjoint list of oracles $L=(\MQ(g^{1}),\sX_1), \ldots, (\MQ(g^{\ell}),\sX_\ell)$ for some $\ell\in [\mu]$. \\
 \textbf{Output:} Accept or Reject.\\
 \begin{tikzpicture}
\draw [thick,dash dot] (0,1) -- (16.5,1);
\end{tikzpicture}
\begin{algorithmic}[1]
\For{each $i\in [\ell]$}
\State Reject if $\UniformJunta(\MQ(g^i), 1, 1/30, 1/100)$ rejects.\red{\Comment{Type~1 Query}}
\State Repeat $10$ times: Draw $\ba \sim \{0,1\}^{\sX_i}$ and reject if 
${{g^{i}}(\ba)= {g^{i}} (\overline{\ba})}$. \red{\Comment{Type 1 Query}}
\EndFor
\State Accept. 
\end{algorithmic} 
\end{algorithm}
\caption[$\CheckLit$]{$\CheckLit$. The algorithm takes as input a domain-disjoint list of oracles $L$, and checks that $L$ is a list of variable oracles (i.e., each $g^i$ in $L$ is $(1/30)$-close to a literal). }
\end{figure}

We prove the following lemma about \CheckLit:
\begin{lemma}\label{thm: checkIfLit}
    $\CheckLit(L)$ uses $O(\mu)$ queries to  oracles in $L$ and satisfies the following conditions:
    \begin{flushleft}\begin{enumerate}
        \item If $L$ is a perfect list of variable oracles, then $\CheckLit$ always accepts. 
        \item If there exists some $g^i$ in $L$ such that ${g^{i}}$ is $(1/30)$-far from every literal under the uniform distribution, then $\CheckLit$ rejects with probability at least $99/100$.
    \end{enumerate}\end{flushleft}
\end{lemma}
\begin{proof}
The query complexity follows from  \Cref{thm: uniform junta blais} given that each call to $\UniformJunta$~uses $O(1)$ queries. 
    If every ${g^{i}}$ in $L$ is a literal, $\UniformJunta$ never rejects because by  \Cref{thm: uniform junta blais} it is one-sided. And we always have ${g^{i}}(a)\ne {g^{i}}(\overline{a})$. Hence $\CheckLit$ always accepts.

    On the other hand, assume that some ${g^{i}}$ in $L$ is $(1/30)$-far from every literal with respect to the uniform distribution. Then it is either \begin{enumerate}
        \item $(1/30)$-far from every $1$-Junta under the uniform distribution; or
        \item $(1/30)$-far from every literal and $(1/30)$-close to a $0$-Junta (i.e., a constant function).
    \end{enumerate}
    In case $1$ by \Cref{thm: uniform junta blais}, $\UniformJunta$ rejects with probability at least $99/100$.
    In case $2$, if ${g^{i}}$ is $(1/30)$-close to a constant function (say, without loss of generality, the all-$0$ function), then %
    \begin{align*}
        \Prx_{\ba\sim \{0,1\}^{\sX_i}}\left[{g^{i}}(\ba) \neq {g^{i}}(\overline{\ba})\right] &= \Prx_{\ba \sim \{0,1\}^{\sX_i}}\left[{g^{i}}(\ba)=0 \land {g^{i}}(\overline{\ba}) =1\right]  + \Prx_{\ba \sim \{0,1\}^{\sX_i}}\left[{g^{i}}(\ba)=1 \land {g^{i}}(\overline{\ba}) =0\right]\\
        &\leq \Prx_{\ba \sim \{0,1\}^{\sX_i}}\left[ {g^{i}}(\overline{\ba}) =1\right]  + \Prx_{\ba \sim \{0,1\}^{\sX_i}}\left[{g^{i}}(\ba)=1\right] \leq 1/15,
    \end{align*}
    where the last line follows from the fact ${g^{i}}$ is $(1/30)$-close to the all-$0$ function under the uniform distribution and $\ba$, $\overline{\ba}$ are both sampled uniformly at random.
 Thus, we have that the probability we fail to reject ${g^{i}}$ on line 3 is at most $(1/15)^{10} \leq 1/100$.
\end{proof}

\subsubsection{Oracle simulators}

We start with the simpler $\MQ$ simulator called $\SimMQA$.
Given $\sR$, a domain-disjoint list of oracles $L=((\MQ(g^1),\sX_1),\ldots,(\MQ(g^\ell),\sX_\ell))$ and the truth table of some $\fJ:\{0,1\}^\ell\rightarrow \{0,1\}$, it answers a membership query on $x\in \{0,1\}^n$ by running $\Extract$ to obtain $z\in \{0,1\}^\ell$ and returning $\fJ(z)$.

We prove the following lemma about $\SimMQA$:

\begin{lemma}\label[lemma]{lem: MQ(A)}
Each call to $\SimMQA(x,\sR,L,\fJ)$ uses $O(\mu\log(\mu/{{\kappa}}))$  queries to  oracles in $L$,
and it satisfies the following conditions: %
\begin{flushleft}\begin{enumerate}
  \item When $L$ is a list of variable oracles, $\SimMQA$ is a $\kappa$-accurate simulator for $\MQ(\fJ_{\sigma_{L}})$; and
\item When $L$ is perfect, $\SimMQA$  simulates $\MQ(\fJ_{\sigma_{L}})$ perfectly.
\end{enumerate}\end{flushleft}
\end{lemma}

\begin{proof}
When $L$ is a list of variable oracles, we have from \Cref{lem: Extract} that $\Extract$ returns $\sigma_L^{-1}(x)$ with probability at least $1-\kappa$. 
When $L$ is perfect,   $\Extract$ always works correctly.

The query complexity follows from the fact that $\Extract$ makes $O(\mu\log(\mu/\kappa))$ queries. 
\end{proof}

The simulation of sample access, $\SimSAMPA$, to $\fJ_{\sigma_L}$ is more involved. Assume that for each $i \in [\ell]$ $g^i$ is $(1/30)$-close to $\tau(i)$ or $\overline{\tau(i)}$ under the uniform distribution and let $\Sigma:=\big\{\sigma(i) :i\in [\ell]\big\}$.
Given $y\in \{0,1\}^\sR$~and $z \in \zo^{\ell}$, let $y_{{\sigma_{L}},z}\in \zo^\sR$ denote the
  following string:
For $j \not \in \cup_{i\in [\ell]} \sX_{i}$, set $(y_{{\sigma_{L}},z})_j=y_j$;
$$
\text{the $\sX_{i}$ block of $y_{{\sigma_{L}},z}$}= \begin{cases} y_{\sR_{i}}& \text{if $y_{\tau(i)}=z_i$}\\[0.5ex]\overline{y_{\sR_{i}}} & \text{if $y_{\tau(i)}\ne z_i$}
\end{cases}
$$
for each $i\in [\ell]$.
A key observation is that, when $\by\sim \{0,1\}^\sR$ and $\bz \sim \fJ^{-1}(1)$ then $\by_{{\sigma_{L}},z}$ as defined above is distributed exactly the same as ${\sigma_{L}} (\bz)\circ \bw$
  with $\smash{\bw\sim\{0,1\}^{\sR \setminus \Sigma}}$, which corresponds exactly to a draw from  $\fJ_{{\sigma_{L}}}^{-1}(1)$.

\begin{figure}[t!]
\begin{algorithm}[H]
\caption{\protect\hypertarget{Sim-MQ-J}{$\SimMQA$}}  \label{alg: MQ(A)}
 \textbf{Input: }$x\in \{0,1\}^n$, $\sR\subseteq [n]$, a domain-disjoint list of oracles $L$ and truth table of $\fJ$ over $\zo^{|L|}$. \\
    \textbf{Output: }Either $0$ or $1$.\\
\begin{tikzpicture}
\draw [thick,dash dot] (0,1) -- (16.5,1);
\end{tikzpicture}
\begin{algorithmic}[1]
\State Let $z=\Extract(L,x,\kappa)$
\State Return $\fJ(z)$.
\end{algorithmic} 
\end{algorithm}
\caption[$\SimMQA$]{$\SimMQA$. The algorithm takes as input a domain-disjoint list of oracles $L$, the truth table of a function $\fJ : \zo^{|L|} \to \zo$. The goal of the algorithm is to simulate a $\MQ$ oracle for $\fJ_{\sigma_L}$ (assuming $L$ is a list of variable oracles). }
\end{figure}

\begin{figure}[t!]
\begin{algorithm}[H]
\caption{\protect\hypertarget{Sim-SAMP-J}{$\SimSAMPA$}}\label{alg:SAMPfA}
 \textbf{Input: }$\sR\subseteq [n]$, a domain-disjoint list of oracles $L$ and the truth table of $\fJ$ over $\{0,1\}^{|L|}$.\\ %
    \textbf{Output: }A string in $\{0,1\}^n$.\\
\begin{tikzpicture}
\draw [thick,dash dot] (0,1) -- (16.5,1);
\end{tikzpicture}

\begin{algorithmic}[1]

\State Draw $\by \sim \zo^{\sR}$ and $\bz \sim \fJ^{-1}(1)$.
\State Let $\bx=\Extract(L,w,\kappa)$.
\For{$i=1$ to $|L|$}

 \State Flip the $\sX_{i}$-block $\by_{\sX_i}$ of $\by$ to $\overline{\by_{\sR_{i}}}$ if $\bx_i\ne \bz_i$\;
 \EndFor
 
\State Return $\by$.

\end{algorithmic} 
\end{algorithm}
\caption[\SimSAMPA]{\SimSAMPA. The algorithm takes as input a domain-disjoint list of oracles $L$, the truth table of a function $\fJ : \zo^{|L|} \to \zo$. The goal of the algorithm is to simulate a $\SAMP$ oracle for $\fJ_{\sigma_L}$ (assuming $L$ is a list of variable oracles).}
\end{figure}

\begin{lemma}\label[lemma]{lem:mapback}
Each call to $\SimSAMPA(\sR,L,\fJ)$ uses $O(\mu\log(\mu/{{\kappa}}))$ queries to oracles in $L$, and it satisfies the following conditions:
\begin{flushleft}\begin{enumerate}
  \item When $L$ is a list of variable oracles, 
    $\SimSAMPA$ is a $\kappa$-accurate simulator for $\SAMP(\fJ_{\sigma_L})$; 
  \item When $L$ is perfect, 
    $\SimSAMPA$ simulates   $\SAMP(\fJ_{\sigma_{L}})$
    perfectly.
\end{enumerate}\end{flushleft}
\end{lemma}

\begin{proof}
The query complexity follows from that of $\Extract$, and 
  our choice of $\delta={{\kappa}}$.

   Whenever $\bx=\sigma_{L}^{-1}(\by)$, we have that on line $6$ the algorithm returns $\by_{\sigma_{L},\bz}$, which corresponds to a uniformly random string in $\fJ_{\sigma_{L}}^{-1}(1)$ (by our above discussion). 
   By \Cref{lem: Extract}, if $L$ is perfect, we always have $\smash{\bx=\sigma_{L}^{-1}(\by)}$ and if $L$ is a list of variable oracles,   $\smash{\bx=\sigma_{L}^{-1}(\by)}$ with probability $1-\kappa$.
\end{proof}

\subsection{Putting it all together}

We now prove \Cref{thm: learner yes case} and \Cref{thm: learner no case}.   We recall these theorems below:

\learnerYes*

\begin{proof}
The theorem follows by 
  combining \Cref{thm:FRB},
  \Cref{thm: a good approx is found},
  \Cref{lem: MQ(A)} and \Cref{lem:mapback}.
\end{proof} 

\learnerNo*

\begin{proof}%
$\DNFLearner$ either (1) returns $(\MQ(\mathbf{1}_\sR), \SAMP(\mathbf{1}_\sR))$ on line 3 (for which the condition of the theorem holds), or (2) rejects during Phase 1 or 2, or (3)
reaches Phase 3 with a domain-disjoint list of oracles $L$ and a function $\fJ\in \Approximator(|L|)$. By \Cref{lem: dnf approx remapping}, we have for any injective map $\sigma:[|L|]\rightarrow \sR$, $\rd(\fJ_\sigma,\fD)\le \kappa$ for some $r$-term, $\mu$-junta DNF $\fD$. 

By \Cref{thm: checkIfLit}, it is unlikely for $L$ to pass the test if any $g^i$ in it is $(1/30)$-far from any literal under the uniform distribution. 
On the other hand, when $L$ is a list of variable oracles, we have from 
  \Cref{lem: MQ(A),lem:mapback} that 
  $\SimMQA$ and $\SimSAMPA$ are $\kappa$-accurate
  simulators for $\fJ_{\sigma_L}$.
\end{proof}

It remains only to bound the query complexity of $\DNFLearner$. We recall 
\Cref{lem: complexity of Learner}:

\learnerComplexity*
\begin{proof}
The number of calls to $\MQ(h)$ and $\SAMP(h)$ during Phase 1 and Phase 2 can be bounded using \Cref{lem:FBR complexity} and \Cref{lem:TrimCan_complexity}
as well as our choices of parameters. 
In Phase 3, $\CheckLit$ uses $O(\mu)$ many queries to   oracles in $L$, each is simulated by a call to $\MQ(h)$.
    
The query complexity of each call to $\SimMQA$ and $\SimSAMPA$ is given in 
  \Cref{lem: MQ(A),lem:mapback}.
Note that $\kappa=\eps^2$ and each call to an oracle in the list is simulated by a call to $\MQ(h)$.
\end{proof}

%% file: sections/DNF-consistency-checking.tex
\section{Consistency Checking} \label{sec:checking}

The goal of this section is to establish the results we need about \ConsCheck, namely \Cref{lem: complexity of ConsCheck}, \Cref{thm: ConsCheck: yes case}, and \Cref{thm: ConsCheck: no case}.

$\ConsCheck$ takes as input a set $\sR\subseteq [n]$, 
  oracles $\MQ(h)$ and $\SAMP(h)$ for the input function
  $h$ and a pair of oracle simulators $\MQ^*$ and $\SAMP^*$ for some function $\fJ:\{0,1\}^{\sR}\rightarrow \{0,1\}$.
(Recall that in the yes case, they are perfect simulators for $\fJ$ and in the no case, they are $\kappa$-accurate simulators for $\fJ$ with $\kappa:=\eps^2$.)
Let us recall the results about $\ConsCheck$ that we will prove in this section:

\begin{figure}[t!]
\begin{algorithm}[H]
\caption{\protect\hypertarget{ConsCheck}{\ConsCheck}}\label{algo: consistency_checker}
 \textbf{Input: }$\emptyset \neq \sR  \subseteq [n], \epsilon, \MQ(h),\SAMP(h), \MQ^*$ and $ \SAMP^*$. \\
 \textbf{Output:} Accept or Reject.\\ 
 \begin{tikzpicture}
\draw [thick,dash dot] (0,1) -- (16.5,1);
\end{tikzpicture}
\begin{algorithmic}[1]
     \Algphase{Phase 1: Verify $h$} 
     \State Repeat $O(1/\epsilon)$ times: 
     \State \hskip2em Draw $\bz \sim \SAMP(h)$; reject
       if $\MQ^*(\bz_{\sR})=0$. %
     \State \hskip2em Draw $100$ points $\bv \sim \SAMP^*$.
     \State \hskip2em Reject if ${\MQ(h)(\bz_\sU \circ \bv)=0}$ for any of the points $\bv$ drawn on line 3.
    \Statex \cyan{\Comment{Type~2 Query}}
    \State Repeat $O(1/\epsilon)$ times: 
     \State \hskip2em Draw $\bz \sim \SAMP(h)$ and $\bv \sim \SAMP^*$; reject if ${\MQ(h)(\bz_\sU \circ \bv)=0}$.\vspace{0.1cm} %
     \cyan{\Comment{Type~2 Query}}
    
    \Algphase{Phase 2: Is this a conjunction?}
    \Statex Let
    $\Gamma : \{0,1\}^{ \sU} \to \{0,1\}$ be the following function:
    $$\Gamma(\alpha)=\begin{cases}
        1 &\text{if $\Prx_{\bz \sim \fJ^{-1}(1)}[h(\alpha \circ \bz)] \geq 0.9$} \\
        0& \text{otherwise}
    \end{cases}$$

    \State Let $\MQ'\leftarrow \MQGamma(\cdot,\MQ(h),\SAMP^*)$
    and $\SAMP'\leftarrow \SampleSub(\SAMP(h), \sU)$.
    \State Run  $\ConjTest(\MQ',\SAMP',\eps/10)$ and return the same answer. %
\end{algorithmic}
\end{algorithm}
\caption[\ConsCheck: Given functions $f$ and $g$, the algorithm checks if $f$ is relative error close to $C \land g$ for some {conjunction} $C$]{\ConsCheck. The algorithm takes as input $\sR \subseteq [n], \epsilon>0$, oracle access to $h$ and oracles $\MQ^*, \SAMP^*$ to some function $\fJ:\zo^\sR \to \zo$. The algorithm checks if there is some conjunction $C$ such that $f$ is relative error close $C(x_{\overline{R}}) \land \fJ(x_\sR)$. In particular, the algorithm will always be used with $\fJ$ being close to some DNF $r$-term, $\mu$-junta DNF, hence if $f$ is $\epsilon$ far in relative distance from every $(r,\mu)-$factored DNFs, the algorithm rejects with high probability. 
}
\end{figure}

\begin{figure}[t!]
\begin{algorithm}[H]\caption{\protect\hypertarget{Sim-MQ-$\Gamma$}{\MQGamma}}\label{algo:gamma_query} \textbf{Input: }$\alpha \in \{0,1\}^{ \sU}$, $\MQ(h)$ and $\SAMP^*$.\\ 
\textbf{Output: }Either Reject or $1$. \\
 \begin{tikzpicture}
\draw [thick,dash dot] (0,1) -- (16.5,1);
\end{tikzpicture}
    \begin{algorithmic}[1]
        \State Draw $\bw \sim \SAMP^*$ for $O(\log(1/\epsilon))$ %
        many times. 
        \State Return $1$ if ${\MQ(h)(\alpha \circ {\bw})=1}$ for all $\bw$ drawn in the previous step. \cyan{\Comment{Type~2 Query}}
        \State Otherwise, halt and reject.
    \end{algorithmic}
\end{algorithm}
\caption[\MQGamma]{\MQGamma. An algorithm to simulate an $\MQ$ oracle for the function $\Gamma$ defined in \ConsCheck.}
\end{figure}

\setcounter{algorithm}{0}
\begin{algorithm}[H]
\caption{\protect\hypertarget{Sim-SampV2}{\protect\hyperlink{Sim-SampV2}{\textcolor{violet}{\textbf{Sim-SAMP}}}}}
\textbf{Input: }$\SAMP(h)$ of some function $h:\zo^n \to \zo$ and $\emptyset \neq \mathsf{Y} \subseteq [n]$.  \\
\textbf{Output: }a point in $\{0,1\}^\mathsf{Y}$.\\
\begin{tikzpicture}
\draw [thick,dash dot] (0,1) -- (16.5,1);
\end{tikzpicture}
    \begin{algorithmic}[1]
    \State Draw $\bz \sim \SAMP(h)$.
    \State Return $\bz_\mathsf{Y}$.
    \end{algorithmic}
\end{algorithm}

\constestComplexity*
\conscheckYes*
\conscheckNo*

$\ConsCheck$ is presented in \Cref{algo: consistency_checker}. 
We quickly give a proof of \Cref{lem: complexity of ConsCheck}. For this we recall the query complexity of $\ConjTest$ from \cite{rel-error-conj-DL}:

\begin{theorem} \label{thm:a}
    $\ConjTest(\MQ(f), \SAMP(f), \epsilon)$ makes $O(1/\epsilon)$ queries to  $\MQ(f)$ and $\SAMP(f)$. 
\end{theorem}

\begin{proofof}{\Cref{lem: complexity of ConsCheck}}
    Phase~1 makes $O(1/\epsilon)$ queries. 
    In Phase~2, $\ConjTest$ uses $O(1/\epsilon)$ calls to~$\MQ'$ and $\SAMP'$.
    $\SAMP'$ uses $\SampleSub(\SAMP(h),\sU)$, so each call to it uses one call to $\SAMP(h)$. 
    $\MQ'$ uses $ \MQGamma(\cdot,\MQ(h),\SAMP^*)$, so each call to $\MQ'$ uses $O(\log(1/\epsilon))$ calls to $\MQ(h)$ and $\SAMP^*$, respectively. 
    Hence, the total number of oracle calls  is $O(\log(1/\epsilon)/\epsilon)$. 
\end{proofof}

It remains to prove \Cref{thm: ConsCheck: yes case} and \Cref{thm: ConsCheck: no case}. Before doing this, we record a ``robustness'' property of the $\ConjTest$ algorithm that we will need.

\subsection{Robustness of $\ConjTest$}
\label{sec:robustness-conjtest}

We recall a few results about $\ConjTest$ from %
\cite{rel-error-conj-DL}.
The first one says that if $\ConjTest$ is run on a conjunction, it accepts with probability 1 even if the sampling oracle is defective (as long as it only returns satisfying assignments of $f$):

\begin{theorem}[Theorem~27 of
\cite{rel-error-conj-DL}]
\label{thm: CPPR conj test}
    Let $f: \zo^n \to \zo$ be a conjunction.
    Let $\SAMP^\ast(f)$ be an oracle which, when called, always returns some string in $f^{-1}(1)$.
Then \ConjTest, running on $\MQ(f)$ and $\SAMP^*(f)$,
accepts $f$ with probability 1.
\end{theorem}

The second result says that if $\ConjTest$ is run on a function $f$ that is relative-error far from every conjunction, then it rejects with high probability as long as it is run with the $\MQ(f)$ oracle and a version of the $\SAMP(f)$ oracle that is ``not too inaccurate'' (in the sense of putting non-trivial weight on every satisfying assignment):

\begin{theorem} [Theorem~28 of %
\cite{rel-error-conj-DL}]
\label{thm: conjtest: reject weaker assumption}
    Suppose that $f: \zo^n \to \zo$ satisfies $\rd(f,f') > \epsilon$ for every conjunction $f'$. Let 
    $\SAMP^\ast(f)$ be a sampling oracle for $f$ such that the underlying~distribution 
    $\mathcal{D}$ (the distribution that $\SAMP^\ast(f)$ generates samples from) satisfies
\[\Prx_{\bx \sim {\mathcal{D}}}\big[\bx = z\big] \geq \frac{1}{20|f^{-1}(1)|},\quad\text{for every $z\in f^{-1}(1)$}.\]
    Then {$\ConjTest$}, running on $\MQ(f)$ and $\SAMP^*(f)$,    rejects $f$ with probability at least $99/100$. 
\end{theorem}

For our purposes we will also need a ``robustness'' guarantee on the performance of $\ConjTest$, when it is run on a conjunction with imperfect sampling and membership-query oracles, which goes beyond the guarantee provided by \Cref{thm: CPPR conj test}. We establish the guarantees that we will need in \Cref{thm: conjtest: accept weaker assumption samples}, and for completeness we reproduce the $\ConjTest$ algorithm in \Cref{ap:overview-conjtest}:

\begin{lemma}\label{thm: conjtest: accept weaker assumption samples}
    Let $C: \zo^n \to \zo$ be a conjunction.
    Assume the oracle $\SAMP^\ast(C)$ always returns some string in $C^{-1}(1)$. 
    Then $\ConjTest$ on $\SAMP^\ast(C)$, $\MQ^\ast(C)$ has the following properties:
    
    (1) It only calls $\MQ^\ast(C)$ on  points that have $C(x)=1$. 
    
    (2) If  all these queries are answered with $1$, then $\ConjTest$ accepts with probability $1$
\end{lemma}
\begin{proof}

By inspection of $\ConjTest$ (see Algorithm~2 and Remark~10 of \cite{rel-error-conj-DL}, or \Cref{ap:overview-conjtest}), we see that $\ConjTest$ accepts if and only if every call to $\MQ^*(C)$ that it makes is answered with~$1$; this establishes (2). 

We now turn to (1). First, observe that (again by inspection of $\ConjTest$) $\ConjTest$ is non-adaptive, and hence the calls to $\MQ^*(C)$ that it makes do not depend on whether it is run with $\MQ(C)$ or $\MQ^\ast(C)$. Hence, it suffices to show that given as oracles $\MQ(C)$ and an oracle $\SAMP^\ast(C)$ which only returns strings in $C^{-1}(1)$, $\ConjTest$ only calls $\MQ(C)$ on points $x$ with $C(x)=1$.

Assume, for the sake of contradiction, that there exists some conjunction $g$ such that $\ConjTest$, when run on $\MQ(g)$ and $\SAMP^*(g)$ (which only returns strings in $g^{-1}(1)$), {has a nonzero probability of} calling $\MQ(g)$ on some point $z$ with $g(z)=0$. Since $\ConjTest$ accepts if and only if every call to $\MQ(g)$ it makes is answered with $1$, this would imply that $\ConjTest$ rejects {with nonzero probability}. This contradicts \Cref{thm: CPPR conj test}, and hence the proof of the lemma is complete.
\end{proof}

\subsection{The yes case: Proof of \Cref{thm: ConsCheck: yes case}}

We recall the assumptions of \Cref{thm: ConsCheck: yes case}:  $\sR$ is stable, and   $\smash{(\MQ^*,\SAMP^*)}$ are perfect simulators for some function $\fJ:\{0,1\}^\sR\rightarrow \{0,1\}$ with $\rd(\fh,\fJ)\le \kappa.$ 

Since $\kappa \ll 1/2$, by symmetry (see \Cref{lem:approx-symetric}), we have $\rd(\fJ,\fh) \leq 2\kappa$. By definition, since $\sR$ is stable, we have that 
there there exists $u \in \zo^\sV$, where $\sV:=\overline{\sR}\cap (\sS\cup \var(H))$, such that
$$
\Prx_{\bz\sim h^{-1}(1)}\big[\bz_\sV\ne u\big]\le \xi,$$
which implies that
$$
\Prx_{\bz \sim h^{-1}(1)}\big[h{\upharpoonleft_{\bz_{\overline{\bR}}}}\not\equiv \fh\big] \leq {{{\xi}}}.$$
In particular, whenever we sample $\bz \sim h^{-1}(1)$ with probability at least $1-\xi$ we have  $h{\upharpoonleft_{\bz_\sU}}\equiv\fh$ and $\bz_\sR \sim \fh^{-1}(1)$.  
Given the bounds on $\rd(\fh,\fJ)$ and $\rd(\fJ,\fh)$, we will use the following: $$\Prx_{\bw \sim \fJ^{-1}(1)}\big[\fh(\bw)=1\big]=O(\kappa)\quad\text{and}\quad \Prx_{\bz \sim \fh^{-1}(1)}\big[\fJ(\bz)=1\big] = O(\kappa).$$

With these in mind, we first prove the following:

\begin{lemma}\label{lem: ConsCheck good 1-8}
    $\ConsCheck$ rejects on lines 1--6 with probability at most $1/80$.  
\end{lemma}

\begin{proof}
Between lines 1--6, $\ConsCheck$ draws $O(1/\epsilon)$ samples $\bz \sim \SAMP(h)$. By a union bound, with probability at least $1-O(\xi/\eps)$,  every $\bz$  satisfies
$h{\upharpoonleft_{\bz_\sU}}=\fh$ and thus, $\bz_\sR$ is drawn from $\fh^{-1}(1)$. 
 
So assume the above holds. We now look at the probability to reject on lines~2,~4 and~6. Each time line $2$ is executed, the algorithm rejects if $\fJ(\bz_\sR)=0$ where $\bz_\sR \sim \fh^{-1}(1)$, so this happens with probability at most $O(\kappa)$. 
On line $4$, the algorithm rejects if $\fh(\bv)=0$ when $\bv\sim \fJ^{-1}(1)$, which happens with probability at most $O(\kappa)$ as well.
By the same argument, each time line $6$ is executed, the algorithm rejects with probability at most $O(\kappa)$. Hence, by a union bound, the algorithm rejects on these lines with probability at most $O(\kappa/\epsilon).$ 

Overall the algorithm rejects on lines 1--6 with probability $O(\kappa/\eps)+O(\xi/\eps)\le 1/80$.
\end{proof}

Now, let $C^*: \zo^\sU\rightarrow \zo$ be the conjunction defined by $C^*(\alpha)=1$ if and only $\alpha_\sV=u$. We now prove that line~8 rejects with small probability.
\begin{lemma}\label{lem: conj-test accept whp}
   The call to $\ConjTest$ on line~12 rejects  with probability at {most $1/80$}.
\end{lemma}
\begin{proof}
    During this phase, $\ConjTest$ will make $O(1/\epsilon)$ many calls to $\MQ'$ and $\SAMP'$ (see \Cref{thm:a}). Since we're using $\MQGamma$ and $\SampleSub$ to simulate the oracles, running $\ConjTest$ will use $O(\log(1/\epsilon)/\epsilon)$ many samples from $\SAMP(h)$ and $\SAMP^*$. 

    Hence, with failure probability $O(  \xi \cdot \log(1/\epsilon) /\epsilon) \leq 1/160$, every sample  $\bw \sim \SAMP(h)$ satisfies $\bw_\sV=u$. 
    When this happens, we have that $\SAMP^*(\Gamma)=\SampleSub(\SAMP(h), \sU)$ only returns strings in $(C^*)^{-1}(1)$.  Hence by \Cref{thm: conjtest: accept weaker assumption samples} $\ConjTest$ only queries points $\alpha$ with $C^*(\alpha)=1,$ and if $\MQ'$ answers $1$ on every of these points, $\ConjTest$ accepts. 
    
    Now, note that for any $\alpha$ with $C^*(\alpha)=1$, we have that $\MQ'$ returns $1$ unless there is some $\bw \sim \SAMP^*$ (which is the same as $\SAMP(\fJ)$) such that $h(\alpha \circ \bw)=0$.  Since $\alpha_\sV=u$, this is equivalent to $\fh(\bw)=0$ and happens with probability $O(\kappa)$. Hence, by union bound over all at most $O(\log(1/\epsilon)/\epsilon)$ samples $\bw \sim \SAMP^*$, with probability at least $1-1/160$ every query to $\MQ'$ is answered with $1$, and $\ConjTest$ accepts. 

    Hence, $\ConjTest$ accepts with probability at least $1-1/80$. 
\end{proof}

Finally, we can prove \Cref{thm: ConsCheck: yes case}:
\begin{proof}[Proof of \Cref{thm: ConsCheck: yes case}]
    By \Cref{lem: ConsCheck good 1-8}, $\ConsCheck$ rejects on lines~1--6 with probability at most $1/80$, and by \Cref{lem: conj-test accept whp} $\ConsCheck$ rejects on line~8 with probability at most $1/80$. So by a union bound, $\ConsCheck$ rejects with probability at most $1/40.$
\end{proof}

\subsection{The no case: Proof of \Cref{thm: ConsCheck: no case}}

Recall that in the no case, 
  $h$ is $\eps$-far from every $(r,\mu)$-factored DNF in 
  relative distance and $(\MQ^*,\SAMP^*)$ are $\kappa$-accurate oracle simulators for some function $\fJ:\{0,1\}^{\sR}\rightarrow \{0,1\}$ with $\rd(\fJ,\calD)\le \kappa$ for some $r$-term, $\mu$-junta DNF $\fD$\footnote{Recall that this means $\fD$ depends on at most $\mu$ variables.}.
Given that $\kappa=\eps^2$ is much smaller than the number of calls made by $\ConsCheck$ on $\MQ^*$ and $\SAMP^*$ (\Cref{lem: complexity of ConsCheck}), we may assume without loss of generality that they are perfect simulators for $\MQ(\fJ)$ and $\SAMP(\fJ)$, and prove under this assumption that $\ConsCheck$ rejects with probability at least $0.96$.

To analyze $\ConsCheck$, we introduce a new function $\gamma:\zo^n \to \zo$ defined as
\begin{equation} \label{eq:gammadef}
\gamma(z)=\begin{cases} 1 &
        \text{if $\Gamma(z_\sU)=1$ and $\fJ(z_{\sR})=1$} \\
        0 &\text{otherwise}
    \end{cases}
\end{equation}
where $\Gamma:\{0,1\}^{\overline{\sR}}\rightarrow \{0,1\}$ was defined at the beginning of Phase 2 in $\ConsCheck$.

We now enter into the formal proof. We first show that if $h$ has ``many'' satisfying assignments that are not satisfying assignments of $\gamma$, then $\ConsCheck$ rejects with high probability:
\begin{lemma}\label{lem: constcheck line 6}
    Suppose that $$|h^{-1}(1) \setminus \gamma^{-1}(1)| \geq \frac{\epsilon {|h^{-1}(1)|}}{20}. $$ Then $\ConsCheck$ rejects on lines 1--4 with probability at least $99/100$.
\end{lemma}
\begin{proof}
    Since $|h^{-1}(1) \setminus \gamma^{-1}(1)| \geq \epsilon {|h^{-1}(1)|}/20$, 
whenever we sample $\bz \sim h^{-1}(1)$ on line $2$, with probability at least $\epsilon/20$ we have $\gamma(\bz)=0$. Hence with failure probability at most $(1-\epsilon/20)^{{O(1/\epsilon)}} \leq 1/200$, we get such a sample $\bz^*$ during some iteration of lines 1--4. For this $\bz^*$, if $\fJ(\bz^*_{\sR})=0$ then the algorithm rejects on line $2$. Otherwise, if $\fJ(\bz^*_{\sR})=1$, recalling \Cref{eq:gammadef} we must have $\Gamma(\bz^*_{\sU})=0$, but in this case $\ConsCheck$ rejects on line $4$ with probability at least $\smash{1-(0.9)^{O(1)} \geq 1-1/200}$. Hence, the algorithm rejects $h$ during lines 1--4 with probability at least $1-1/100$. 
\end{proof}

Next, we show that if almost all satisfying assignments of $h$ are satisfying assignments of $\gamma$, then each satisfying assignment of $\Gamma^{-1}$ is ``fairly likely'' to occur as the $\sU$-portion of a uniform satisfying assignment of $h$:

\begin{lemma}\label{lem:almost_uniform_for_Gamma}
    Assume that $|h^{-1}(1) \setminus \gamma^{-1}(1)| \leq {\epsilon {|h^{-1}(1)|}}/{20}.$
    For any $\alpha\in \{0,1\}^{\sU}$ with $\alpha \in \Gamma^{-1}(1)$, %
$$
\Prx_{\bz \sim h^{-1}(1)}\big[\bz_\sU= \alpha\big] \geq \frac{1}{2{|\Gamma^{-1}(1)|}}.$$
\end{lemma}
\begin{proof}
    By assumption we have that $|\gamma^{-1}(1)| \geq {|h^{-1}(1)|}(1-\frac{\epsilon}{20})$, and since $|\gamma^{-1}(1)|={|h^{-1}(1)|}\cdot {|\Gamma^{-1}(1)|}$ we have $1.1 \cdot  {|\fJ^{-1}(1)|}\cdot {|\Gamma^{-1}(1)|} \geq {|h^{-1}(1)|}$. So consider any $\alpha \in \zo^\sU$ with $\Gamma(\alpha)=1$. By the definition of $\Gamma$ we have
    $$|\{ z \in \zo^n \mid \fJ(z_{\sR})=1 \land z_\sU= \alpha \} \cap h^{-1}(1)| \geq 0.9 {|\fJ^{-1}(1)|}.$$
   This implies that
     $$\Prx_{\bz \sim h^{-1}(1)}[\bz_\sU = \alpha] \geq  \frac{0.9{|\fJ^{-1}(1)|}}{{|h^{-1}(1)|}} \geq \frac{0.9}{1.1 {|\Gamma^{-1}(1)|}} \geq \frac{1}{2{|\Gamma^{-1}(1)|}}.\qedhere$$
\end{proof}

Next, we show that even if almost all satisfying assignments of $h$ are also satisfying assignments of $\gamma$, if there are ``many'' satisfying assignments of $\gamma$ that are not satisfying assignments of $h$ then $\ConsCheck$ is likely to reject:
\begin{lemma}\label{lem: constcheck line 8}
    Assume that $\ConsCheck$ reaches line $5$ with $|h^{-1}(1) \setminus \gamma^{-1}(1)| \leq {\epsilon {|h^{-1}(1)|}}/{20}.$ If  $$|\gamma^{-1}(1) \setminus h^{-1}(1)| \geq \frac{\epsilon {|h^{-1}(1)|}}{20},$$  then $h$ gets rejected on line $6$ with probability at least $99/100$.
\end{lemma}

\begin{proof}

Assume that $|\gamma^{-1}(1) \setminus h^{-1}(1)| \geq \epsilon {|h^{-1}(1)|}/20$. This can be rewritten as
$$\Prx_{\bz \sim \gamma^{-1}(1)}[h(\bz)=0] \geq \frac{\epsilon {|h^{-1}(1)|}}{20|\gamma^{-1}(1)|}.$$
If $|\gamma^{-1}(1)| \geq 2{|h^{-1}(1)|},$ then we trivially have $\Prx_{\bz \sim \gamma^{-1}(1)}[h(\bz)=0] \geq 1/2 \geq \epsilon/40.$ %
Otherwise,
$$\Prx_{\bz \sim \gamma^{-1}(1)}[h(\bz)=0] \geq \frac{\epsilon  |\gamma^{-1}(1)|} {40  |\gamma^{-1}(1)|} = \epsilon/40,$$
so in both cases we have
$\Prx_{\bz \sim \gamma^{-1}(1)}[h(\bz)=0] \geq \epsilon/40.$

   Now, observe that the following process generates a uniformly random sample $\bw \sim \gamma^{-1}(1)$: 
   Independently sample $\boldsymbol{\alpha} \sim \Gamma^{-1}(1)$ and $\bv \sim \fJ^{-1}(1)$, and return $\bw=\boldsymbol{\alpha} \circ \bv$.
   Hence, by \Cref{lem:almost_uniform_for_Gamma} we can conclude that for any fixed $w \in \gamma^{-1}(1),$ the following process returns $w$ with probability at least $1/(2|\gamma^{-1}(1)|)$: independently sample $\bz \sim h^{-1}(1)$ and  $\bv \sim \fJ^{-1}(1)$, and return $\bw=\bz_\sU \circ \bv$. 
   It follows that
each time line~$8$ is executed, we have that $\Pr[h(\bz_\sU \circ \bv)=0] \geq \epsilon/80$. Hence, the algorithm rejects on line~8 with probability at least  $$1 - (1-\epsilon/80)^{O(1/\epsilon)} \geq 1 - 1/100. \qedhere$$
\end{proof}

The following lemma, roughly speaking, tells us that if $h$ reaches line $7$ with high probability then $\Gamma$ must be far from every conjunction (which is later used to show that $\ConjTest$ rejects on line 8 with high probability): 
\begin{lemma}\label{lem: Gamma far from conj}
    Assume that $\rd(h,g)>\epsilon$ for every $(r,\mu)$-factored-DNF $g$, and that moreover the following conditions all hold:
    \begin{itemize}
        \item [(a)] $\rd(\fJ,\fD)\leq \epsilon/10$ for some $r$-term, $\mu$-junta DNF $\fD$. 
        \item [(b)] $|h^{-1}(1) \setminus \gamma^{-1}(1)| \leq {\epsilon {|h^{-1}(1)|}}/{20}$.
        \item [(c)] $|\gamma^{-1}(1) \setminus h^{-1}(1)| \leq {\epsilon {|h^{-1}(1)|}}/{20}$.
    \end{itemize}
    Then $\rd(\Gamma,C)>\epsilon/10$ for every conjunction $C$.
\end{lemma} 

\begin{proof}
    Assume for contradiction that  $\rd(\Gamma,C)\leq \epsilon/10$
  for some conjunction $C$ over $\{0,1\}^{\sU}$.
We argue, using conditions (a)---(c), that $\rd(h,g)< \epsilon$ for some $(r,\mu)$-factored-DNF $g$, giving the desired contradiction with the lemma's assumption that $\rd(h,g)>\epsilon.$

First, using items (b) and (c), we have $\rd(h,\gamma)\leq \epsilon/10$. 

The proof proceeds in two steps, defining two Boolean functions $\psi$ and $g$ over $\{0,1\}^n$, where $g$ will be a $k$-term DNF.
    We will prove the following inequalities: 
$$ \text{(i)~} \rd(\gamma,\psi)\le \epsilon/10, \quad\text{and}\quad \text{(ii)~} \rd(\psi,g)\le \epsilon/10 .$$
Combining with $\rd({h},\gamma)\le \epsilon/10 $, it follows from  
  the approximate triangle inequality  for relative distance (\Cref{lem: approx triangle ineq})
  that  $\rd({h},g)<\epsilon$; so it remains to establish (i) and (ii).
    
We first establish (i).  By condition (a), we have $\rd(\fJ,\fD)\le \epsilon/10 $ for some $r$-term, $\mu$-junta DNF $\fD$. Our first function
$\psi$ over $\{0,1\}^n$ is defined as %
    $$\psi(z) :=\Gamma(z_\sU) \wedge \fD(z_{\sR}).$$    
Observe that for any $\alpha \in \{0,1\}^{\sU}$, if $\Gamma(\alpha)=0$ then $\gamma{\upharpoonleft_\alpha} \equiv \psi{\upharpoonleft_\alpha} \equiv 0$. And if $\Gamma(\alpha)=1$, then $\gamma{\upharpoonleft_\alpha} \equiv \fJ$ while $\psi{\upharpoonleft_\alpha} \equiv \fD$, in which case ${\rd(\gamma{\upharpoonleft_\alpha}, \psi{\upharpoonleft_\alpha} )}=\rd(\fJ, \fD) \leq \epsilon/10 $. We thus have:
     \begin{align*}
        \rd(\gamma, \psi)&=\frac{1}{|\gamma^{-1}(1)|} \sum_{\alpha: \Gamma(\alpha)=1} \left| \big\{w \in \{0,1\}^{\sR}:\gamma(\alpha\circ w) \neq \psi(\alpha\circ w)\big\}\right| \\
        &= \frac{1}{|\gamma^{-1}(1)|} \sum_{\alpha:\Gamma(\alpha)=1} \rd(\gamma{\upharpoonleft_\alpha}, \psi{\upharpoonleft_\alpha}) \cdot \big|(\gamma{\upharpoonleft_\alpha)^{-1}}(1)\big| \\
        &\leq \frac{{\epsilon}}{10 }\cdot \frac{1}{|\gamma^{-1}(1)|} \sum_{\alpha:\Gamma(\alpha)=1} \big|\gamma{\upharpoonleft_\alpha^{-1}}(1)\big|.
    \end{align*}
Since $|\gamma^{-1}(1)|=\sum_{\alpha }|(\gamma{\upharpoonleft_\alpha)^{-1}}(1)|$, we get that $\rd(\gamma, \psi) \leq \epsilon/10$, giving (i).

Turning to (ii), first recall that $C$ is a conjunction over $\{0,1\}^{\sU}$ that satisfies
  $\rd(\Gamma,C)\le \epsilon/10$.
The last function $g:\{0,1\}^n\rightarrow \{0,1\}$ is defined  as
    $$g(z):=C(z_\sU) \wedge\fD(z_{\sR})$$
    (i.e.~we are swapping out $\Gamma(z_\sU)$ for $C(z_\sU)$ in the definition of $\psi$).
    Note that since $\fD$ is a $r$-term, $\mu$-junta DNF, $g$ is a $(r,\mu)$-factored-DNF.
    Below we bound $\rd(\psi,g)$.

    Observe that for each $\alpha \in \{0,1\}^{ \sU}$, if $\Gamma(\alpha)=C(\alpha)$ then $g{\upharpoonleft_\alpha} \equiv \psi{\upharpoonleft_\alpha}$. Otherwise, we have that~one of  $g{\upharpoonleft_\alpha}$ or $\psi{\upharpoonleft_\alpha}$ is $\fD$ while the other is the all-$0$ function. 
    So we have 
    \begin{align*}
        \rd(\psi, g)&=\frac{1}{\left|\psi^{-1}(1)\right|} \cdot \left|\big\{\alpha\in \{0,1\}^{\sU}: \Gamma(\alpha)\ne C(\alpha)\big\}\right|\cdot \big|\fD^{-1}(1)\big|.
    \end{align*}
Since $ 
         |\psi^{-1}(1) |
         = | \Gamma^{-1}(1) |\times  |\fD^{-1}(1) |,
    $  
    we get that
    \begin{align*}
          \rd(\psi, g)
          = \frac{ | \{\alpha\in \{0,1\}^{\sU}: \Gamma(\alpha)\ne C(\alpha) \} |}{|\Gamma^{-1}(1)|}= \rd(\Gamma,C)\le \frac{\epsilon}{10},
           \quad \text{giving (ii).}
    \end{align*}
    This finishes the proof of the lemma.
\end{proof}

By \Cref{thm:a} we have that when $\ConsCheck$ is run in Phase~2 with distance $\epsilon/10$, it will make $O(1/\epsilon)$ queries to $\MQ', \SAMP'$. Let $c_0/\epsilon$ be an upper bound on this number of queries. We require the following simple lemma establishing the correctness of $\MQ'$: 

\begin{lemma}\label{thm: no query mistake when far}
   For any $\alpha\in \{0,1\}^{\sU}$, 
    $\MQ'(\alpha)$ satisfies
    $$
    \Prx\big[\MQ'(\alpha) \text{ does not reject 
     and  returns $1-\Gamma(\alpha)$}\big]
    \le \frac{\epsilon}{100c_0}.$$
\end{lemma}
\begin{proof}
   Recall we have $\MQ'\leftarrow \MQGamma(\cdot,\MQ(h),\SAMP^*, \epsilon)$, where $\SAMP^*$ is just $\SAMP(\fJ)$. We first suppose that $\Gamma(\alpha)=0$; this implies that $\Prx_{\bz \sim \fJ^{-1}(1)}[h(\alpha \circ {\bz})] \leq 0.9$. Consequently the probability that \Cref{algo:gamma_query} 
  returns $1$ without rejecting $h$  is at most
$(0.9)^{O(\log(1/\epsilon))}\le \frac{\epsilon}{100c_0}$ by picking the hidden constant large enough. The other possibility is that $\Gamma(\alpha)=1$; in this case 
  the algorithm either returns $1$ or rejects. 
\end{proof}

\begin{lemma}\label{lem: ConjTest rejects whp}
    Assume that $\rd(\Gamma, C) \geq \epsilon/10$ for every conjunction $C$ and that $|h^{-1}(1) \setminus \gamma^{-1}(1)| \leq {\epsilon {|h^{-1}(1)|}}/{20}$. Then if $\ConsCheck$ reaches line 7, the call to $\ConjTest$ on line 8 rejects with probability at least $49/50$.
\end{lemma}
\begin{proof}
    We set $\SAMP'=\SampleSub(\SAMP(h), \sU)$, so by \Cref{lem:almost_uniform_for_Gamma}, we have for each $\alpha \in \Gamma^{-1}(1)$, $$\Prx\big[ \SAMP' \text{ returns }\alpha\big] \geq \frac{1}{2|\Gamma^{-1}(1)|}.
    $$
If all the (at most $c_0/\epsilon$ many)  membership queries on $\MQ'$ made by $\ConjTest$ are answered correctly, then
  by \Cref{thm: conjtest: reject weaker assumption} $\ConjTest$ rejects $\Gamma$ 
  with probability at least $99/100$.
By \Cref{thm: no query mistake when far} and a union bound, with probability at least $99/100$, either $\MQ'$ decides to reject or returns the correct answer.
As a result, line 8 of \ConsCheck rejects with probaiblity at least $49/50$.
\end{proof}

Finally, we can put all the pieces together.
\begin{proof}[Proof of \Cref{thm: ConsCheck: no case}]
    If $|h^{-1}(1) \setminus \gamma^{-1}(1)| \geq \epsilon {|h^{-1}(1)|}/20$, then $\ConsCheck$  rejects during lines 1--4 with probability at least $99/100$ by \Cref{lem: constcheck line 6}, so we may suppose that $\ConsCheck$ reaches line~5 with $|h^{-1}(1) \setminus \gamma^{-1}(1)| \leq {\epsilon {|h^{-1}(1)|}}/{20}.$ In this case, if $|\gamma^{-1}(1) \setminus h^{-1}(1)| \geq \epsilon {|h^{-1}(1)|}/20$ then by \Cref{lem: constcheck line 8} $\ConsCheck$ rejects on line~$6$ with probability at least $99/100$.
So we may suppose that we have both $|h^{-1}(1) \setminus \gamma^{-1}(1)|,|h^{-1}(1) \setminus \gamma^{-1}(1)| \leq \epsilon {|h^{-1}(1)|}/20$.  Then, by \Cref{lem: Gamma far from conj}, we have that $\rd(\Gamma,C)>\epsilon/10$ for every conjunction $C$, and hence by \Cref{lem: ConjTest rejects whp}, $h$ is rejected on line $9$ with probability at least $49/50$. Combining all the failure probabilities, we get that $\ConsCheck$ rejects with probability at least $1 - (1/100 + 1/100 + 1/50)= 0.96.$
\end{proof}

%% file: sections/ap-DDS.tex

\section{A $n^{O(\log(s/\epsilon))}$-query relative-error tester for $s$-term DNFs}
\label{ap:DDS}

\begin{observation} \label{obs:DDS}
For $0<\eps \leq 1/2$, there is an $\eps$-relative-error testing algorithm for the class of $s$-term DNF formulas over $\{0,1\}^n$ which makes $n^{O(\log(s/\eps))}$ calls to the sampling oracle $\SAMP(f)$ and $O(1/\eps)$ calls to the membership query oracle $\MQ(f)$. 
\end{observation}
 
 \begin{proof}[Proof Sketch]
 
 In \cite{DDS15} De et al.~gave an algorithm that uses $n^{O(\log(s/\eps))}$  independent samples from $f^{-1}(1)$ to \emph{learn} any unknown $s$-term DNF formula $f$ over $\zo^n$ to $\eps$-accuracy in relative error (see Theorem~1.4 of \cite{DDS15}).  Such a learning algorithm easily yields a testing algorithm in our relative-error model as follows:
 
 \begin{enumerate}
     \item 
     First, run the learning algorithm with accuracy parameter $\eps/10000$ to obtain a hypothesis function $h': \zo^n \to \zo$.  (Note that this hypothesis $h'$ need not be an $s$-term DNF.) 

     \item Next, perform a brute-force search over all $s$-term DNF formulas $h$ to find the one that minimizes $\max\{\rd(h,h'),\rd(h',h)\}$.  If no $s$-term DNF $h$ has $\max\{\rd(h,h'),\rd(h',h)\}$ $ \leq \eps/1000$ then halt and reject, otherwise continue.

     \item Finally, draw $10/\eps$ random samples from $f^{-1}(1)$ (respectively $h^{-1}(1)$) and check that they are also satisfying assignments of $h$ (respectively $f$); if this is the case then accept, otherwise reject.

\end{enumerate}

It is clear that the number of calls to $\SAMP(f)$ and $\MQ(f)$ are as claimed; we briefly sketch the argument establishing correctness.
     
First, suppose that $f$ is indeed an $s$-term DNF formula. Then the learning algorithm will w.h.p.~(say, at least $99/100$) construct a hypothesis $h'$ such that $\rd(f,h') \leq \eps/10000.$  In step~2, the brute-force search will find an $s$-term DNF $h$ such that $\rd(h,h'),\rd(h',h) \leq \eps/1000$ (since $f$ itself is such an $s$-term DNF), so the procedure will continue to step~3.  Since $\rd(f,h') \leq \eps/10000$ and $\rd(h',h),\rd(h,h') \leq \eps/1000$, using \Cref{lem:approx-symetric,lem: approx triangle ineq} we get that  $\rd(f,h),\rd(h,f) \leq \eps/100$. This  implies that in step~3, there is at most a $1/10$ chance that any of the $10/\eps$ random samples from $f^{-1}(1)$ fails to satisfy $h$, and likewise with the roles of $f$ and $h$ reversed; so with probability at least $79/100$ the algorithm will accept.

Now, suppose that $f$ has $\rd(f,g)>\eps$ for every $s$-term DNF formula $g$.  If the procedure rejects in step~2 then we are done, so suppose that step~2 succeeds in finding an $s$-term DNF $h$ with $\rd(h,h'),\rd(h',h) \leq \eps/1000$.  Now, since  $h$ is an $s$-term DNF we must have that $\rd(f,h)>\eps$. 
 From this it is easy to show that with high probability (at least $2/3$), either some sample among the $10/\eps$ random samples from $f^{-1}(1)$ will fail to satisfy $h$ or some sample among the $10/\eps$ random samples from $h^{-1}(1)$ will fail to satisfy $f$. Hence the procedure rejects with probability at least $2/3$, and we are done.
 \end{proof}

%% file: sections/DNF-incomparable.tex
\section{Distribution-free testing and relative-error testing are incomparable} \label{sec:incomparable}

In this section we record the fact that distribution-free testing and relative-error testing are incomparable models:  for each of these two models, there is a class of functions which is easy to test in the first model (i.e.~testable with only constantly many queries) but hard to test in the other one (i.e.~the query complexity must grow with the dimension $n$). The following observation makes this quantitatively precise:

\begin{observation} [Distribution-free testing and relative-error testing are incomparable]
\label{obs:incomparable}
~
\begin{enumerate}

    \item There is a class of functions ${\cal F}_1$ which can be $\eps$-tested in the relative-error model using only $O(1/\eps)$ samples and membership queries, but which requires $\tilde{\Omega}(n^{1/3})$ samples and membership queries for $\eps_0$-testing in the distribution-free model, for some absolute constant $\eps_0>0$.
    
    \item There is a class of functions ${\cal F}_2$ which can be $\eps$-tested in the distribution-free model, for any $\eps >2^{-n/3}$, using zero samples and membership queries, but which requires $\Omega(2^{n/3})$ samples and membership queries for $\eps$-testing in the relative-error model for any $\eps<1$.
\end{enumerate}
\end{observation}

For part~1, we can take ${\cal F}_1$ to be the class of all conjunctions over $\zo^n$.  The positive result for relative-error testing is established as Theorem~1 of \cite{rel-error-conj-DL}, and the negative result is Theorem~1.4 of \cite{CX16}.

For part~2, we take ${\cal F}_2$ to be the class of all $g: \zo^n \to \zo$ that have $|g^{-1}(1)| \geq 2^{2n/3}.$
For the lower bound, an easy birthday paradox argument given in Appendix~A of \cite{CDHLNSY25} shows that for any $\eps < 1$, any $\eps$-relative-error testing algorithm for ${\cal F}_2$ must make $\Omega(2^{n/3})$ black-box queries or samples from $f^{-1}(1)$.
For the upper bound, fix any $\eps=\omega(2^{-n/3})$ and any distribution ${\cal D}$ over $\zo^n$.  We claim that every $f: \zo^n \to \zo$ is $\eps$-close to an element of ${\cal F}_2$ under ${\cal D}$ (so a testing algorithm can simply accept without making any black-box queries or samples from $f^{-1}(1)$).  

To see why the claim is true, let $S \subset \zo^n$ be the $2^{2n/3}$ points of $\zo^n$ that have the smallest weight under ${\cal D}$ (breaking ties arbitrarily). We have that ${\cal D}(S) \leq 2^{2n/3}/2^n=2^{-n/3} < \eps$.  So letting $g$ be the function obtained from $f$ by fixing the value to 1 on every point in $S$, we have that $\Pr_{\bx \sim {\cal D}}[f(\bx) \neq g(\bx)] < \eps$. Since $g$ has $|g^{-1}(1)| \geq 2^{2n/3}$ we have $g \in {\cal F}_2$, and the claim is proved.

%% file: sections/ConjTest.tex
\section{An overview of $\ConjTest$ from \cite{rel-error-conj-DL}}\hypertarget{ConjTest}{}
\label{ap:overview-conjtest}

In this section, we present the relative-error tester for conjunction of \cite{rel-error-conj-DL} called $\ConjTest$. The authors first note that for any conjunction $f$ and for any $y \in f^{-1}(1)$, the function
\[f_{y}(x) :=f(x \oplus y)\] 
is an anti-monotone conjunction \footnote{A conjunction is an anti-monotone conjunction if it only contains negated literals.}. It's also is easy to see that if $f$  is $\epsilon$-relative-error far~from every conjunction, then $f_{{y}}$ must be $\epsilon$-relative-error far from every anti-monotone conjunction.

As such, the authors do not explicitly give $\ConjTest$, but instead give the following tester for anti-monotone conjunctions. To obtain $\ConjTest$, it suffices to simply draw a single $y \sim f^{-1}(1)$ and then run the tester for anti-monotone conjunctions (\Cref{algo: mono conjunction tester}) on $f_{y}$, {where $\MQ(f_y)$ and $\SAMP(f_y)$ can easily be simulated  using 
 $\MQ(f)$ and $\SAMP(f)$, respectively}.
\setcounter{algorithm}{18}
\begin{algorithm}[t!] 
\caption{Relative-error Anti-monotone Conjunction Tester. }
\label{algo: mono conjunction tester}
\vspace{0.15cm}\textbf{Input: } $\MQ(f)$ and $\SAMP(f)$ of some function $f:\zo^n \to \zo$ and $\epsilon$. \\
\textbf{Output: }``Reject'' or ``Accept.''

\begin{tikzpicture}
\draw [thick,dash dot] (0,1) -- (15.9,1);
\end{tikzpicture}
\begin{algorithmic}[1] \vspace{-0.15cm}
\Algphase{Phase 1: \vspace{-0.1cm}}
    \State Repeat the following $O(1/\epsilon)$ times :
    \State \hskip4em Draw $\bx,\by \sim \SAMP(f)$.
    \State \hskip4em If $\MQ(f)(\bx \oplus \by)=0$, halt and reject $f$.\vspace{-0.15cm}

\Algphase{Phase 2:\vspace{-0.15cm}}
\State Repeat the following $O(1)$ times :
    \State \hskip4em Draw $\bx \sim \SAMP(f)$.
    \State \hskip4em Draw a uniform random $\by \preceq \bx$.
    \State \hskip4em Draw $\bu \sim \SAMP(f)$.
    \State \hskip4em If $\MQ(f)(\by \oplus \bu)=0$, halt and reject $f$.

\State Accept $f$.\vspace{0.15cm}
\end{algorithmic} 
\end{algorithm}

From the above, it's easy verify that $\ConjTest$ indeed accepts if and only if every $\MQ$ queries it makes is answered with $1$.